\documentclass[11pt,a4paper,twoside]{article}

\usepackage[utf8]{inputenc}
\usepackage[margin=2cm]{geometry}
\usepackage[shortlabels]{enumitem}
\usepackage[page,toc]{appendix}
\usepackage{amsmath,amssymb,graphicx,xcolor,subcaption,float,amsthm,mathrsfs,mathtools,bbold,framed,url,amscd,soul,mathabx,tikz,tikz-cd}
\usepackage[color=yellow]{todonotes}
\usepackage{comment}
\setstcolor{red}

\usepackage{hyperref}
\hypersetup{
colorlinks=true,
}
\numberwithin{equation}{section}

\newcommand{\nocontentsline}[3]{}
\newcommand{\tocless}[2]{\bgroup\let\addcontentsline=\nocontentsline#1{#2}\egroup}

\newtheorem{theorem}{Theorem}
\newtheorem{corollary}[theorem]{Corollary}

\newtheorem{proposition}{Proposition}

\theoremstyle{definition}

\theoremstyle{remark}
\newtheorem{remark}{Remark}[section]

\usepackage{fancyhdr}
\emergencystretch=\maxdimen
\newcommand{\bs}[1]{\boldsymbol{#1}}
\newcommand{\wh}[1]{\widehat{#1}}

\newcommand{\negphantom}[1]{\settowidth{\dimen0}{#1}\hspace*{-\dimen0}}
\def\p{{\partial}}
\def\zh{{\bs{\widehat{z}}}}
\def\rmd{{{\rm d}}}

\begin{document}
	\title{Wave-current interaction on a free surface}
	\author{Dan Crisan\thanks{Department of Mathematics, Imperial College London.} \and Darryl D. Holm\footnotemark[1] \and Oliver D. Street\footnotemark[1]}
	\date{}
	\maketitle
	\begin{center}
	\emph{The rest \dots is a series of speculations, which we hope to verify eventually.}-- H. Segur et al. \cite{HMSS1995}
	\end{center}
	\begin{abstract}
	    The classic evolution equations for potential flow on the free surface of a fluid flow are not closed because the pressure and the vertical velocity dynamics are not specified on the free surface. Moreover, their wave dynamics does not cause circulation of the fluid velocity on the free surface. The equations for free-surface motion we derive here are closed and they are not restricted to potential flow. Hence, true wave-current interaction dynamics can occur. In particular, the Kelvin-Noether theorem demonstrates that wave activity can induce fluid circulation and vorticity dynamics on the free surface. The wave-current interaction equations introduced here open new vistas for both the deterministic and stochastic analysis of nonlinear waves on free surfaces.
	\end{abstract}

\tableofcontents

%%%%%%%%%%%%%%%%%%%%%%%%%%%%%%%%%%%%%%%%%%%%%%%%%%%%%%%%%%%%%%%%%%%%%%%%%%%%%%
\section{Introduction}

\paragraph{Background.}Waves are disturbances in a medium which propagate due to a restoring force, such as gravity. Currents are flows which transport physical properties, such as mass and heat. When waves propagate in a moving medium, the motion of the medium can affect the waves, and vice versa, the waves can affect the motion of the medium, as they both respond to the same force. The primary example is wave-current interaction on the free surface of a fluid flow under the influence of gravity. This mutual wave-current interaction is the province of nonlinear water-wave theory. { In nonlinear water-wave dynamics on a free surface, the distinction between waves and currents is clear: the vertical velocity and surface elevation are wave variables; while the horizontal fluid velocity components and areal mass density are current variables. This is particularly clear in the Hamiltonian formulation of wave-current interaction dynamics, in which the symplectic Poisson operator for the two independent degrees of freedom separates into block diagonal form.}

Water waves -- waves on the surface of a body of water -- have fascinated observers over the ages, not only because water waves are so easily observed, and not only because they move; but primarily because they form coherent moving deformations of the water surface which can interact with each other in a multitude of ways. Any disturbance -- even scooping your hand in a narrow channel of shallow water, for example -- will resolve itself into a train of coherent solitary waves with a few extra ripples which are left behind as the coherent solitary waves propagate away from the disturbance. The fascination in observing the creation of coherent water waves from arbitrary disturbances was captured in the famous report by the Victorian engineer John Scott Russell in August 1834, when he saw a solitary wave create itself from an impulse of current and then start propagating along a Scottish canal. As he wrote \cite{JSRussell-1834},  
\begin{quote}
    I followed it on a horseback, and \dots after a chase of one or two miles I lost it in the windings of the channel. Such, \dots was my first chance interview with that singular and beautiful phenomenon.
\end{quote}
Although the classic water-wave theory introduced in 1847 by Stokes \cite{Stokes1847} now has a long history, see, e.g., \cite{Craik2004,Craik2005}, the excitement in the chase for mathematical understanding of water waves still continues. In particular, John Scott Russell's ``singular and beautiful phenomenon'' is now called a \emph{soliton}. The  word `soliton' was coined in a 1965 paper by Zabusky and Kruskal \cite{ZK1965} and this word has more than 6 million Google hits, as of this writing. The sequence of approximate shallow water equations exhibiting soliton behaviour includes the Korteweg-de Vries (KdV) equation in 1D \cite{KdV1895}, as well as the Kadomtsev-Petviashvili (KP) \cite{KP1970} equation, which extends the KdV equation to allow weak transverse spatial dependence. Indeed, the solution behaviour of soliton water wave equations still inspires mathematical progress in the theory of integrable Hamiltonian systems of nonlinear partial differential equations and their discretizations in space and time. For a good summary of the early developments of soliton theory, see Ablowitz and Segur \cite{Ablowitz-Segur1981}. {For historical discussions of water wave theory see \cite{Craik2004, Craik2005, Darrigol2003}.} For modern mathematical discussions of the classic water-wave theory introduced by Stokes \cite{Stokes1847}, see, e.g., \cite{Ablowitz-etal2006, CastroLannes2015, Lannes2005, LannesBook2013, Lannes2020}. {A few references among modern treatments of water wave theory which are similar in spirit to the present work are \cite{Benney1973,CastroLannes2015,Choi1995,Wu1999,Wu2001}.
A classic review of the various historical formulations of the wave-current interaction problem is \cite{Peregrine1976}.
}

\paragraph{Objectives and methodology.} { 

%Most studies of wave-current interaction distinguish between waves and currents using a fluid flow decomposition defined on a case-by-case basis to suit particular research questions. Waves can be considered as surface disturbances, periodic solutions, or fluctuations around a mean bulk flow. Currents are often defined within the bulk of the fluid, and their definition can involve averaging \cite{Peregrine1976}. 

%However, when the flow is restricted to the free surface, the distinction between waves and currents is quite clear: the `currents' are the horizontal components of the flow velocity and the `waves' are represented by the vertical component of velocity and the free surface elevation above the equilibrium position when the fluid motion is still. These wave and current variables will be shown to be represented as two independent canonically conjugate pairs of dynamical variables in a Hamiltonian formulation of their free-surface wave-current dynamics.\footnote{The remainder of the paper will use the following terminology. The `current', as described here, will be denoted by $\bs{\wh v}$. By `wave variables', we mean the vertical component of velocity, $\wh w$, and the free surface elevation $\zeta(x,y,t)$.}

Within the framework of wave-current interaction, one notes that `waves' can propagate along the free surface either with the flow, or as a travelling shift in the phase of the elevation which does not carry mass as it propagates on the surface of the flow. For example, if one were to place dye within a wave elevation, the wave need not carry that dye along with it (although waves with this property certainly can exist). Indeed, the wave reported by Russell in 1834 was propagating along the canal at a speed which required a horse to keep up with it, although it is a safe assumption that the `current' flow velocity in the canal was much slower. In fact, modern water-wave experiments such as those of T.Y. Wu \cite{Wu1987} report observations of periodic emission of waves which propagate in the opposite direction of the current in shallow-water flow over a submerged obstacle. In a situation where the free surface elevation follows the currents, the waves would be associated with mass transport and their rate of propagation would equal the fluid transport velocity. The present work will develop models in which the free surface waves can either propagate with the fluid transport velocity as in the classic water wave (CWW) theory, or propagate in either direction on the background flow of the fluid transport velocity.
}

The classic water wave equations (CWWE) for potential flow on a free surface comprise the kinematic constraint at the free boundary and the horizontal gradient of Bernoulli's Law for the case of potential flow restricted to the surface. This paper has two primary objectives based on the CWWE. 

The first objective is to augment the CWWE to include fundamental physical aspects of wave-current interaction on a free surface (WCIFS). These physical aspects include vorticity, wave-current coupling in which the wave activity creates fluid circulation, non-hydrostatic pressure, incompressibility, and horizontal gradients of buoyancy. 

The multi-scale, fast-slow aspects of the wave-current interaction comprise a grand challenge for modern computational simulation. This challenge is particularly important in computational simulations of global ocean circulation.  In ocean physics, the fast-slow aspects of wave-current interaction tend to introduce irreducible imprecision even beyond computational uncertainty  because of unresolvable, or even unobservable, processes \cite{McW-irred2007}. This situation leads to the paper's second objective.

The paper's second objective aims to introduce stochastic transport of wave activity by fluid circulation which is intended to be used in combination with data assimilation to model uncertainty due to the effects of fast, computationally unresolvable, or unknown effects of WCI on its slower, computationally resolvable aspects. 

To pursue these two objectives, we will begin by using the Dirichlet-Neumann operator (DNO) for 3D potential flow of a homogeneous Euler fluid to impose the kinematic and dynamic boundary conditions of CWWE as \emph{constraints} on the motion of the free surface in the Euler-Poincar\'e (EP) variational principle for ideal fluids, \cite{HMR1998}. The EP formulation is an extension of earlier variational principles for the CWWE, \cite{Bateman1929,Luke1967}. In using the CWWE as constraints, the EP variational principle introduces additional dynamical equations for the Lagrange multipliers. The Lagrange multipliers are interpreted as the vertical velocity $\wh{w}$ and the areal mass density $D$, arising  as Hamiltonian variables canonically conjugate to the elevation $\zeta$ and the surface velocity potential $\wh{\phi}$, respectively. The resulting Hamiltonian equations are referred to as extended CWWE, abbreviated ECWWE. 

After a discussion of alternative formulations which elicit a variety of properties of the ECWWE solutions, we use the EP approach to add further aspects of wave-current interaction, which include vorticity, as well as non-hydrostatic pressure and buoyancy gradients in the free surface flow. 

We also introduce a wave-current minimal coupling (WCMC) term into the action integral for the EP variational principle which generates fluid circulation from wave activity and vice-versa. The EP variational equations are referred to here as wave-current interaction on a free surface (WCIFS) equations. Finally, to model the uncertainties associated with the computations of these multi-scale fast-slow WCIFS equations, we introduce stochastic advection by Lie transport (SALT) of the wave activity by the current flow, again following the EP variational approach, as in \cite{Holm2015}. In the analytical sections of the paper, both the deterministic and SALT versions of our WCIFS equations are shown to be locally well-posed, in the sense of existence, uniqueness, and continuous dependence on initial conditions.

\paragraph{Plan of the paper.} 
\begin{itemize}

\item 
Section \ref{sec: Problem statement} reviews the problem statement, boundary conditions and key relations in the classical framework of three-dimensional fluid flows under gravity with a free surface. The free surface elevation is measured from its rest position, which defines the origin of the vertical coordinate $z = 0$. The elevation, $z = \zeta(\bs{r}, t)$, is a function of the horizontal position vector $\bs{r}= (x, y, 0)$ and time $t$. A key relation is stated in equation \eqref{hat-advectionX}. Namely, when evaluated on the free surface, the material time derivative of a function $f(\bs{r}, z, t)$ and its projection onto the free surface $f(\bs{r}, \zeta(\bs{r}, t), t)$ are equal. The projection relation in equation \eqref{hat-advectionX} then leads to \emph{Choi's relation} \eqref{ChoiEqnX} for the dynamics of the free surface, \cite{Choi1995}.

\item 
Section \ref{subsec:CWWE} reviews the formulation of the classic water-wave equations (CWWE) for free surface dynamics  in terms of the Dirichlet-Neumann operator (DNO). Section \ref{subsec:derivingCWWE}  then derives the \emph{extended} CWWE (ECWWE) which include equations for the vertical velocity $\wh{w}$ and preserved area measure $D$ on the free surface. The derivation of ECWWE proceeds by regarding the CWWE as constraints imposed by Lagrange multipliers $\wh{w}$ and $D$ in a new variational principle in equation \eqref{ActionIntegral-FS}, defined in terms of functions on the horizontal mean level of the free surface. The Lie-Poisson Hamiltonian form of the ECWWE is derived in section \ref{sec: HamForm-CWWE}. In section \ref{subsec:CWWEPressure}, non-hydrostatic pressure is incorporated into the ECWWE and the comparison to the key relation in \eqref{ChoiEqnX} is shown in Theorem \ref{Thm: pressure projection}. The ECWW theory satisfies a sort of time-dependent non-acceleration theorem, by which the the fluid and wave circulations are preserved separately. Hence, the ECWWE does not really qualify as genuine wave-current interaction, since the time-dependent flows of real fluids allow exchange of circulation between waves and currents. 

{The non-acceleration property of the ECWWE may be rectified by inserting a term into the Lagrangian which represents the dependence of the kinetic energy on the wave slope. The new kinetic energy term results in a system where the transport velocity is unaffected and the wave dynamics creates circulation in the current flow when the gradients of vertical velocity and surface elevation are not aligned.}

\item 
Section \ref{sec: ACWWE} introduces the \emph{augmented} CWW system (ACWW) which includes a different wave-current interaction which rectifies the non-acceleration property of the ECWWE via a minimal coupling (WCMC) construction. {The additional term involved in this construction introduces a shift in the transport velocity for the wave elevation which depends on the wave slope. The result is that the waves can move relative to the fluid parcels on the free surface, and hence the system can support waves which do not transport mass.} The Kelvin-Noether circulation theorem for the ACWW model in Theorem \ref{KNthm-CFS} shows that the wave dynamics of the ACWW model can create circulation, as shown locally in Theorem \ref{KelThmACWWE} and Corollary \ref{cor: wave gen circ}. In addition, as shown in Theorem \ref{Erg-conserv}, the ACWW model with genuine wave-current interaction preserves the same physical energy as the ECWW model does, for which the fluid and wave circulations are preserved separately. That is, the ACWW model with its wave-current minimal coupling (WCMC) term preserves the same energy as the ECWW model.

Section \ref{sec: FECWWE} includes additional physical properties into the ACWWE to produce our final model of wave-current interaction on a free surface (WCIFS).  
The WCIFS model is derived by modifying Hamilton's principle for ACWW with its wave-fluid coupling term, to add an advected scalar buoyancy variable $\rho$ with nonzero horizontal gradients, as well as non-hydrostatic pressure in \eqref{eq: WCIFS}. The Kelvin-Noether circulation theorem for the WCIFS model is given in Theorem \ref{KNthm-WCIFS}.

\item 
Section \ref{sec: Stoch WCIFS eqns} introduces a method of incorporating stochastic noise into this theory which preserves its variational structure. This is achieved by applying the method of Stochastic Advection by Lie Transport (SALT) \cite{Holm2015}. Within this section, we derive a stochastically perturbed version of the classical water-wave equations, as well as a stochastic variational models of the ECWW, ACWW, and WCIFS models of wave-current interaction of free surfaces. 

\item 
Section \ref{sec: Stoch analysis} lays out the analytical framework for dealing with the compressible and incompressible ECWWE equations discussed in sections \ref{subsec:CWWE} and \ref{subsec:CWWEPressure}, respectively.

\item 
Section \ref{sec: Conclude} discusses potential future research directions and identifies new problems opened up by this work.

\item
Appendix \ref{app: FluidTransTheory}  discusses transformation theory for ideal fluid dynamics and derives the Kelvin circulation theorem by using the Lie chain rule. 

\item 
Appendix \ref{app:ActionIntBdy-1} reviews the 3D inhomogeneous Euler equations for incompressible fluid flow under gravity with a free upper surface and a fixed bottom topography. This is done by deriving these equations using a constrained variational principle. 

\item
Appendix \ref{app:LegXform} treats the reduced Legendre transformation from Hamilton's principle to the Hamiltonian formulation for ECWWE in the Eulerian fluid representation. 

\end{itemize}

\begin{center}
\begin{figure}[H]
\begin{center}
\textbf{Roadmap of the paper}
\end{center}
\bigskip

\begin{tikzcd}
[row sep = 3em, 
column sep=2em, 
cells = {nodes={top color=green!20, bottom color=blue!10,draw=blue!90}},
arrows = {draw = black, rightarrow, line width = .03cm}]
& & \hyperref[sec: Problem statement]{\begin{matrix}
\text{3D free-surface Euler}\\ \text{equations and} \\ \text{Choi's relation}
\end{matrix}} \arrow[d,"{\begin{matrix} \text{2D variational principle using a}\\ \text{Dirichlet-Neumann operator}\end{matrix}}", shorten <= 1mm, shorten >= 1mm] & &\\
& & \hyperref[subsec:derivingCWWE]{\begin{matrix} \text{A variational 2D field}\\\text{theory containing the `Classical }\\\text{water-wave equations' (CWWE)}\end{matrix}} \arrow[dl, "{\begin{matrix} \text{incompressibility of flow of}\\ \text{currents on the free surface}\end{matrix}}" swap, shorten <= 4mm, shorten >= 4mm] \arrow[dr, "{\begin{matrix} \text{wave-current minimal}\\ \text{coupling (WCMC) } \end{matrix}}", shorten <= 4mm, shorten >= 4mm] \arrow[dd,"{\begin{matrix} \text{including pressure,}\\ \text{WCMC, and} \\ \text{buoyancy} \end{matrix}}", shorten <= 2mm, shorten >= 2mm] \arrow[bend right = 62,ddd,crossing over,shorten <= 4mm, shorten >= 4mm]& &\\
& \hyperref[subsec:CWWEPressure]{\begin{matrix} \text{Variational CWWE}\\ \text{with pressure} \end{matrix}} \arrow[dr, swap,crossing over, shorten <= 4mm, shorten >= 4mm]\arrow[bend right = 35,ddr,dashed,shorten <= 4mm, shorten >= 4mm] & & \hyperref[sec: ACWWE]{\begin{matrix} \text{Augmented CWWE with}\\ \text{wave-current coupling} \end{matrix}} \arrow[dl, shorten <= 4mm, shorten >= 4mm]\arrow[bend left = 35,ddl,dashed,shorten <= 4mm, shorten >= 4mm] & \\
& & \hyperref[sec: FECWWE]{\begin{matrix} \text{A variational wave-current}\\ \text{interaction model with} \\ \text{buoyancy (WCIFS)}\end{matrix}}\arrow[d,"{\begin{matrix}\text{Stochastic advection}\\ \text{by Lie transport} \end{matrix}}", shorten <= 2mm, shorten >= 2mm] & & \\
& & \hyperref[sec: Stoch WCIFS eqns]{\begin{matrix} \text{Stochastic variational} \\ \text{wave models.}  \end{matrix}}
\end{tikzcd}
\caption{The figure sketches the relationships among the different models which are derived in the remainder of this article. Dashed arrows represent connections of these exact models to their stochastic versions which are not derived explicitly. However, the missing derivations follow the same patterns as that described in full for the stochastic wave-current interaction model in section \ref{subsec:Stochastic_WCIFS}.}
\label{Roadmap}
\end{figure}
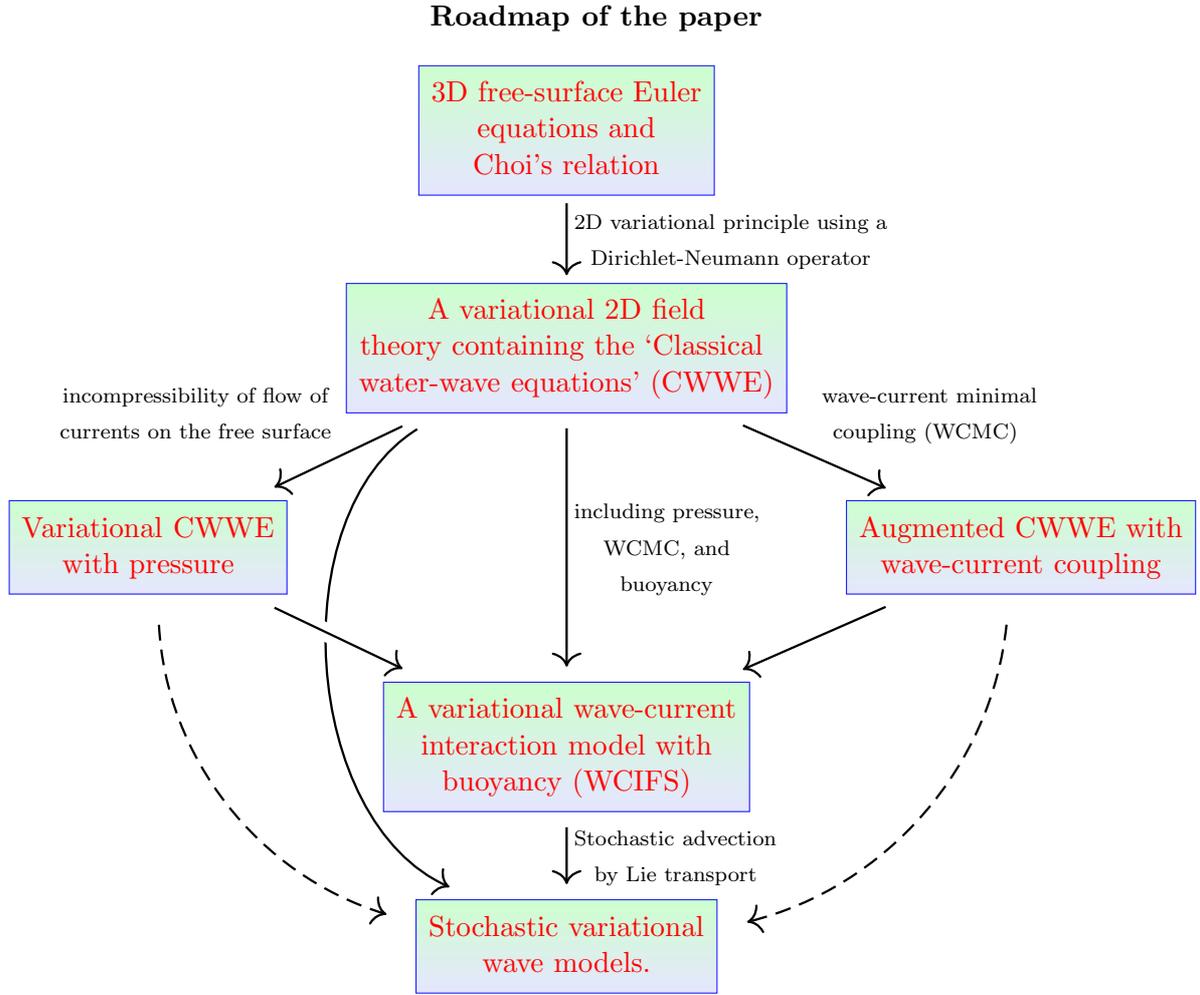
\end{center}

\paragraph{List of abbreviations.}
\begin{itemize}
\item Classical water-wave equations (CWWE)
\item Wave-current minimal coupling (WCMC)
\item Extended CWWE (ECWWE)
\item Augmented CWWE (ACWWE)
\item Wave-current interaction (WCI) 
\item Wave-current interaction on a free surface (WCIFS)
\item Dirichlet-Neumann operator (DNO)
\item Stochastic Advection by Lie Transport (SALT)
\item Euler-Poincar\'e (EP)
\end{itemize}

\section{Problem statement, Boundary Conditions and Key Relations}\label{sec: Problem statement}

\paragraph{Problem statement.}
We study the dynamics of fluid parcels which are constrained to remain on the free surface of a three dimensional fluid with coordinates $\bs{x}=(\bs{r},z)$. Here $\bs{r}=(x,y)$ (resp. $z$) denotes horizontal (resp. vertical) Eulerian spatial coordinates in an inertial (fixed) domain. The fluid domain is bounded below by a rigid bottom at $z = -B(\bs{r})$ and is bounded above by the free surface of the fluid at $z = \zeta(\bs{r}, t)$, which is measured from its rest position at $z = 0$ as a function of the horizontal position vector $\bs{r}= (x, y, 0)$ and time $t$.  

\paragraph{Three-dimensional fluid equations.}
The fluid moves in three dimensions with velocity $\bs{u}(\bs{x},t) = (\bs{v}(\bs{x},t),w(\bs{x},t))$ in which $\bs{v}(\bs{x},t)$ and $w(\bs{x},t)$ denote, respectively, the horizontal and vertical velocity fields.
Incompressible and inviscid fluid motion is governed by the Euler equations of horizontal and vertical momentum dynamics under the constant acceleration of gravity, $g$. The equations are given by
\begin{align}
\begin{split}
\mathcal{D}\bs{v}
&:= \bs{v}_t + \bs{v}\cdot\nabla_{\bs{r}}\bs{v} + w{\bs{v}}_z 
= -\,\frac{1}{\rho}\nabla_{\bs{r}}\pi
\,,\\
\mathcal{D}w&:= w_t + \bs{v}\cdot\nabla_{\bs{r}}w + ww_z 
= -\,\frac{1}{\rho}\pi_z - g 
\,,\\
\hbox{with}\quad
\mathcal{D} &:= \p_t + \bs{v}\cdot\nabla_{\bs{r}} + w\p_z
\quad\hbox{and}\quad 
\nabla_{\bs{r}}\cdot \bs{v} + w_z = 0
\,.
\end{split}
	\label{motioneqn:componentsX}
\end{align}
%ecww
We denote by $\pi$ the pressure with three dimensional spatial dependence. The volume element is $d^3x=d^2r\wedge dz$, and its measure $Dd^3x$ is preserved under the incompressible fluid flow. The mass density is given by $\rho=\rho_0(1+b(\bs{x},t))>0$, in which $b(\bs{x},t) \coloneqq (\rho-\rho_0)/\rho_0$ is the fluid buoyancy and $\rho_0>0$ is the (constant, positive) reference value of mass density.  The mass in each fluid volume element is given by $\rho Dd^3x$. The condition $\zeta(\bs{r},t)-z = 0$, which defines the free surface, is assumed to be preserved under the flow. This condition ensures that a particle initially on the free surface will remain on it. 

These three preservation relationships may be expressed as the following three \emph{advection relations}, 
\begin{align}
\begin{split}
(\p_t + \mathcal{L}_{\bs{u}})(Dd^3x) = \big(\partial_tD+\nabla\cdot (D{\bs{u}})\big)d^3x&=0
\,,\\
(\p_t + \mathcal{L}_{\bs{u}})\rho = \partial_t\rho+{\bs{u}}\cdot\nabla\rho&=0
\,,\\
(\p_t + \mathcal{L}_{\bs{u}})(\zeta(\bs{r},t) - z) = (\partial_t+{\bs{u}}\cdot\nabla)(\zeta(\bs{r},t)-z) &= 0
\,,
\end{split}
	\label{advectioneqns}
\end{align}
where the operator $(\p_t + \mathcal{L}_{\bs{u}})$ is the advection operator (see Appendix \ref{app: FluidTransTheory}).
{ Requiring the volume measure $Dd^3x$ to remain constant in the first advection relation in \eqref{advectioneqns} implies that the flow velocity remains divergence-free, $\nabla\cdot {\bs{u}} =0$. The preservation of the divergence-free condition under the fluid flow then implies a Poisson equation for the fluid pressure $\pi$ in the motion equation \eqref{motioneqn:componentsX}. }

The motion equations in \eqref{motioneqn:componentsX} and the initial values for the advected quantities $\rho(\bs{r},z,t)$ and $(\zeta(\bs{r},t)-z)$ in the advection relations in \eqref{advectioneqns} provide a complete specification of the initial value problem for the fluid motion with appropriate boundary conditions in three dimensions. 

\paragraph{Boundary conditions.}
Whilst the horizontal boundary conditions are yet to be specified and can be chosen to suit specific problems, the vertical boundary conditions must be carefully defined. 

%The boundary conditions on the fixed parts of the boundary are standard. Namely, we take the velocity to be tangential and the buoyancy to be constant on the fixed components of the boundary. 

The kinematic boundary condition on the free surface is given by 

\begin{equation}\label{bdyconditions:componentsX}
    \wh{w} = \wh{\mathcal{D}}\zeta \qquad \big(\wh{\mathcal{D}} = \p_t + \bs{\wh{v}}\cdot\nabla_{\bs{r}}\big)\,,
\end{equation}
where the $\wh{f}$ notation in $\wh{\mathcal{D}}$, $\wh{\bs{v}}$, and $\wh{w}$ is defined for an arbitrary flow variable $f$ to represent evaluation on the free surface, namely,
\begin{align}
\wh{f}(\bs{r},t) = f(\bs{r},z,t)\quad \hbox{on}\quad z = \zeta(\bs{r},t)\,.
	\label{hat-notationX}
\end{align}
Notice that evaluating on the free surface before taking derivatives is not equivalent to taking the derivative before evaluating. In particular,  $\p_t\wh{f}(\bs{r},t)\ne \widehat{\p_t f}(\bs{r},t)$ and 
$\nabla_{\bs{r}}\wh{f}(\bs{r},t) \ne \widehat{\nabla_{\bs{r}}f}(\bs{r},t)$, 
where 
\begin{equation}
\widehat{\p_t f}(\bs{r},t) = [\p_t f(\bs{r},z,t)]_{z = \zeta(\bs{r},t)}
\quad\hbox{and}\quad 
\widehat{\nabla_{\bs{r}}f}(\bs{r},t) = [\nabla_{\bs{r}} f(\bs{r},z,t)]_{z = \zeta(\bs{r},t)}
\,.
\end{equation}
Instead, from the chain rule we have 
\begin{equation}
\begin{aligned}
\p_t\wh{f}(\bs{r},t) 
&= \Big[\p_t f + f_z \p_t\zeta \Big]\Big|_{z = \zeta(\bs{r},t)} = \widehat{\p_t f} + \widehat{\p_zf}\p_t\zeta 
\,,\\
\nabla_{\bs{r}}\wh{f}(\bs{r},t) 
&= \Big[\nabla_{\bs{r}} f + f_z  \nabla_{\bs{r}}\zeta \Big]\Big|_{z = \zeta(\bs{r},t)} = \widehat{\nabla_{\bs{r}} f} + \widehat{\p_zf}\nabla_{\bs{r}}\zeta
\,.
	\label{hat-chainruleX}
\end{aligned}
\end{equation}
Consequently, we have the following remarkable proposition.
\begin{proposition}[T.Y. Wu \cite{Wu2001}]
    The advection operator on the free surface satisfies the identity
    \begin{align}
        \wh{\,\mathcal{D}f\,} = \wh{\mathcal{D}}\wh{f}
        \,,
    	\label{hat-advectionX}
    \end{align}
    where $\mathcal{D} \coloneqq (\p_t + \bs{u}\cdot\nabla)$ and $\wh{\mathcal{D}}\coloneqq (\p_t + \bs{\wh v}\cdot\nabla_{\bs{r}})$.
\end{proposition}
That is, on the free surface, the 3D material time derivative of $f(\bs{r}, z, t)$ equals the 2D surface material derivative of the surface evaluation of $f$.
\begin{proof}
    Applying the chain rules in \eqref{hat-chainruleX} leads to
    \begin{align*}
        \wh{\,\mathcal{D}f\,} &= \wh{\p_tf} + \bs{\wh v}\cdot\wh{\nabla_{\bs{r}}f} + \wh{w}\wh{\p_zf} \\
       \hbox{By \eqref{hat-chainruleX} } &= \p_t\wh f - \wh{\p_zf}\p_t\zeta + \bs{\wh v}\cdot(\nabla_{\bs{r}}\wh f - \wh{\p_zf}\nabla_{\bs{r}}\zeta) + \wh w \wh{\p_zf} \\
        &= \p_t\wh f + \bs{\wh v}\cdot\nabla_{\bs{r}}\wh f + \wh{\p_zf}\big(\wh w - \p_t\zeta - \bs{\wh v}\cdot\nabla_{\bs{r}}\zeta \big) \\
        &= \p_t\zeta + \bs{\wh v}\cdot\nabla_{\bs{r}}\wh f =: \wh{\mathcal{D}}\wh f\,,
    \end{align*}
    in which the final line is implied by the kinematic condition \eqref{bdyconditions:componentsX}.
\end{proof}

\paragraph{Choi's relation at the free surface.}\label{Remark-Open}
One may use the relation \eqref{hat-advectionX} to evaluate the horizontal and vertical coordinates of the motion equation \eqref{motioneqn:componentsX} onto the free surface. Hence, one finds for constant buoyancy $\rho=\rho_0$ on the free surface $\zeta(x,y,t)-z=0$ that
\begin{align}
\begin{split}
\wh{\mathcal{D}}\wh{\bs{v}} = -\left[\frac{1}{\rho_0}\nabla_{\bs{r}} \pi\right]\bigg|_{z=\zeta}
&= -\,
\Big[\frac{1}{\rho_0} \nabla_{\bs{r}} \wh{\pi} 
- \frac{1}{\rho_0} \pi_z\nabla_{\bs{r}}\zeta \Big]\Big|_{z=\zeta}
\\&= -\,\frac{1}{\rho_0}\nabla_{\bs{r}} \wh{\pi} 
- (\wh{\mathcal{D}}\wh{w} +g)\nabla_{\bs{r}}\zeta
\,,
\end{split}
	\label{projectedmotioneqnX}
\end{align}
where, in the last step, we have used the vertical motion equation in \eqref{motioneqn:componentsX} to evaluate $\pi_z$ for $\rho=\rho_0$ and the relation \eqref{hat-advectionX} for the vertical acceleration of the free surface, $dw/dt_{z=\zeta(\bs{r},t)} = \wh{\,\mathcal{D}w\,} = \wh{\mathcal{D}}\wh{w} $. In conclusion, upon using the boundary condition $\wh{w}=\wh{\mathcal{D}}\zeta$ in \eqref{bdyconditions:componentsX} we find \emph{Choi's relation} at the free surface \cite{Choi1995},
\begin{align}
\wh{\mathcal{D}}\wh{\bs{v}} + (\wh{\mathcal{D}}^2\zeta +g) \nabla_{\bs{r}}\zeta
= -\,\frac{1}{\rho_0}\nabla_{\bs{r}} \wh{\pi}
\,.\label{ChoiEqnX}
\end{align}
{
\begin{remark}[Closing Choi's relation \eqref{ChoiEqnX}]
The fundamental relation in \eqref{ChoiEqnX} is not restricted to irrotational flows. 
However, at this stage, the dynamical system comprising equations \eqref{bdyconditions:componentsX} and \eqref{ChoiEqnX} for the motion of the free surface \emph{is not yet closed} since, (i) the pressure gradient $\nabla_{\bs{r}}\wh\pi$ is still unknown and, (ii) an evolutionary equation for $\wh{\mathcal{D}}\wh{w}$ is missing.
\end{remark}
}

\paragraph{Choi's relation at the bottom boundary.}
One may also consider boundary conditions at either a lower free surface,  $z=-B(\bs{r},t)$, or at fixed  bathymetry, $z=-B(\bs{r})$. Denote by $\widecheck f$ the evaluation on the bathymetry, i.e.
\begin{equation}
    \widecheck f(\bs{r},t) = f(\bs{r},z,t)\quad\hbox{on}\quad z = -B(\bs{r},t)\,.
\end{equation}
The bottom boundary condition is
\begin{equation}
    -\widecheck{\mathcal{D}} B(\bs{r},t) = -(\p_t+\widecheck{\bs{v}}\cdot\nabla_{\bs{r}})B(\bs{r},t) = \widecheck{w}\,, \quad\hbox{on}\quad z = -B(\bs{r},t)\,.
\label{BottomBCX}
\end{equation}
and we have, by the same chain-rule calculations as on the upper free surface,
\begin{equation}
    \widecheck{\,\mathcal{D}f\,} = \widecheck{\mathcal{D}} \widecheck{f}\,.
\end{equation}
We may now evaluate equations \eqref{motioneqn:componentsX} onto the lower surface $z = -B(\bs{r},t)$ in the same manner as we have evaluated onto the upper surface $z = \zeta(\bs{r},t)$ to give
\begin{equation}
    \widecheck{\mathcal{D}}\widecheck{\bs{v}} 
    + (\widecheck{\mathcal{D}}\widecheck w + g)\nabla_{\bs{r}}B(\bs{r},t)
    = -\, \frac{1}{\rho_0}\nabla_{\bs{r}}\widecheck \pi \,,
    \quad\hbox{on}\quad z = -B(\bs{r},t)\,.
\end{equation}
By the bottom boundary condition \eqref{BottomBCX} we then find 
\begin{equation}
    \widecheck{\mathcal{D}}\widecheck{\bs{v}} 
    + (- \widecheck{\mathcal{D}}^2\widecheck B(\bs{r},t) + g)\nabla_{\bs{r}}B(\bs{r},t)
    = -\, \frac{1}{\rho_0}\nabla_{\bs{r}}\widecheck \pi \,,
    \quad\hbox{on}\quad z = -B(\bs{r},t)\,.
\label{BottomBDynX}
\end{equation}
When $B(\bs{r},t)$ a time-dependent variable, then equation \eqref{BottomBDynX} is not closed. However, when the bottom boundary is taken to be time-independent, so that $z = -B(\bs{r})$ and $\widecheck{\mathcal{D}}  = \widecheck{\bs{v}}\cdot\nabla_{\bs{r}}$ there, then equation \eqref{BottomBDynX} would be closed, provide either the bottom pressure $\widecheck \pi $ were prescribed, or the bottom velocity $\widecheck{\bs{v}}$ were taken to be divergence-free.

\paragraph{Mean continuity relation.}
Another exact result about the dynamics of the free surface elevation $\zeta (x,y,t)$ should be mentioned. This result is the following mean continuity relation for the elevation in terms of the vertically averaged horizontal velocity components, \cite{Wu1999}.
\begin{proposition}[Mean continuity relation]\label{prop: mean FS cont eqn}
For uniform mass density, $\rho_0$, the boundary conditions \eqref{bdyconditions:componentsX}, as well as incompressibility and equation \eqref{hat-advectionX} for advection of the free surface  together imply  the following \emph{vertically integrated} continuity relation for the wave elevation on the free surface, $\zeta$, 
\begin{align}
\partial_t \zeta (\bs{r},t)
	+ \partial_x\int_{-B}^{\zeta(\bs{r},t)} u(\bs{r},z,t)dz + \partial_y\int_{-B}^{\zeta(\bs{r},t)} v(\bs{r},z,t)dz
= 0
\,.
\label{mean-continuity-relationX}
\end{align}
In terms of vertically averaged quantities, denoted by
\begin{align}
\overline{f}(\bs{r},z,t) := \frac{1}{\zeta+B}\int_{-B}^{\zeta(\bs{r},t)} f(\bs{r},z,t)dz
\,,
\label{mean-vert-avgX}
\end{align}
equation \eqref{mean-continuity-relationX} may be written equivalently as a mean (i.e., vertically-averaged) continuity equation,
\begin{align}
\partial_t \big(\zeta (\bs{r},t) + B \big)
	+ \partial_x\big((\zeta+B)\overline{u}\big) 
	+ \partial_y\big((\zeta+B)\overline{v}\big)
= 0
\,.
\label{mean-continuity--ubarX}
\end{align}

\end{proposition}
\begin{remark}[Physical interpretation]
Essentially, the continuity equation \eqref{mean-continuity--ubarX} arises because the incompressible flow conserves the fluid volume measure, $Dd^3x$. 
In particular, the vertically integrated continuity relation \eqref{mean-continuity--ubarX} in Proposition \ref{prop: mean FS cont eqn} proved below represents volume preservation of the divergence-free 3D Euler fluid equations in \eqref{EulerPoincareX} of appendix \ref{app:ActionIntBdy-1} for the advective boundary relations in \eqref{advectioneqns} and \eqref{ChoiEqnX}, see  e.g., \cite{Benney1973, Choi1995, Wu1999}.

\end{remark}

\begin{proof}
By direct computation, using the advection condition for $\zeta$ and the vertical integral of the divergence free condition ${\rm div}\bs{u}=0$, and upon noticing that no contribution arises from the flat bottom boundary, one finds that
\begin{align}
\begin{split}
	\partial_t \zeta (\bs{r},t)
	&+ \partial_x\int_{-B}^{\zeta(\bs{r},t)} u(\bs{r},z,t)dz + \partial_y\int_{-B}^{\zeta(\bs{r},t)} v(\bs{r},z,t)dz
	\\&= \partial_t\zeta (\bs{r},t) + u(\bs{r},z,t)\partial_x\zeta (\bs{r},t) + v(\bs{r},z,t)\partial_y \zeta(\bs{r},t)
	\\& \qquad  + \int_{-B}^{\zeta(\bs{r},t)} \Big( u_x + v_y + w_z\big)(\bs{r},z,t) dz 
	     - w (\bs{r},\zeta (\bs{r},t),t) 
	\\& = 0\,,
\end{split}	
\label{mean-elevation-proofX}
\end{align}
where we have added and subtracted $w (\bs{r},\zeta (\bs{r},t),t)$ and applied a tangential flow condition at the bottom boundary.
Thus, the boundary conditions and the divergence-free nature of the three dimensional flow combine to produce the mean continuity relation in \eqref{mean-continuity-relationX}.
\end{proof}

\paragraph{The route to a closure scheme for  \eqref{bdyconditions:componentsX} and \eqref{ChoiEqnX} which includes the CWWE.}
The closure problem for equations \eqref{bdyconditions:componentsX} and \eqref{ChoiEqnX} will be resolved in section \ref{subsec:derivingCWWE} in the context of the CWWE, which will imply $\wh{\mathcal{D}}\wh{w} = -\,g$. In section \ref{sec: HamForm-CWWE}, the pair of wave variables $\wh w$ and $\zeta$ will be understood as a canonically conjugate subset of a Hamiltonian system of Eulerian equations for planar fluid motion. This system will also contain the hydrostatic CWWE introduced in section \ref{subsec:CWWE}. In section \ref{subsec:CWWEPressure}, non-hydrostatic pressure $\wh{\pi}$ will be incorporated into the CWW problem to complete the closure of Choi's relation in \eqref{ChoiEqnX}. The rest of the paper will then build additional physics into the resulting system of planar fluid equations, e.g., by including horizontal gradients of buoyancy on the free surface. Refer to Figure \ref{Roadmap} for more perspective.

%%%%%%%%%%%%%%%%%%%%%%%%%%%%%%%%%%%%%%%%%%%%%%%%%%%%%%%%%%%%%%%%%%%%%%%%%%%%%%
%%%%%%%%%%%%%%%%%%%%%%%%%%%%%%%%%%%%%%%%%%%%%%%%%%%%%%%%%%%%%%%%%%%%%%%%%%%%%%
\section{Free surface dynamics}\label{sec:FreeSurfDyn}

%%%%%%%%%%%%%%%%%%%%%%%%%%%%%%%%%%%%%%%%%%%%%%%%%%%%%%%%%%%%%%%%%%%%%%%%%%%%%%
\subsection{The Classic Water-Wave equations (CWWE)}\label{subsec:CWWE}

\paragraph{The Dirichlet-Neumann operator (DNO).}
In this section, we consider the much studied potential flow governed by the CWWE. The qualitative information obtained here from the CWWE will inspire our derivation of a constrained variational principle below. The CWWE are derived from the free surface three dimensional Euler equations via the \emph{Dirichlet-Neumann} operator (DNO). The DNO maps the solution of Laplace's equation in an external domain with a Dirichlet boundary conditions to its solution on the boundary with a Neumann flux condition (see e.g. \cite{Lannes2005}). In particular, the CWWE  assume that the flow is incompressible and irrotational and, thus, there exists some $\phi(\bs{r},z,t)$ such that $\bs{u} = \nabla\phi$, where $\bs{u}$ is the three-dimensional velocity field throughout the domain. 

In the hat-notation of \eqref{hat-notationX}, the variable $\wh{\phi}(\bs{r},t) = \phi(\bs{r},\zeta(\bs{r},t),t)$ evaluates the velocity potential $\phi(\bs{r},z,t)$ on the free surface $z=\zeta(\bs{r},t)$. 
The action of the DNO $G(\zeta)$ on $\wh{\phi}(\bs{r},t)$  is defined as the normal component of the three-dimensional velocity field for the potential flow $\bs{u} = \nabla\phi$ evaluated at the free surface $z=\zeta(\bs{r},t)$. Namely, 
\begin{equation}
    G(\zeta)\widehat\phi  \coloneqq (- \nabla_{\bs{r}} \zeta, 1) \cdot \widehat{\,\nabla\phi\,}
    \coloneqq - \nabla_{\bs{r}}\zeta(\bs{r},t) \cdot \widehat{\nabla_{\bs{r}} \phi} + \widehat w\,,
    \label{def-DNO}
\end{equation}
in which the horizontal gradient of the velocity potential $\phi(\bs{r},z,t)$ is first taken, then evaluated at the surface $z=\zeta (\bs{r},t)$, cf. equation \eqref{projectedmotioneqnX}.

Thus, the DNO in \eqref{def-DNO} takes Dirichlet data for $\widehat\phi$ on $z=\zeta(\bs{r},t)$, solves Laplace's equation $\Delta\phi=0$ for $\phi(\bs{r},z,t)$ together with the condition that the velocity $\bs{u} = \nabla\phi$ have no normal component on the fixed parts of the boundary of the full domain volume, and then returns the corresponding Neumann data, i.e., the three-dimensional fluid normal velocity on the free surface, $z=\zeta(\bs{r},t)$. 

\paragraph{The classical water-wave equations (CWWE).}
The classical water-wave equations (CWWE), as stated in \cite{Lannes2005}, can be written in terms of the DNO as
\begin{align}
    \p_t\zeta - G(\zeta)\widehat\phi &= 0 \,, \label{LannesZetaEquationX} \\
    \p_t\widehat\phi + g\zeta + \frac12 |\nabla_{\bs{r}}\widehat\phi|^2 - \frac{1}{2\big(1+|\nabla_{\bs{r}}\zeta|^2\big)}\big(G(\zeta)\widehat\phi + \nabla_{\bs{r}}\zeta \cdot \nabla_{\bs{r}}\widehat\phi \,\,\big)^2 &= 0 \label{LannesPsiEquationX}\,.
\end{align}
From the chain rules in \eqref{hat-chainruleX}, we have, in the hat notation of equation \eqref{hat-notationX}, that
\begin{align}
    \widehat{\nabla_{\bs{r}}\phi} &= \nabla_{\bs{r}}\widehat\phi (\bs{r},t) - \widehat{\p_z\phi}\nabla_{\bs{r}}\zeta\,, \label{ChainRuleNablaX} \\
    \widehat{\p_t\phi} &= \p_t\widehat\phi - \widehat{\p_z\phi}\p_t\zeta\,. \label{ChainRuleTX}
\end{align}
In terms of the Dirichlet-Neumann operator these are expressed as
\begin{equation}
    G(\zeta)\widehat\phi \coloneqq 
    - \nabla_{\bs{r}}\zeta(\bs{r},t) \cdot \nabla_{\bs{r}}\widehat\phi (\bs{r},t) 
    + \widehat{\p_z\phi}|\nabla_{\bs{r}}\zeta|^2 + \widehat w\,.
    \label{D_N operatorX}
\end{equation}
We consider \eqref{ChainRuleNablaX} together with $\widehat w = \widehat{\mathcal{D}} \zeta$. 
Observe that in the hat notation $\bs{\wh{v}} := \widehat{\nabla_{\bs{r}}\phi}$,  
$\bs{V} := \nabla_{\bs{r}}\widehat\phi$ and $\widehat{w} :=\widehat{\p_z\phi}$ we have%
\footnote{The distinction between velocities $\bs{\wh{v}}$ and $\bs{V}$ is standard \cite{CastroLannes2015,Lannes2020}.}
\begin{equation}
     \widehat{\nabla_{\bs{r}}\phi} 
    = \nabla_{\bs{r}}\widehat\phi - \widehat{w} \nabla_{\bs{r}}\zeta
    \quad\Longrightarrow\quad
    \bs{\wh{v}} = \bs{V} -  \widehat{w} \nabla_{\bs{r}}\zeta
\,,
    \label{v-hatX}
\end{equation}
and hence
\begin{equation}\label{w-hat-eqnX}
    \widehat w = \widehat{\mathcal{D}} \zeta = \p_t\zeta + \bs{\wh{v}}\cdot\nabla_{\bs{r}}\zeta = \p_t\zeta + (\nabla_{\bs{r}}\widehat\phi - \widehat w \nabla_{\bs{r}}\zeta)\cdot\nabla_{\bs{r}}\zeta\,,
\end{equation}
which after rearranging is equivalent to
\begin{equation}\label{u3-zeta1X}
    \p_t\zeta + \bs{\wh{v}}\cdot\nabla_{\bs{r}}\zeta
    = \widehat{w} 
    = \frac{\p_t\zeta + \nabla_{\bs{r}}\widehat\phi\cdot\nabla_{\bs{r}}\zeta}{1+|\nabla_{\bs{r}}\zeta|^2}
    = \frac{\p_t\zeta + \bs{V}\cdot\nabla_{\bs{r}}\zeta}{1+|\nabla_{\bs{r}}\zeta|^2}
    \,.
\end{equation}
Thus, applying the chain rule in the Dirichlet-Neumann operator appearing in the kinematic boundary condition has implied the alternative expression for $\widehat{w}$ in equation \eqref{u3-zeta1X}. The alternative equations for $\wh{w}$ in \eqref{u3-zeta1X} will be used next in section \ref{subsec:derivingCWWE} to close the system defined by Choi's relation \eqref{ChoiEqnX} by using a variational principle reminiscent of the approach in \cite{Luke1967} to derive an evolutionary equation for $\widehat{w}$. The alternative expressions in \eqref{u3-zeta1X} obtained from the Dirichlet-Neumann operator will also inspire a wave-current coupling term in section \ref{subsec:CWWE_with_Coupling}.

\begin{remark}[Direct derivation of the CWWE]
The CWWE  \eqref{LannesZetaEquationX} and \eqref{LannesPsiEquationX} may be derived from the standard three-dimensional form of the Euler motion equation with constant $\rho=\rho_0$. In the case of three-dimensional irrotational flow, one finds Bernoulli's integrated form of the Euler equation
\begin{equation}\label{BernoulliEulerX}
    \p_t\phi + \frac{1}{2}|\nabla\phi|^2 + gz = -\frac{1}{\rho_0}\pi\,.
\end{equation}
Evaluating \eqref{BernoulliEulerX} on $z=\zeta(\bs{r},t)$ with the constraint $\pi|_{z=\zeta}=0$ that the non-hydrostatic pressure vanishes on the free surface yields
\begin{equation}
    \p_t\widehat\phi - \widehat{\p_z\phi}\,\p_t\zeta + \frac{1}{2}|\nabla_{\bs{r}}\widehat\phi|^2 + \frac{1}{2}\widehat{\p_z\phi}^2(1+|\nabla_{\bs{r}}\zeta|^2) - \widehat{\p_z\phi}\nabla_{\bs{r}}\zeta\cdot\nabla_{\bs{r}}\widehat\phi + g\zeta = 0\,.
\end{equation}
Upon adding and subtracting $\widehat{\p_z\phi}^2|\nabla_{\bs{r}}\zeta|^2$ in the previous equation, one finds
\begin{equation}
    \p_t\widehat\phi + g\zeta + \frac{1}{2}|\nabla_{\bs{r}}\widehat\phi|^2 + \frac{1}{2}\widehat{\p_z\phi}^2(1+|\nabla_{\bs{r}}\zeta|^2) - \widehat{\p_z\phi}(\p_t\zeta + \nabla_{\bs{r}}\widehat\phi\cdot\nabla_{\bs{r}}\zeta - \widehat{\p_z\phi}|\nabla_{\bs{r}}\zeta|^2) - \widehat{\p_z\phi}^2|\nabla_{\bs{r}}\zeta|^2 = 0\,.
\end{equation}
Considering this in tandem with equation \eqref{w-hat-eqnX} for $\wh{w}$ yields
\begin{equation}
    \p_t\widehat\phi + g\zeta +\frac{1}{2}|\nabla_{\bs{r}}\widehat\phi|^2 - \frac{1}{2}\widehat{\p_z\phi}^2(1+|\nabla_{\bs{r}}\zeta|^2) = 0\,,
    \label{Bernoulli-eqn}
\end{equation}
which one observes is equivalent to \eqref{LannesPsiEquationX}.
\end{remark}

\begin{remark}
Equation \eqref{u3-zeta1X} expresses $\widehat w$ in terms of the time derivative of $\zeta$ in the frame of reference moving with horizontal velocity $\nabla_{\bs{r}}\widehat\phi$ rather than with velocity $\widehat{\nabla_{\bs{r}}\phi}$. We recall the relation \eqref{v-hatX} and write
\begin{equation}
\bs{V} := \nabla_{\bs{r}}\widehat\phi =  \widehat{\nabla_{\bs{r}}\phi}  + \widehat{w} \nabla_{\bs{r}}\zeta 
=: \bs{\wh{v}} +  \bs{s} \,,
    \label{slipvelocity-defX}
\end{equation}
Physically, $\bs{\wh{v}} := \widehat{\nabla_{\bs{r}}\phi}$ may be interpreted as the fluid transport velocity relative to a Galilean frame moving with velocity $\bs{s}:=\widehat{w}\nabla_{\bs{r}}\zeta$, while the quantity $\bs{V} := \nabla_{\bs{r}}\widehat\phi $ is the total fluid velocity in the inertial frame of the Eulerian fluid description. In fact, the variational formulation taken below will show that  the quantity $\bs{V}$ is the momentum per unit mass given by the variational derivative with respect to transport velocity $\bs{\wh{v}} $ of the Lagrangian in Hamilton's principle for the wave-current dynamics. Likewise,  the quantity $ \bs{s} =  \widehat{w} \nabla_{\bs{r}}\zeta$ will turn out to be the CWW wave momentum per unit fluid mass derived from Hamilton's principle. 
\end{remark}

The surface boundary condition \eqref{u3-zeta1X} yields the evolution equation for the elevation $\zeta$ written in the two different frames of motion as,
\begin{align}
    \p_t\zeta + \bs{\wh{v}}\cdot\nabla_{\bs{r}}\zeta &= \widehat w\, 
    \quad\hbox{and}
\label{FSdyn-v-hat}\\
    \p_t\zeta + \bs{V}\cdot\nabla_{\bs{r}}\zeta &= \widehat{w}(1+|\nabla_{\bs{r}}\zeta|^2)
    = \widehat{w} - \bs{s}\cdot\nabla_{\bs{r}}\zeta
    \,.
\label{FSdyn-V}
\end{align}
Equating $\widehat w$ in these two expressions then  yields 
\begin{equation}\label{u3-zetaX}
    \p_t\zeta + \bs{\wh{v}}\cdot\nabla_{\bs{r}}\zeta 
    = \widehat w = \frac{\p_t\zeta + \bs{V}\cdot\nabla_{\bs{r}}\zeta}{1+|\nabla_{\bs{r}}\zeta|^2}\,.
\end{equation}
Likewise, equations \eqref{Bernoulli-eqn} and \eqref{slipvelocity-defX} yield the Bernoulli evolution equation for zero non-hydrostatic pressure in terms of the rotational fluid velocities $\bs{\wh{v}}$ and $\bs{V}$,
\begin{equation}
    \p_t\widehat\phi + g\zeta +\frac{1}{2}|\bs{V}|^2 - \frac{1}{2}\widehat{w}^2(1+|\nabla_{\bs{r}}\zeta|^2) 
    = 0
    =  \p_t\widehat\phi + \bs{\wh{v}}\cdot\nabla_{\bs{r}}\widehat\phi + g\zeta  - \frac{1}{2}|\bs{\wh{v}}|^2 
    - \frac{1}{2}\widehat{w}^2
    \,.
    \label{Bernoulli-eqn-rot}
\end{equation}
This completes the direct derivation of the CWWE.

%%%%%%%%%%%%%%%%%%%%%%%%%%%%%%%%%%%%%%%%%%%%%%%%%%%%%%%%%%%%%%%%%%%%%%%%%%%%%%
\subsection{Imposing CWWE as constraints in Hamilton's principle}\label{subsec:derivingCWWE}

Let us introduce the following \emph{dimension-free} action integral for a variational principle, $\delta S = 0$, which comprises the sum of the kinematic boundary condition for the elevation $\zeta$ in \eqref{FSdyn-v-hat} and the Bernoulli equation for zero non-hydrostatic pressure in \eqref{Bernoulli-eqn-rot}. These two CWWE conditions are constrained to hold by the two Lagrange multipliers $\lambda$ and $D$, respectively, 
\begin{align}
\begin{split}
	S &=\int \ell(\bs{\wh{v}}, D,\widehat{\phi},{\wh w}, \zeta, \lambda)dt
	\\&= \int \int  
	(\sigma^2\lambda)\Big( \partial_t\zeta  + \bs{\wh{v}}\cdot\nabla_{\bs{r}}\zeta - \widehat{w} \Big) 
	- D\left(\p_t\widehat\phi + \bs{\wh{v}}\cdot\nabla_{\bs{r}}\widehat\phi + \frac{\zeta}{Fr^2}  - \frac{1}{2}|\bs{\wh{v}}|^2 
    - \frac{\sigma^2}{2}\widehat{w}^2\right) \,d^2r\,dt
	\\&= \int \int  
	D\left(\frac{1}{2}\big(|\bs{\wh{v}}|^2+\sigma^2\widehat{w}^2\big)  - \frac{\zeta}{Fr^2}\right)   
	+ \sigma^2\lambda\Big( \partial_t\zeta  + \bs{\wh{v}}\cdot\nabla_{\bs{r}}\zeta - \widehat{w} \Big)
	\\
	&\qquad \quad
	+ \widehat{\phi}\big(\partial_t D + \text{div}_{\bs{r}}(D\bs{\wh{v}})\big)
	\,d^2r\,dt\,.
\end{split}
	\label{ActionIntegral-FS}
\end{align}
For spatial integration by parts, we take natural boundary conditions so the boundary terms vanish. The temporal integration by parts introduces a total time derivative, so it also does not contribute to the equations of motion. 
In equation \eqref{ActionIntegral-FS}, $\bs{r}=(x,y)$ denotes horizontal Eulerian spatial coordinates in an inertial (fixed) domain. 
We have integrated by parts in time and in space after the second line, dropping boundary terms both times.
The constants $\sigma^2$ and $Fr^2$ here are squares of the aspect ratio and the Froude number, respectively, which are obtained in making the expression dimension-free. Finally, we make the distinction between the wave variables $\wh w$ and $\zeta$, and the current variables $D$ and $\bs{\wh v}$. The remaining variables $\lambda$ and $\wh\phi$ in the final form of the action integral \eqref{ActionIntegral-FS} are Lagrange multipliers, to be determined from the others.

\begin{remark}[Non-dimensional parameters]\label{non-dim-scales}
Explicitly, the action integral for free surface motion in \eqref{ActionIntegral-FS} has been cast into dimension-free form by introducing natural units for horizontal length, $[L]$, horizontal velocity, $[V]$, time, $[T]=[L]/[V]$, vertical velocity, $[W]$ and vertical wave elevation, $[\zeta]$. In terms of these units we have defined the following dimension-free parameters: aspect ratio, $[W]/[V]=\sigma = [\zeta]/[L]$ and Froude number, $Fr^2=[V]^2/([g][\zeta])$, for typical wave elevation scale $[\zeta]$.
\end{remark}

\paragraph{Interpreting the two equivalent forms of the action integral in \eqref{ActionIntegral-FS}.}
\begin{itemize}
\item
The second line of the action integral in \eqref{ActionIntegral-FS} may be regarded as a variant of the action integral in Luke \cite{Luke1967}. An action integral for CWWE in \cite{Luke1967} was derived in terms of vertically integrated expressions. In contrast, here the action integral in \eqref{ActionIntegral-FS} has been made two-dimensional by using the Dirichlet-Neumann operator relation to project out the third (vertical) dimension. 
The Lagrange multiplier $\lambda$ enforces the kinematic boundary condition for the elevation $\zeta$ in \eqref{FSdyn-v-hat}. Likewise, 
the Lagrange multiplier $D$ enforces the zero-pressure Bernoulli law \eqref{Bernoulli-eqn-rot} obtained from the Dirichlet-Neumann operator. 
%See also \cite{Bateman1929}.
\item
In the last line of \eqref{ActionIntegral-FS}, we rearrange the constraints in the action integral in the second line into the standard Clebsch advection form for two-dimensional fluid motion, by integrating by parts in time and (horizontal) space. We may then regard the quantity $D\,d^2r$ as the area measure on the horizontal domain. That is, the area measure $D\,d^2r$ is advected by $\bs{\wh{v}}$, which is imposed in the last line by regarding the trace of the velocity potential on the free surface $\widehat{\phi}$ as a Lagrange multiplier. 

\end{itemize}

\begin{remark}[The velocity $\bs{\wh v}$ can have non-zero vorticity]
The momentum map in equation \eqref{eq: momap-FS} makes it clear that the advective transport velocity $\bs{\wh v}$ has non-zero vorticity
\begin{align}
{\wh \omega}:={\rm curl}_{\bs{r}}\,\bs{\wh v} = -\,\bs{\wh{z}}\cdot \nabla_{\bs{r}}\wh{w}\times \nabla_{\bs{r}}\zeta =: -J({\wh w},\zeta)\ne0\,.
	\label{eq: omega-hat-FS}
\end{align}
Consequently, the canonical constraint equations appearing in the last line of \eqref{ActionIntegral-FS} are not potential flows. In contrast, equation \eqref{slipvelocity-defX} shows that $\bs{V}=\nabla_{\bs{r}}\wh{\phi}$ is indeed a potential velocity.
\end{remark}

\paragraph{Variational formulas.}
Taking variations of the action integral \eqref{ActionIntegral-FS} yields 
\begin{align}
\begin{split}
	\delta \bs{\wh{v}}:&\quad 
	D \bs{\wh{v}}\cdot d\bs{r} + \sigma^2\lambda\,d\zeta 
	=  Dd\widehat{\phi} 
	\quad\Longrightarrow\quad 
	\bs{V}\cdot d\bs{r} :=
	\bs{\wh{v}}\cdot d\bs{r} + \sigma^2\widehat{w}\,d\zeta 
	=  d\widehat{\phi}
	\,,\\ 
	\delta \widehat{w} :&\quad 
	D\widehat{w} - \lambda = 0
	\,,\\
	\delta \lambda:&\quad 
	 \partial_t\zeta  + \bs{\wh{v}}\cdot\nabla_{\bs{r}}\zeta =  \widehat{w}
	\,,\\
	\delta \zeta:&\quad 
	\partial_t \lambda + \text{div}_{\bs{r}}(\lambda\bs{\wh{v}})
	=
	 - \,\frac{D}{\sigma^2Fr^2}
	\quad\Longrightarrow\quad 
	\partial_t \widehat{w} + \bs{\wh{v}}\cdot\nabla_{\bs{r}}\widehat{w}
	=
	 - \,\frac{1}{\sigma^2Fr^2}
	\,,\\
	\delta \widehat{\phi} :& \quad 
	\partial_t D + \text{div}_{\bs{r}}(D\bs{\wh{v}}) =0 
	\,,\\
	\delta D:&\quad 
	\big(\partial_t + \bs{\wh{v}}\cdot \nabla_{\bs{r}}\big)\widehat{\phi} 
	= \frac{1}{2}\big(|\bs{\wh{v}}|^2+\sigma^2\widehat{w}^2\big)  - \frac{\zeta}{Fr^2} =:\varpi 
	\,.
\end{split}
	\label{var-eqns-FS-CWW}
\end{align}
Applying the Lagrangian time derivative $(\p_t +\mathcal{L}_{\bs{\wh{v}}})$ to the first relation in \eqref{var-eqns-FS-CWW} yields the ECWW motion equation,
\begin{equation}\label{EP-WW-motion-eqn}
    (\p_t +\mathcal{L}_{\bs{\wh{v}}})(\bs{\wh{v}}\cdot d\bs{r} + \sigma^2\widehat{w}d\zeta) 
    = (\p_t +\mathcal{L}_{\bs{\wh{v}}})(\bs{V}\cdot d\bs{r})
    = (\p_t +\mathcal{L}_{\bs{\wh{v}}})d\widehat{\phi}
    = d\varpi \,.
\end{equation}
See appendix \ref{app: FluidTransTheory} for more discussion of the Lie derivative notation (as in $\mathcal{L}_{\bs{\wh{v}}}$) which is defined by the Lagrangian time derivative. The quantity $d\varpi$ is the spatial differential (i.e., the gradient) of Bernoulli's law in the last line of \eqref{var-eqns-FS-CWW}. 

\paragraph{Kelvin circulation theorems for ECWWE in their dimensional form.}
Moving to the dimensional form and continuing to calculate from the ECWW motion equation in \eqref{EP-WW-motion-eqn}, we have
\begin{equation}
    \begin{aligned}
        0 &= (\p_t +\mathcal{L}_{\bs{\wh{v}}})(\bs{\wh{v}}\cdot d\bs{r} + \widehat{w}d\zeta) - d\varpi
        \\&= (\p_t +\mathcal{L}_{\bs{\wh{v}}})(\bs{\wh{v}}\cdot d\bs{r}) - \frac12d|\bs{\wh{v}}|^2
        \\&\qquad 
        + (\p_t +\mathcal{L}_{\bs{\wh{v}}}) (\wh{w}d\zeta) - d \Big(\frac12 \wh{w}^2 - g\, \zeta\Big)
        \,.
    \end{aligned}
    \label{EP-WW-motion-eqn-dim}
\end{equation}
Remarkably, the $(\wh{w},\zeta)$ equations in \eqref{var-eqns-FS-CWW} imply that the previous equation separates into two transport equations, namely,
\begin{equation}
    \begin{aligned}
        (\p_t &+\mathcal{L}_{\bs{\wh{v}}})(\bs{\wh{v}}\cdot d\bs{r}) - \frac12d|\bs{\wh{v}}|^2 = 0
        \,,\\
        (\p_t &+\mathcal{L}_{\bs{\wh{v}}}) (\wh{w}d\zeta) - d \Big(\frac12 \wh{w}^2 - g\, \zeta\Big) = 0
        \,.
    \label{eq: 2 KelThms}
    \end{aligned}
\end{equation}
Thus, the wave and current circulations are conserved separately, in a \emph{mutual non-acceleration pact},
\begin{equation}
    \begin{aligned}
        \frac{d}{dt}\oint _{c(\bs{\wh{v}})}\bs{\wh{v}}\cdot d\bs{r} &=  \oint _{c(\bs{\wh{v}})} \frac12d|\bs{\wh{v}}|^2 = 0
        \,,\\
        \frac{d}{dt}\oint _{c(\bs{\wh{v}})} \wh{w}d\zeta &=  \oint _{c(\bs{\wh{v}})} d \Big(\frac12 \wh{w}^2 - g\, \zeta\Big) = 0
        \,.
    \label{eq: 2 KelCircThms}
    \end{aligned}
\end{equation}
The separation of conservation laws in \eqref{eq: 2 KelCircThms} means that the two degrees of freedom do not influence each other's circulation. 
Actually, this separation is a general feature of wave-current interaction theories which arise from Hamilton's principle with a phase-space Lagrangian, \cite{Holm2020}.  

\paragraph{Reduction of the ECWW motion equation to the pressureless Euler fluid equation.}
Because of a cancellation of $\frac12d|\bs{\wh{v}}|^2$ in equation \eqref{eq: 2 KelThms} with the Lie derivative term, the $\bs{\wh{v}}$-equation simplifies further to produce the following \emph{pressureless Euler fluid equation} for the transport velocity $\bs{\wh v}=\nabla\wh{\phi} -  \wh{w} \nabla \zeta $,
\begin{equation}\label{WW-Pressureless2DEuler}
    \p_t\bs{\wh v} + \bs{\wh v}\cdot\nabla_{\bs{r}}\bs{\wh v} = 0
    \quad\hbox{and}\quad
    \partial_t D + \text{div}_{\bs{r}}(D\bs{\wh{v}})\,.
\end{equation}
Thus, while the vector $\bs{\wh v}$ transports the density $D$, it also transports itself as though it were an array of two advected scalars, $(\wh{v}_1,\wh{v}_2)$. This feature further simplifies the interpretation of the $\bs{\wh{v}}$-equation, because it can now be seen as an \emph{inviscid Burgers equation}.
However, note that the compressible ``Burgers velocity'' $ \bs{\wh{v}}$ in \eqref{WW-Pressureless2DEuler} has vorticity $\wh{\omega} := \bs{\wh{z}}\cdot{\rm curl \bs{\wh v}} = - J(\wh{w} , \zeta)$ which does not vanish, in general. However, the relation $\bs{\wh v}=\nabla\wh{\phi} -  \wh{w} \nabla \zeta $ and the second equation in \eqref{eq: 2 KelThms} do imply that 
\begin{equation}\label{WW-Pressureless2DEulervort}
    \p_t\wh{\omega} + \bs{\wh v}\cdot\nabla_{\bs{r}}\wh{\omega} = 0\,.
\end{equation}
Hence, if the vorticity $\wh{\omega}$ vanishes initially, it will remain so. In this case, the pressureless 2D Euler equation in \eqref{WW-Pressureless2DEulervort} reduces to the well-studied two-dimensional Hamilton-Jacobi equation for $\wh{\phi}$. See, e.g., \cite{John1953,FrischBec2001} for reviews. 

\paragraph{Back to the ECWWE in their dimensional forms.}
We may restore the ECWWE to their dimensional forms as
\begin{align}
\begin{split}
	\partial_t\widehat{\phi}  + \bs{\wh v}\cdot \nabla_{\bs{r}}\widehat{\phi} 
	&= \frac{1}{2}\big(|\bs{\wh v}|^2
	+ \widehat{w}^2\big)  
	- g\zeta
	\,,
	\\
	\partial_t\zeta  + \bs{\wh v}\cdot\nabla_{\bs{r}}\zeta 
	&=  \widehat{w}
	\,,
	\\
	\partial_t \widehat{w} + \bs{\wh v}\cdot\nabla_{\bs{r}}\widehat{w}
	&=
	 - \,g
	 \,,
	 \\
	 \partial_t D + \text{div}_{\bs{r}}(D\bs{\wh{v}}) &=0\,,
\end{split}
	\label{WW-eqns1}
\end{align}
where $\bs{\wh v}$ evolves according to \eqref{WW-Pressureless2DEuler}.
One observes that the first two equations in \eqref{WW-eqns1} are equivalent to the CWWE discussed in section \ref{subsec:CWWE}, since $\bs{\wh v} = \widehat{\nabla_{\bs{r}}\phi}$.

\subsection{Derivation of ECWW equations with non-hydrostatic pressure}\label{subsec:CWWEPressure}

In the standard derivation of the CWW equations \eqref{LannesZetaEquationX} and \eqref{LannesPsiEquationX}, the three dimensional pressure $\pi$ is taken to be zero on the surface and thus the resulting equations of motion have no pressure term. In order for the variational equations we have derived in section \ref{subsec:derivingCWWE} to match equations \eqref{LannesZetaEquationX} and \eqref{LannesPsiEquationX}, we have derived \emph{compressible} equations and thus the system also contains the additional equation for $D$.
Should we want to model an incompressible flow, and avoid having an equation for $D$, we must introduce pressure as a Lagrange multiplier which enforces that $D=1$. Of course, such a non-hydrostatic pressure would be incompatible with assuming the pressure is zero on the surface. We derive the ECWW equations with non-hydrostatic pressure and incompressible transport velocity by varying the action integral defined  in its dimensional form by

\begin{align}
\begin{split}
	S &= \int \int  
	D\left(\frac{1}{2}\big(|\bs{\wh{v}}|^2+\widehat{w}^2\big)  - g\zeta\right)   
	+ \lambda\Big( \partial_t\zeta  + \bs{\wh{v}}\cdot\nabla_{\bs{r}}\zeta - \widehat{w} \Big)
	\\
	&\qquad 
	+ \widehat{\phi}\big(\partial_t D + \text{div}_{\bs{r}}(D\bs{\wh{v}})\big) - p(D-1)
	\,d^2r\,dt\,.
\end{split}
	\label{ActionIntegral-FS_pressure}
\end{align}

Proceeding in the same manner as in section \ref{subsec:derivingCWWE}, and omitting the calculations since they are very much alike, we derive the following system of equations  in their dimensional forms

\begin{align}
\begin{split}
	\wh{\mathcal{D}}\wh\phi \coloneqq \partial_t\widehat{\phi}  + \bs{\wh v}\cdot \nabla_{\bs{r}}\widehat{\phi} 
	&= \frac{1}{2}\big(|\bs{\wh v}|^2
	+ \widehat{w}^2\big)  
	- g\zeta - p
	\,,
	\\
	\wh{\mathcal{D}}\zeta \coloneqq \partial_t\zeta  + \bs{\wh v}\cdot\nabla_{\bs{r}}\zeta 
	&=  \widehat{w}
	\,,
	\\
	\wh{\mathcal{D}}\wh w \coloneqq \partial_t \widehat{w} + \bs{\wh v}\cdot\nabla_{\bs{r}}\widehat{w}
	&=
	- \,g
	\,.
\end{split}
	\label{WW-eqns-D}
\end{align}
Here, the divergence-free transport velocity $\bs{\wh v}$ satisfies a two dimensional Euler equation
\begin{equation}\label{TransportEuler2D}
    \p_t\bs{\wh v} + \bs{\wh v}\cdot\nabla_{\bs{r}}\bs{\wh v} = -\nabla_{\bs{r}} p
    \,,\quad\hbox{where}\quad D=1 \quad\hbox{implies}\quad 
    \text{div}_{\bs{r}}\bs{\wh v} = 0\,.
\end{equation}

We may understand the structure of this problem further by comparing it to equation \eqref{projectedmotioneqnX} with $\rho_0 = 1$. Noting that $\wh{\mathcal{D}}\wh{w} =-g$, equations \eqref{projectedmotioneqnX} and \eqref{TransportEuler2D} together imply
\begin{equation}
    \wh{\mathcal{D}}\bs{\wh v} = - \nabla_{\bs{r}}\wh\pi = - \nabla_{\bs{r}}p\,.
\end{equation}\label{WW-motion-equation}
This comparison implies the following remarkable observation, which turns out to be one of our main conclusions about this approach because it provides a closure of the CWWE. 
\begin{theorem}\label{Thm: pressure projection}
The pressure, $p$, in the two dimensional model \eqref{WW-eqns-D} is equivalent to the pressure of the three dimensional fluid evaluated on the free surface, $\pi$, up to the addition of a spatial constant.  
\end{theorem}

%%%%%%%%%%%%%%%%%%%%%%%%%%%%%%%%%%%%%%%%%%%%%%%%%%%%%%%%%%%%%%%%%%%%%%%%%%%%%%
\subsection{Conservation laws for the compressible ECWW dynamical system}
%%%%%%%%%%%%%%%%%%%%%%%%%%%%%%%%%%%%%%%%%%%%%%%%%%%%%%%%%%%%%%%%%%%%%%%%%%%%%%
From here, we return to the compressible ECWW equations by removing the incompressibility constraint imposed by the pressure.
\subsubsection{Eulerian conservation laws for the ECWWE}
The system of ECWWE in \eqref{WW-Pressureless2DEuler} and \eqref{WW-eqns1} possesses the following fundamental Eulerian conservation laws in a domain $\Omega$ with fixed boundaries. 
\begin{enumerate}
    \item 
    The last equation in \eqref{WW-eqns1} implies conservation of \emph{mass}, $\mathbb{D} := \int_{\Omega} D\,d^2r$,
    \[
    \frac{d\mathbb{D}}{dt} :=\frac{d}{dt} \int_{\Omega} D\,d^2r = \int_{\Omega} \p_t D \,d^2r 
    = - \int_{\Omega} \text{div}_{\bs{r}}(D\bs{\wh{v}})\,d^2r 
    = - \oint_{\p\Omega} D\bs{\wh{v}}\cdot\bs{\wh{n}}\,ds = 0
    \,,\]
    for $\bs{\wh{v}}\cdot\bs{\wh{n}}$ on the boundary $\p\Omega$ with normal vector $\bs{\wh{n}}$.
    \item 
    Equation \eqref{WW-Pressureless2DEuler} and the last equation in \eqref{WW-eqns1} imply conservation of \emph{momentum}, defined by
    \[
    \frac{d\mathbb{M}_j}{dt} :=
    \frac{d}{dt} \int_{\Omega} D\wh{v}_j\,d^2r =  - \int_{\Omega} \p_k (D\wh{v}_j\wh{v}^k)\,d^2r 
    = - \oint_{\p\Omega} D\wh{v}_j \bs{\wh{v}}\cdot\bs{\wh{n}}\,ds = 0
    \,.\]
    \item
    Combining the curl$_{\bs{r}}$ of equation \eqref{WW-Pressureless2DEuler} and the last equation in \eqref{WW-eqns1} 
    implies conservation of \emph{mass-weighted enstrophy}, defined by 
    \[
    \mathbb{C}_{\Phi}:= \int_{\Omega} D \Phi(\wh{\omega}) \,d^2r\,,
    \] 
    for any differentiable function $\Phi$ of vorticity, $\wh{\omega}$, which itself is defined by
    \[
    \wh{\omega} := \bs{\wh{z}}\cdot {\rm curl}_{\bs{r}}\bs{\wh{v}} = - \bs{\wh{z}}\cdot \nabla_{\bs{r}}\wh{w}\times\nabla_{\bs{r}}\zeta
    =: - J(\wh{w},\zeta)
    \,.\]
    Thus, upon noticing that vorticity $\wh{\omega}$ is advected as a scalar by the flow of $\bs{\wh{v}}$, 
    we find advection of any function of $\Phi(\wh{\omega})$, as well, by the chain rule and linearity of the advection operator for scalars. Namely,
    \[
    (\p_t + \bs{\wh{v}}\cdot \nabla_{\bs{r}})\Phi(\wh{\omega}) = 0\,.
    \]
    Thus, we obtain conservation of mass-weighted enstrophy from the continuity equation, the chain rule and integration by parts, as follows,
    \[
    \frac{d\mathbb{C}_{\Phi}}{dt} := 
    \int_{\Omega}\p_t (D \Phi(\wh{\omega})) \,d^2r =  - \int_{\Omega} \p_k (D\Phi(\wh{\omega})\wh{v}^k )\,d^2r 
    = - \oint_{\p\Omega} D\Phi(\wh{\omega})\,\bs{\wh{v}}\cdot\bs{\wh{n}}\,ds = 0
    \,,
    \]
    for $\bs{\wh{v}}\cdot\bs{\wh{n}}$ on the fixed boundary $\p\Omega$. 
    Thus, the $D$-weighted $L^p$ norm of the vorticity $\wh{\omega}={\rm curl}\bs{\wh{v}}$ is controlled. 
\item
The corresponding conserved energy is given by
\begin{align}
\begin{split}
E(\bs{\wh{v}},\wh{w},\zeta , D) &=
\int \left(\frac{1}{2} |\bs{\wh{v}}|^2 + \frac{1}{2} \wh{w}^2
 + g\zeta \right) \,D\,d^2r
\\&=
\int \left(\frac{1}{2} |\widehat{\nabla_{\bs{r}}\phi}|^2 + \frac{1}{2} \wh{w}^2
 + g\zeta \right) \,D\,d^2r
\,.
\end{split}
	\label{WW-LP-Erg-}
\end{align}
This expression follows quite easily from the Legendre transformation of the Lagrangian in 
\eqref{ActionIntegral-FS_pressure}.

\end{enumerate}

%%%%%%%%%%%%%%%%%%%%%%%%%%%%%%%%%%%%%%%%%%%%%%%%%%%%%%%%%%%%%%%%%%%%%%%%%%%%%%
\subsubsection{Moment dynamics of a Lagrangian fluid blob under the ECWWE}\label{Blob-dyn-ECWW}

We re-write the continuity equation in \eqref{WW-eqns1} and its associated motion equation in \eqref{WW-Pressureless2DEuler} as 
Lagrangian conservation laws for mass and momentum,
\begin{align}
\begin{split}
(\p_t + \mathcal{L}_{\wh{v}})(D \,d^2r) &= (\p_t D +  \p_k P^k)\,d^2r = 0
\quad\hbox{with}\quad
P^k := D\wh{v}^k\,,
\\
(\p_t + \mathcal{L}_{\wh{v}})(P^j \,d^2r) &=
\big( \p_t P^j +  \p_k (P^j\wh{v}^k) \big)\,d^2r = 0\,.
\end{split}
	\label{eqn-virial1}
\end{align}
Consider a two dimensional `blob' of fluid mass occupying a Lagrangian domain of fluid $\Omega(t)$ which is deforming under the ECWWE flow of the free-surface fluid velocity $\bs{\wh{v}}$ so that no fluid material enters or leaves through its moving boundary $\p\Omega(t)$. In this situation, we have the following Reynolds transport relations for the dynamics of the spatial moments of the mass distribution within the blob. 
\begin{enumerate}
    \item 
The total mass of a Lagrangian blob is conserved:
\[
\frac{d}{dt}\int_{\Omega(t)} D \,d^2r = \int_{\Omega(t)} (\p_t + \mathcal{L}_{\wh{v}})(D \,d^2r) 
=  \int_{\Omega(t)} \big(\p_t D + \p_k(D\wh{v}^k)\big)\,d^2r  = 0\,.
\]
    \item 
The rate of change of the centre of mass of the blob is its conserved momentum:
\[
\frac{d}{dt}\int_{\Omega(t)} r^j D \,d^2r = \int_{\Omega(t)} (\p_t + \mathcal{L}_{\wh{v}})(r^j D \,d^2r) 
= \int_{\Omega(t)} \wh{v}^j\, D \,d^2r = \int_{\Omega(t)} P^j\,d^2r \,.
\]
Conservation of the blob momentum is shown by a direct computation, 
\[
\frac{d^2}{dt^2}\int_{\Omega(t)} r^j D \,d^2r = \frac{d}{dt}\int_{\Omega(t)} P^j\,d^2r
= \int_{\Omega(t)} (\p_t + \mathcal{L}_{\wh{v}})(P^j \,d^2r) = 0\,.
\]
    \item
The moment of inertia $I^{ij}=\int_{\Omega(t)} r^i r^j D \,d^2r $ represents the elliptical shape of the blob. 
Its rate of change may be computed as
\[
\frac{d}{dt} I^{ij} = \frac{d}{dt}\int_{\Omega(t)} r^i r^j D \,d^2r 
= \int_{\Omega(t)} (\wh{v}^ir^j + r^i\wh{v}^j) D \,d^2r  \,.
\]
    \item 
The acceleration of the elliptical shape of the blob is governed by
\[
\frac{d^2}{dt^2} I^{ij} = \frac{d}{dt}\int_{\Omega(t)} (P^i\,d^2r) r^j + r^i(P^j\,d^2r) 
= \int_{\Omega(t)} (P^i \wh{v}^j + \wh{v}^i P^j)\,d^2r = \int_{\Omega(t)} (\wh{v}^i \wh{v}^j + \wh{v}^i\wh{v}^j)\,D\,d^2r \,.
\]
    \item 
Remarkably, the acceleration of the trace of the moment of inertia ${\rm tr}(I)$ is positive-definite
\[
\frac{d^2}{dt^2} {\rm tr}(I) = 2\int_{\Omega(t)} D |\bs{\wh{v}}|^2 \,d^2r > 0\,.
\]
This is a simple version of the tensor virial theorem \cite{Chandra1989}.  Here, the tensor virial theorem implies that under ECWWE flow equations in \eqref{eqn-virial1} any initial distribution of mass will expand outward at an acceleration rate proportional to the kinetic energy within its Lagrangian boundary. Because this result holds for every Lagrangian blob of fluid undergoing this motion it follows that the mass density cannot become singular in an infinite flow domain. This means that the measure $Dd^2r$ cannot become a Dirac measure.
    \item 
Finally, we notice that blob angular momentum $L^{ij}:= \int_{\Omega(t)}(\wh{v}^i r^j - r^i \wh{v}^j)D\,d^2r$ is conserved under the ECWWE flow, since
\[
\frac{d}{dt}L^{ij} := \frac{d}{dt}\int_{\Omega(t)} (P^i\,d^2r) r^j - r^i(P^j\,d^2r)
= \int_{\Omega(t)} (\wh{v}^i \wh{v}^j - \wh{v}^i\wh{v}^j)\,D\,d^2r = 0.
\]
\end{enumerate}

%%%%%%%%%%%%%%%%%%%%%%%%%%%%%%%%%%%%%%%%%%%%%%%%%%%%%%%%%%%%%%%%%%%%%%%%%%%%%%

%%%%%%%%%%%%%%%%%%%%%%%%%%%%%%%%%%%%%%%%%%%%%%%%%%%%%%%%%%%%%%%%%%%%%%%%%%%%%%
\subsection{Three Hamiltonian formulations of the ECWWE using free-surface variables}\label{sec: HamForm-CWWE}

This section derives three equivalent Hamiltonian formulations of the system of ECWWE in \eqref{WW-Pressureless2DEuler} and \eqref{WW-eqns1}. To set the stage, let us first remark on the previous literature concerning Hamiltonian formulations of fluid dynamics with free-boundaries.

{
\begin{remark}[Previous Hamiltonian formulations of fluid dynamics with free-boundaries.]\label{Rem: PrevHamForms}
    The ECWWE model extends the CWW model to permit rotational flow. Before investigating its Hamiltonian formulation we recall here the result of Zakharov \cite{Zakharov1968} that the CWWE also have a Hamiltonian structure with similarities to the Hamiltonian \eqref{CFS-Ham1}. Indeed, the water wave equations have canonical variables $\zeta$ and $\phi$, and a Hamiltonian defined by
    \begin{equation*}
        \frac{1}{2}\int\int |\nabla\phi|^2\,d^2x\,dz + \frac{g}{2}\int \zeta^2 \,d^2x\,,
    \end{equation*}
    in the case of zero surface tension. There are some similarities between this Hamiltonian structure of the water-wave equations and the full system of equations we have derived. However, the Hamiltonian for the water-wave equations and one of the canonical variables are vertically integrated compared to \eqref{WW-LP-Erg-}, which is evaluated on the free surface.
    
    Lewis et al.  \cite{LMMR1986} generalized the previous canonical structure of Zakharov \cite{Zakharov1968} for irrotational flow to obtain Hamiltonian structures for 2- or 3-dimensional incompressible flows with a free boundary. The Poisson bracket in \cite{LMMR1986} was determined using reduction from canonical variables in the Lagrangian (material) description. The corresponding Hamiltonian form for the equations of a liquid drop with a free boundary having surface tension was also demonstrated, as was the structure of the bracket in terms of a reduced cotangent bundle of a principal bundle was explained. In the case of two-dimensional flows, a vorticity bracket was determined and the generalized enstrophy was shown to be a Casimir function. 
    
A Hamiltonian description of free boundary fluids has also been studied in  Mazer and Ratiu \cite{MazerRatiu1989}. 
In \cite{MazerRatiu1989}, the Hamiltonian formulation of adiabatic free boundary inviscid fluid flow using only physical variables was presented in both the material and spatial formulation. By using the symmetry of particle relabeling, the noncanonical Poisson bracket in Eulerian representation was derived as a reduction from the canonical bracket in the  Lagrangian representation. When the free boundary of the fluid was specified as the zero level set of an array of advected functions (e.g., Lagrangian labels carried by the fluid flow), the formulation of \cite{MazerRatiu1989} recovered the Lie Poisson bracket of \cite{HDIA-etal1988}, as well as the corresponding potential vorticity and other conserved quantities found in \cite{HDIA+Holm1987}.

In a tour-de-force, Gay-Balmaz, Marsden and Ratiu \cite{GBMR2012} carried out Lagrangian reduction for free boundary fluids and deduced both the equations of motion and their associated constrained variational principles in both the convective and spatial representations. To follow up, Gay-Balmaz and Vizman  constructed dual pairs for free boundary fluids in \cite{GBViz2015}. 

Finally, we mention that Castro and Lannes \cite{CastroLannes2015} proposed what they called a \emph{formal} Poisson bracket for an extension of CWWE which differs from the present work by combining vertically integrated variables with free-surface variables possessing a noncanonical Poisson bracket. They expressed reservations about whether their formal Poisson bracket would satisfy the Jacobi identity and they cited \cite{Kolev2007} which describes potentially problematic technical pitfalls in this regard. However, the \emph{formal} Poisson bracket in Castro and Lannes \cite{CastroLannes2015} actually does satisfy the Jacobi identity, because their Poisson bracket is equivalent to that in Lewis et al.  \cite{LMMR1986}, which does satisfy the Jacobi identity for admissible functionals $A$ such that for every triple of functionals $f,g,h \in A$, the bracket of any of two of them lies in $A$.

None of the previous Hamiltonian formulations of fluid dynamics with free-boundaries described above have represented the free-boundary dynamics in terms of projection/evaluation properties of the Dirichlet-Neumann operator (DNO) representation of CWW theory, as is done in the present approach.  

In contrast, the late Walter Craig and his collaborators in \cite{CG,CGNS} used asymptotic expansions of the DNO for CWWE to derive Hamiltonian formulations of certain soliton equations. 
The efforts of Craig et al. \cite{CG,CGNS} took advantage of the DNO representation of the CWWE to formulate Hamiltonian equations which do not involve vertically integrated variables. These Hamiltonian equations also enabled the study of interesting bathymetry by introducing a more general DNO. See also \cite{AP2015} for a review and bibliography of previous work in Hamiltonian formulations of the wave-current interaction based on the DNO. In contrast, the present work uses the DNO map to extend the CWWE to ECWWE. 

\end{remark}
}

%%%%%%%%%%%%%%%%%%%%%%%%%%%%%%%%%%%%%%%%%%%%%%%%%%%%%%%%%%%%%%%%%%%%%%%%%%%%%%
\subsubsection{Canonical Hamiltonian formulation of the ECWWE }
{
In order to consider the Hamiltonian formulation of this problem, we define a \emph{third} form of the Lagrangian \eqref{ActionIntegral-FS} by performing a Legendre transform as follows
\begin{equation}
\begin{aligned}
    S &= \int \int  
	(\sigma^2\lambda)\Big( \partial_t\zeta  + \bs{\wh{v}}\cdot\nabla_{\bs{r}}\zeta - \widehat{w} \Big) 
	- D\left(\p_t\widehat\phi + \bs{\wh{v}}\cdot\nabla_{\bs{r}}\widehat\phi + \frac{\zeta}{Fr^2}  - \frac{1}{2}|\bs{\wh{v}}|^2 
    - \frac{\sigma^2}{2}\widehat{w}^2\right) \,d^2r\,dt
 	\\&= \int \int  
	(\sigma^2\lambda) \partial_t\zeta + \widehat{\phi} \partial_t  D  
    - \left(\frac{1}{2D}\Big( |D\bs{\wh{v}}|^2 + (\sigma\lambda)^2 \Big)+ \frac{D\zeta}{Fr^2}\right) 
	\,d^2r\,dt\,,
\end{aligned}
\end{equation}
where, to go from the first line to the second, we have integrated by parts in time and made use of the first two relations in \eqref{var-eqns-FS-CWW}.

We have now expressed the Lagrangian within the action integral \eqref{ActionIntegral-FS}, $\ell(\bs{\wh{v}}, D,\widehat{\phi},{\wh w}, \zeta; \lambda)$, as a phase-space Lagrangian, by re-writing it as a Legendre transform. The phase-space form of the Lagrangian immediately identifies the canonically conjugate pairs of field variables $(\wh{\phi},D)$ and $(\sigma^2\lambda,\zeta)$ and determines the Hamiltonian as 
\begin{align}
H (\wh{\phi},D;\lambda,\zeta) = \int 
\frac{1}{2D}\Big( |D\nabla\wh{\phi} - (\sigma^2\lambda) \nabla \zeta|^2 + (\sigma\lambda)^2 \Big)+ \frac{D\zeta}{Fr^2}
\,d^2r\,.
	\label{eq: Hamiltonian-FS}
\end{align}
The variation of the Lagrangian in any of its equivalent representations in \eqref{ActionIntegral-FS} with respect to the vector-field velocity $\bs{\wh{v}}$ yields the momentum density relation
\begin{align}
\bs{m} :=
D\bs{\wh{v}} = D\nabla\wh{\phi} - D \wh{w} \nabla \zeta  = D\nabla\wh{\phi} - (\sigma^2\lambda) \nabla \zeta
\,.
	\label{eq: momap-FS}
\end{align}
This expression provides a (cotangent lift) momentum map from the canonically conjugate pairs of field variables $(\wh{\phi},D)$ and $(\sigma^2\lambda,\zeta)$ to the momentum density $\bs{m} :=D\bs{\wh{v}}$ which is in concert with equation \eqref{v-hatX}.
}

{
Restoring the dimensions, in the canonical Hamiltonian field variables for currents $(\wh{\phi},D)$ and for waves $(\lambda,\zeta)$, the Bernoulli function ${\varpi}$ in \eqref{var-eqns-FS-CWW} is expressed as
\begin{equation}
	\varpi
	:= \frac{1}{2}\Big(\big|\bs{\wh{v}}\big|^2 + \wh{w}^2\Big) - g\zeta
	= \frac{1}{2} \Big|\nabla\wh{\phi} - \frac{\lambda}{D} \nabla \zeta \Big|^2
    + \frac{\lambda^2}{2D^2}  
	- g\zeta
	\,.
	\label{Bernoulli-varpi-1}
\end{equation}
The corresponding energy Hamiltonian in these variables is given by
\begin{equation}
H(\bs{\wh{v}},\wh{w},\zeta , D) =
\int \left(\frac{1}{2} \Big|\nabla\wh{\phi} - \frac{\lambda}{D} \nabla \zeta \Big|^2
    + \frac{\lambda^2}{2D^2}  
	+ g\zeta \right) \,D\,d^2r
	\label{Hamiltonian-canonical}
\end{equation}
The canonical Hamiltonian equations for ECWWE in terms of the two degrees of freedom comprising wave variables $(\lambda,\zeta)$ and current variables $(\wh{\phi},D)$ are given by
\begin{align}
\begin{split}
\frac{\p}{\p t}
\begin{bmatrix}
\wh{\phi} \\ D \\ \lambda \\ \zeta 
\end{bmatrix}
= -
\begin{bmatrix}
0 & 1 & 0 & 0
\\
-1	    & 0  & 0 & 0
\\
0	& 0 & 0  & 1
\\
0	 & 0  & -1 & 0
\end{bmatrix}
\begin{bmatrix}
 {\delta h}/{\delta \wh{\phi}} = - \,{\rm div}_{\bs{r}}(D\bs{\wh{v}})
 \\ {\delta h}/{\delta D} = \bs{\wh{v}}\cdot\nabla_{\bs{r}}\wh{\phi} -\,\varpi
 \\ 
{\delta h}/{\delta \lambda} = - \bs{\wh{v}}\cdot\nabla_{\bs{r}}\zeta + \lambda/D 
\\ {\delta h}/{\delta \zeta} =  {\rm div}_{\bs{r}} (\lambda\bs{\wh{v}}) + g D 
\end{bmatrix}
\,.
\end{split}
	\label{FS-diag-brkt-CWW-canon}
\end{align}
One observes that the symplectic Poisson operator for the two independent degrees of freedom in \eqref{FS-diag-brkt-CWW-canon} 
appears in the canonical block-diagonal form. The wave-current interactions between these two independent degrees of freedom (waves $(\lambda,\zeta)$ with $ \lambda = D\widehat{w} $ and currents $(\wh{\phi},D)$) are determined by the Hamiltonian in \eqref{Hamiltonian-canonical}.
}

\subsubsection{Entangled Hamiltonian formulation of the ECWWE }
In terms of the canonical Hamiltonian field variables $(\wh{\phi},D)$ and $(\lambda,\zeta)$, 
the total momentum density of the fluid $\bs{\wh{m}}:=D\bs{\wh{v}}$ is defined as the sum
\begin{align}
\bs{m} = D\bs{\wh{v}} = D\nabla\wh{\phi} - \lambda \nabla \zeta\,,
	\label{WW-total-momap-CWW-1}
\end{align}
of both wave and current variables.

The Legendre transform with respect to both pairs of canonical wave variables defining the momentum density $\bs{m}$ leads to the following Hamiltonian,
\begin{align}
\begin{split}
h(\bs{m},D,\lambda,\zeta) 
&= \int \bs{m}\cdot \bs{\wh{v}}\, d^2r - \ell(\bs{\wh{v}}, D,\widehat{\phi},{\wh w}, \zeta; \lambda)
\\&=
\int \frac{1}{2D} \big|\bs{m} \big|^2
+ \frac{\lambda^2}{2D} + gD\zeta \,d^2r
\,.
\end{split}
	\label{WW-LP-Ham-CWW-m}
\end{align}
The corresponding conserved energy was already mentioned in \eqref{WW-LP-Erg-} as
\begin{align}
\begin{split}
E(\bs{\wh{v}},\wh{w},\zeta , D) &=
\int \left(\frac{1}{2} |\bs{\wh{v}}|^2 + \frac{1}{2} \wh{w}^2
 + g\zeta \right) \,D\,d^2r
\\&=
\int \left(\frac{1}{2} |\widehat{\nabla_{\bs{r}}\phi}|^2 + \frac{1}{2} \wh{w}^2
 + g\zeta \right) \,D\,d^2r
\,.
\end{split}
	\label{WW-LP-Erg-CWW-m}
\end{align}

This change of variables leads to the following Lie-Poisson Hamiltonian formulation,
\begin{align}
\begin{split}
\frac{\p}{\p t}
\begin{bmatrix}
m_i \\ D \\ \lambda \\ \zeta 
\end{bmatrix}
= -
\begin{bmatrix}
\partial_j m_i + m_j \partial_i & D\partial_i  & \lambda\partial_i & -\zeta_{,i}
\\
\partial_jD 	    & 0  & 0 & 0
\\
\partial_j\lambda	& 0 & 0  & 1
\\
\zeta_{,j}	 & 0  & -1 & 0
\end{bmatrix}
\begin{bmatrix}
 {\delta h}/{\delta {m_j}} = \wh{v}^j
 \\ {\delta h}/{\delta D} = -\,\varpi
 \\ 
{\delta h}/{\delta \lambda} =  \lambda/D 
\\ {\delta h}/{\delta \zeta} =   g D 
\end{bmatrix}
\,,
\end{split}
	\label{FS-diag-brkt-CWW-1}
\end{align}
where $\varpi$ is defined in equation \eqref{Bernoulli-varpi-1}.
Here, the Poisson operator is the direct sum of the usual semidirect-product Lie-Poisson bracket for ideal fluids \cite{HMR1998} and 
a symplectic Poisson bracket for the canonical wave variables, $(\lambda,\zeta)$. 

The Poisson operator in \eqref{FS-diag-brkt-CWW-1} is said to \emph{entangle}
the dynamics of the wave variables $(\lambda,\zeta)$ with the combined variables $(m_i,D)$. Next, we will \emph{untangle} the entangled Poisson operator to put it back into 
block-diagonal form as before in \eqref{FS-diag-brkt-CWW-canon} by considering only the momentum density corresponding to the potential part of the fluid flow. 

\subsubsection{Untangled Hamiltonian formulation of the ECWWE }

The momentum density of only the purely potential part of the fluid flow is given in terms of the canonical wave variables and the transport velocity $\bs{\wh{v}}:=\wh{\nabla_{\bs{r}}\phi}$ by
\begin{align}
\bs{M} = D\bs{\wh{v}} + \lambda \nabla_{\bs{r}} \zeta = D\nabla\wh{\phi} 
= D \bs{V}
\,.
	\label{WW-total-momap-CWW}
\end{align}
The Legendre transform with respect to only the $D$ and $\wh{\phi}$ variables corresponding to the potential part of the 
flow leads to the following Hamiltonian,
\begin{align}
\begin{split}
h(\bs{M},D,\lambda,\zeta) 
&= \int \bs{M}\cdot \bs{\wh{v}} + \lambda\p_t \zeta\, d^2r - \ell(\bs{\wh{v}}, D,\widehat{\phi},{\wh w}, \zeta; \lambda)
\\&=
\int \frac{1}{2D} \big|\bs{M} - \lambda \nabla_{\bs{r}} \zeta \big|^2
+ \frac{\lambda^2}{2D} + gD\zeta \,d^2r
\\&= E(\bs{\wh{v}},\wh{w},\zeta , D)
\,.
\end{split}
	\label{WW-LP-Ham-CWW}
\end{align}
where the energy $E(\bs{\wh{v}},\wh{w},\zeta , D)$ is defined in \eqref{WW-LP-Erg-CWW-m}. 
Thus, the Hamiltonian in \eqref{WW-LP-Ham-CWW} is yet another representation of the conserved energy for 
the ECWWE system in \eqref{LannesZetaEquationX} and \eqref{LannesPsiEquationX}.

\paragraph{Variations of the Hamiltonian in \eqref{WW-LP-Ham-CWW}.}
 In the Hamiltonian variables, the Bernoulli function ${\varpi}$ in \eqref{var-eqns-FS-CWW} is denoted as
\begin{equation}
	\varpi
	:= \frac{1}{2}\Big(\big|\bs{\wh{v}}\big|^2 + \wh{w}^2\Big) - g\zeta
	= \frac{1}{2D^2} \big|\bs{M} 
	- \lambda \nabla_{\bs{r}} \zeta \big|^2
    + \frac{\lambda^2}{2D^2}  
	- g\zeta
	\,.
	\label{Bernoulli-varpi-M}
\end{equation}
After evaluating the corresponding variational derivatives of the Hamiltonian in \eqref{WW-LP-Ham-CWW}, 
the system of equations in \eqref{FS-diag-brkt-CWW-1} may be written in block-diagonal form, as
\begin{align}
\begin{split}
\frac{\p}{\p t}
\begin{bmatrix}
M_i \\ D \\ \lambda \\ \zeta 
\end{bmatrix}
= -
\begin{bmatrix}
\partial_j M_i + M_j \partial_i & D\partial_i  & 0 & 0
\\
\partial_jD 	    & 0  & 0 & 0
\\
0	& 0 & 0  & 1
\\
0	 & 0  & -1 & 0
\end{bmatrix}
\begin{bmatrix}
 {\delta h}/{\delta {M_j}} = \wh{v}^j
 \\ {\delta h}/{\delta D} = -\,\varpi
 \\ 
{\delta h}/{\delta \lambda} = - \bs{\wh{v}}\cdot\nabla_{\bs{r}}\zeta + \lambda/D 
\\ {\delta h}/{\delta \zeta} =  {\rm div}_{\bs{r}} (\lambda\bs{\wh{v}}) + g D 
\end{bmatrix}
\,.
\end{split}
	\label{FS-diag-brkt-CWW}
\end{align}
This \emph{untangled} form of the Poisson operator comprises a direct product of the standard Lie-Poisson bracket for fluid variables $(\bs{M},D)$ and 
a symplectic Poisson bracket for the canonical wave variables, $(\lambda,\zeta)$. 

\begin{remark}[Physical meaning of the model]
The dual entangled and untangled forms of the Lie-Poisson brackets seen in \eqref{FS-diag-brkt-CWW-1} and \eqref{FS-diag-brkt-CWW} are familiar in Hamiltonian formulations of wave-current interactions and other compound Eulerian-Lagrangian fluid systems, as well as body-space mechanical systems. These dual formulations are particularly well-known in the investigations of systems whose dynamics is governed by variational principles which are averaged over time, phase, or some other fluctuating or stochastic parameter. See, e.g., \cite{HolmSWCI2021} for a recent review and bibliography relevant to the current investigation. 
\end{remark}

{
\paragraph{The Non-acceleration Theorem for ECWWE.}
The Lie-Poisson Hamiltonian structure in \eqref{FS-diag-brkt-CWW} provides insight into the physical interactions occurring in the ECWWE. Namely, the Eulerian fluid variables are Lie-Poisson in the total momentum in \eqref{WW-total-momap-CWW} and the area element $D$, while the canonically conjugate wave variables undergo symplectic dynamics in the elevation $\zeta$ and its canonical momentum density  $\lambda=D\wh{w}$. Since the Poisson structure block-diagonalises for the two types of fields, both fields are seen to contribute on the same footing to the Hamiltonian formulation of the combined motion.

This dual wave-current contribution is already clear from the coordinate-free form of the motion equation in \eqref{EP-WW-motion-eqn}, since it immediately implies conservation of a two-component Kelvin circulation integral involving both types of fields present in the momentum density, 
\begin{align}
\begin{split}
    \frac{d}{dt} \oint_{c(\bs{\wh{v}})} (\bs{\wh{v}}\cdot d\bs{r} + \widehat{w}d\zeta)
    &= \oint_{c(\bs{\wh{v}})}(\p_t +\mathcal{L}_{\bs{\wh{v}}})
    \big(\bs{\wh{v}}\cdot d\bs{r} + \widehat{w}d\zeta\big)
    \\
    &= \oint_{c(\bs{\wh{v}})} (\p_t +\mathcal{L}_{\bs{\wh{v}}})(\bs{V}\cdot d\bs{r})
    = \oint_{c(\bs{\wh{v}})} (\p_t +\mathcal{L}_{\bs{\wh{v}}})d\widehat{\phi}
    = \oint_{c(\bs{\wh{v}})} d\varpi = 0\,,
\end{split}
	\label{EP-WW-circ-thm}
\end{align}
where $c(\bs{\wh{v}})$ denotes a closed material loop moving with the fluid transport velocity, $\bs{\wh{v}}$. 
However, a closer look at the separate equations of motion for $\lambda$ and $\zeta$ in the Hamiltonian form in equation \eqref{FS-diag-brkt-CWW} shows that the wave dynamics takes place independently in the moving frame the fluid flow, without actually influencing the flow. Indeed, a closer look at the circulation dynamics in \eqref{EP-WW-circ-thm} verifies that the two components of the circulation are \emph{conserved separately}, as we already know from equation \eqref{eq: 2 KelThms}. In particular, this means that in solutions of ECCWE the waves cannot generate circulation of the currents. This is known in the literature as the ``Non-acceleration Theorem'' \cite{AM1978}.

\paragraph{Modifying the ECWWE to rectify the Non-acceleration Theorem.}
The Non-acceleration Theorem arises in time-dependent solutions of ECCWE because the model has no dependence on the combination of wave slope $\nabla \zeta$ and vertical velocity $\wh{w}$ which would characterise a wave-slope dependence of the energy of the wave field. 
One natural proposal for rectifying this situation would be to introduce an additional energy density as in the last term in the following,
\begin{align}
\begin{split}
E(\bs{\wh{v}},\wh{w},\zeta , D) 
&=
\int \left(\frac{1}{2} |\widehat{\nabla_{\bs{r}}\phi}|^2 + \frac{1}{2} \wh{w}^2
 + g\zeta 
 + \frac{\epsilon}{2} \big| \nabla_{\bs{r}}\wh{\phi}-\wh{\nabla_{\bs{r}}\phi} \big|^2 \right) \,D\,d^2r
\\&=
\int \left(\frac{1}{2} |\bs{\wh{v}}|^2 + \frac{1}{2} \wh{w}^2
 + g\zeta + \frac{\epsilon}{2} \big| \bs{V}-\bs{\wh{v}}\big|^2 \right) \,D\,d^2r 
\\&=
\int \left(\frac{1}{2} |\bs{\wh{v}}|^2 + \frac{1}{2} \wh{w}^2
 + g\zeta + \frac{\epsilon}{2} | \wh{w}\nabla\zeta |^2 \right) \,D\,d^2r 
\,.
\end{split}
	\label{WW-LP-Erg-CWW-m}
\end{align}
Physically, this would mean that the wave-field energy density $\epsilon D | \wh{w}\nabla\zeta |^2/2$ would increase with mass density $D$, wave slope $\nabla \zeta$, and vertical velocity $\wh{w}$. The multiplier $\epsilon$ could be varied $0\le\epsilon\le1$ to test the sensitivity of the model solutions to the proposed wave-field energy cost.

In nondimensional form, the wave-field energy cost could be included in the action integral as 
\begin{align}
\begin{split}
	S &=\int \ell(\bs{\wh{v}}, D,\widehat{\phi},{\wh w}; \zeta, \lambda)dt
	\\&= \int \int  
	D\left(\frac{1}{2}\big(|\bs{\wh{v}}|^2+\sigma^2\widehat{w}^2
	+ \epsilon\sigma^4 |\wh{w}\nabla_{\bs{r}}\zeta|^2 \big)  - \frac{\zeta}{Fr^2}\right)   
	+ \sigma^2\lambda\Big( \partial_t\zeta  + \bs{\wh{v}}\cdot\nabla_{\bs{r}}\zeta - \widehat{w} \Big)
	\\
	&\qquad 
	+ \widehat{\phi}\big(\partial_t D + \text{div}_{\bs{r}}(D\bs{\wh{v}})\big)
	\ d^2r\,dt\,,
\end{split}
	\label{ActionIntegral-WaveErgMod}
\end{align}
Stationarity of the modified action integral in \eqref{ActionIntegral-WaveErgMod} yields the variational equations,
\begin{align}
\begin{split}
	\delta \bs{\wh{v}}:&\quad 
	D \bs{\wh{v}}\cdot d\bs{r} + \sigma^2\lambda\,d\zeta 
	=  Dd\widehat{\phi} 
	\quad\Longrightarrow\quad 
	\bs{V}\cdot d\bs{r} :=
	\bs{\wh{v}}\cdot d\bs{r} + \sigma^2(\lambda/D)\,d\zeta 
	=  d\widehat{\phi}
	\,,\\ 
	\delta \widehat{w} :&\quad 
	D\widehat{w} - \lambda + \epsilon\sigma^2 D\wh{w}|\nabla_{\bs{r}}\zeta|^2 = 0
	\quad\Longrightarrow\quad 
	(\lambda/D) = \wh{w} \big(1+\epsilon\sigma^2|\nabla_{\bs{r}}\zeta|^2\big)
	\,,\\
	\delta \lambda:&\quad 
	 \partial_t\zeta  + \bs{\wh{v}}\cdot\nabla_{\bs{r}}\zeta
	 -  \widehat{w} = 0
	\,,\\
	\delta \zeta:&\quad 
	\partial_t \lambda + \text{div}_{\bs{r}}
	\big(\lambda\bs{\wh{v}} \big)
	=
	 - \,\frac{D}{\sigma^2Fr^2} 
	 - \text{div}_{\bs{r}}\big( \epsilon \sigma^2 D\wh{w}^2  \nabla_{\bs{r}}\zeta\big)
	\,,\\
	\delta \widehat{\phi} :& \quad 
	\partial_t D + \text{div}_{\bs{r}}\big(D\bs{\wh{v}}\big) =0 
	\,,\\
	\delta D:&\quad 
	\big(\partial_t + \bs{\wh{v}}\cdot \nabla_{\bs{r}}\big)\widehat{\phi} 
	= \frac{1}{2}\Big(|\bs{\wh{v}}|^2+\sigma^2\widehat{w}^2
	\big(1+\epsilon\sigma^2|\nabla_{\bs{r}}\zeta|^2\big)  \Big)
	- \frac{\zeta}{Fr^2} =:\widetilde{\varpi}
	\,.
\end{split}
	\label{var-eqns-WaveErg}
\end{align}
The previous wave system in \eqref{var-eqns-FS-CWW} is recovered when one sets $\epsilon=0$. 

\begin{theorem}\label{KelThmModCWWE}
The system of equations in \eqref{var-eqns-WaveErg} implies the following Kelvin theorem for the fluid circulation
\begin{equation}\label{ModCWW-WCIKelvin}
 \frac{d}{dt} \oint_{c(\bs{\wh{v}})} \bs{\wh v}\cdot d\bs{r}  
 = \epsilon\sigma^4\oint_{c(\bs{\wh{v}})}
 	  \frac{1}{D}\text{div}_{\bs{r}}\big( D\wh{w}^2\nabla\zeta\big)   d\zeta
	 +  |\nabla_{\bs{r}}\zeta|^2 d\wh{w}^2/2
\,,
\end{equation}
in which $c(\bs{\wh{v}})$ is a closed loop moving with the material velocity $\bs{\wh{v}}$.
\end{theorem}
\begin{proof}
The proof follows by first integrating the $\delta \bs{v}$ equation \eqref{var-eqns-WaveErg} around a material loop $c(\bs{\wh{v}})$ moving with velocity $\bs{\wh{v}}$ to find, 
\[
\frac{d}{dt} \oint_{c(\bs{\wh{v}})} 
\big(\bs{\wh{v}} + \sigma^2\lambda\,\nabla\zeta \big)\cdot d\bs{r}
=
\oint_{c(\bs{\wh{v}})}(\partial_t+\mathcal{L}_{\bs{\wh v}})
	\Big(\big(\bs{\wh{v}} + \sigma^2\lambda\,\nabla\zeta \big)\cdot d\bs{r}\Big)
	=  \oint_{c(\bs{\wh{v}})}d\widehat{\phi} = 0
\,,\]
upon using the well-known identity
\begin{align*}
 \frac{d}{dt} \oint_{c(\bs{\wh{v}})} \bs{\wh v}\cdot d\bs{r} 
 &=
 \oint_{c(\bs{\wh{v}})}(\partial_t+\mathcal{L}_{\bs{\wh v}}) \big(\bs{\wh v}\cdot d\bs{r} \big)
\,.
\end{align*}
One then computes 
\begin{align*}
\frac{d}{dt} \oint_{c(\bs{\wh{v}})} \bs{\wh{v}}\cdot d\bs{r}
&= - \sigma^2 \oint_{c(\bs{\wh{v}})}(\partial_t+\mathcal{L}_{\bs{\wh v}}) 
\big((\lambda/D)\,d\zeta \big)
 \\&= \oint_{c(\bs{\wh{v}})}
 	 \Big(\frac{1}{Fr^2} 
	 + \frac{1}{D}\text{div}_{\bs{r}}\big( \epsilon \sigma^4 D\wh{w}^2\nabla\zeta\big)\Big)  d\zeta
	 + \sigma^2 (\lambda/D) d\wh{w}
 \\&= \epsilon\sigma^4\oint_{c(\bs{\wh{v}})}
 	  \frac{1}{D}\text{div}_{\bs{r}}\big( D\wh{w}^2\nabla\zeta\big)   d\zeta
	 +  |\nabla_{\bs{r}}\zeta|^2 d\wh{w}^2/2
\end{align*}
\end{proof}

\begin{corollary}\label{cor: wave gen circ}
The modified CWW model of wave dynamics arising via Hamilton's principle from the action integral in \eqref{ActionIntegral-WaveErgMod} with its additional wave energy creates circulation in the fluid whenever the gradients of the wave variables are not aligned.  
\end{corollary}

\begin{proof}
The proof follows by applying the Stokes theorem to the right-hand side of the material loop $c(\bs{\wh{v}})$ in equation \eqref{ModCWW-WCIKelvin} in Theorem \ref{KelThmModCWWE}.
\end{proof}

{
\begin{remark}[Wave propagation velocity.]\label{remark:Wave_Propagation}
    Whilst the modification made to the energy here does rectify the Non-acceleration Theorem so that the wave evolution can affect the circulation of the current, the equation in \eqref{var-eqns-WaveErg} corresponding to the variation in $\lambda$ indicates that the surface elevation remains a Lagrangian coordinate. This feature does not allow for waves of phase which do not carry mass. 
    
    In Section \ref{sec: ACWWE}, we will introduce a different coupling which will allows for waves which propagate at a speed different to the Lagrangian parcels on the fluid surface. For this purpose we will make use of the two distinct 2D velocities on the free surface which appear in the energy equation in \eqref{WW-LP-Erg-CWW-m}, $\boldsymbol{\widehat{v}}=\widehat{\,\nabla\phi\,}$ and $\boldsymbol{V}=\nabla\widehat\phi$. The first of them, $\boldsymbol{\widehat{v}}$, is the transport velocity of Lagrangian parcels on the fluid surface. The second of them, $\boldsymbol{V}$, is the phase velocity of a level set of the velocity potential $\widehat\phi$ evaluated on the free surface. 
    
    The coupling we will introduce in the next section will include a homotopy coefficient $0\le\epsilon\le1$ which will provide the option to set the Eulerian transport velocity of the wave elevation to a value anywhere between $\boldsymbol{\widehat{v}}=\widehat{\,\nabla\phi\,}$ for $\epsilon=0$ and $\boldsymbol{V}=\nabla\widehat\phi$ for $\epsilon=1$. For $0<\epsilon\le1$ the wave propagation velocity will no longer be equal to the material velocity. 
\end{remark}
}

\section{Augmented Classical Water-Wave equations (ACWWE)}\label{sec: ACWWE}

%%%%%%%%%%%%%%%%%%%%%%%%%%%%%%%%%%%%%%%%%%%%%%%%%%%%%%%%%%%%%%%%%%%%%%%%%%%%%%
\subsection{Variational derivation of the ACWWE}\label{subsec:CWWE_with_Coupling}

Let us propose a less severe modification of the action integral in equation \eqref{ActionIntegral-FS} than the energy modification introduced in \eqref{ActionIntegral-WaveErgMod}. This proposal will not introduce any change in the wave energy. Instead, it will allow a slip in the phase of the wave velocity relative to Lagrangian mass transport velocity which will be imposed by the following constraint in the Lagrangian,
\begin{align}
\begin{split}
	S &=\int \ell(\bs{\wh{v}}, D,\widehat{\phi},{\wh w}; \zeta, \lambda)dt
	\\&= \int \int  
	D\left(\frac{1}{2}\big(|\bs{\wh{v}}|^2+\sigma^2\widehat{w}^2\big)  - \frac{\zeta}{Fr^2}\right)   
	+ \sigma^2\lambda\Big( \partial_t\zeta  
	+ \big(\bs{\wh{v}} - \epsilon\sigma^2 \wh{w}\nabla\zeta \big)\cdot\nabla_{\bs{r}}\zeta 
	- \widehat{w} \Big)
	\\
	&\qquad 
	+ \widehat{\phi}\big(\partial_t D + \text{div}_{\bs{r}}(D\bs{\wh{v}})\big)
	\ d^2r\,dt\,.
\end{split}
	\label{ActionIntegral-FSmod}
\end{align}
Here, we have retained the same non-dimensional parameters as in Remark \ref{non-dim-scales} and the non-dimensional constant parameter $0\le\epsilon\le1$ is the homotopy coefficient mentioned in Remark \ref{remark:Wave_Propagation}. According to Remark \ref{non-dim-scales} each wave variable is multiplied by the aspect ratio $\sigma$ relative the fluid variables.

In equation \eqref{ActionIntegral-FSmod} we have introduced a term which is intended to model \emph{wave-current minimal coupling} (WCMC) \cite{Frenkel1934}. The WCMC term modifies the transport velocity of the wave momentum density, to allow \emph{genuine} wave-current interaction, so the free-surface water waves will no longer be passively advected by the fluid velocity. However, it leaves the wave energy unchanged. 
}

Stationarity of the augmented action integral in \eqref{ActionIntegral-FSmod} now yields the variational equations,
\begin{align}
\begin{split}
	\delta \bs{\wh{v}}:&\quad 
	D \bs{\wh{v}}\cdot d\bs{r} + \sigma^2\lambda\,d\zeta 
	=  Dd\widehat{\phi} 
	\quad\Longrightarrow\quad 
	\bs{V}\cdot d\bs{r} :=
	\bs{\wh{v}}\cdot d\bs{r} + \sigma^2\widetilde{w}\,d\zeta 
	=  d\widehat{\phi}
	\,,\\ 
	\delta \widehat{w} :&\quad 
	D\widehat{w} - \lambda\big(1+\epsilon\sigma^2|\nabla_{\bs{r}}\zeta|^2\big) = 0
	\quad\Longrightarrow\quad 
	\lambda = \frac{D\widehat{w}}{1+\epsilon\sigma^2|\nabla_{\bs{r}}\zeta|^2} =: D\widetilde{w}
	\,,\\
	\delta \lambda:&\quad 
	 \partial_t\zeta  + \big(\bs{\wh{v}} 
	 -  \epsilon \sigma^2 \wh{w}  \nabla_{\bs{r}}\zeta\big) \cdot\nabla_{\bs{r}}\zeta
	 -  \widehat{w} = 0
	\,,\\
	\delta \zeta:&\quad 
	\partial_t \lambda + \text{div}_{\bs{r}}
	\Big(\lambda\big(\bs{\wh{v}} -  \epsilon \sigma^2 \wh{w}  \nabla_{\bs{r}}\zeta\big)\Big)
	=
	 - \,\frac{D}{\sigma^2Fr^2} 
	 + \text{div}_{\bs{r}}\big( \epsilon \sigma^2 \lambda\wh{w}  \nabla_{\bs{r}}\zeta\big)
	\,,\\
	\delta \widehat{\phi} :& \quad 
	\partial_t D + \text{div}_{\bs{r}}\big(D\bs{\wh{v}}\big) =0 
	\,,\\
	\delta D:&\quad 
	\big(\partial_t + \bs{\wh{v}}\cdot \nabla_{\bs{r}}\big)\widehat{\phi} 
	= \frac{1}{2}\big(|\bs{\wh{v}}|^2+\sigma^2\widehat{w}^2\big)  
	- \frac{\zeta}{Fr^2} =:\varpi 
	\,.
\end{split}
	\label{var-eqns-FS}
\end{align}
Notice that the $\epsilon$ coupling term in the modified Lagrangian in \eqref{ActionIntegral-FSmod} does not affect variations in the fluid variables, $(\bs{\wh{v}},D,\wh{\phi})$. However, it changes the previous relationship between $\lambda$ and $\wh{w}$ by a term proportional to $\epsilon$, which means it yields a different definition of $\bs{V}$, as seen in the first and second lines of \eqref{var-eqns-FS}. Of course, the previous wave system in \eqref{var-eqns-FS-CWW} is recovered when one sets $\epsilon\to0$.

\paragraph{ACWW motion equation.}
We may write out the ACWW motion equation obtained by substituting the variational results into the application of the advective time derivative $(\partial_t+\mathcal{L}_{\bs{\wh v}})$ on the first line of the system \eqref{var-eqns-FS}, to find in the notation  $\lambda= D\widetilde{w}$ that
\begin{align}
	D(\partial_t+\mathcal{L}_{\bs{\wh v}})\big( \bs{\wh v}\cdot d\bs{x} + \sigma^2\widetilde{w}\,d\zeta \big) 
	= Dd(\partial_t+\mathcal{L}_{\bs{\wh v}})\widehat{\phi} 
	= Dd{\varpi} 
	\,.
\label{EulerPoincare-3Y}
\end{align}
Recall from \eqref{var-eqns-FS} that Bernoulli function ${\varpi}$ and vertical wave momentum density $\lambda$ are defined as 
\begin{align}
{\varpi} :=  \frac12\big(|\bs{\wh v}|^2+{\sigma^2\widehat w}^2\big) -  \frac{\zeta}{Fr^2} 
\,,\qquad
\lambda = \frac{ D {\widehat w}}{1+\epsilon\sigma^2|\nabla_{\bs{r}}\zeta|^2} =: D\widetilde{w}
\,.
	\label{Bernoulli-wavemom-defY}
\end{align}

At this point, let us collect the ACWWE  in terms of (i) fluid variables, comprising velocity $\bs{\wh v}$ and area density $D$, and (ii) wave variables, comprising surface elevation $\zeta$, and vertical wave momentum density $\lambda$. Namely, the ACWWE  are given by,
\begin{align}
\begin{split}
(\partial_t+\mathcal{L}_{\bs{\wh v}})\big( \bs{\wh v}\cdot d\bs{x} + \sigma^2\widetilde{w}\,d\zeta \big)
&= d{\varpi} 
\,,
\\
\partial_t D + \text{div}_{\bs{r}}(D\bs{\wh v}) &=0
\,,
\\
\partial_t\zeta + {\bs{\wh v}}\cdot\nabla_{\bs{r}}\zeta
&=
{\widehat w}(1+\epsilon\sigma^2|\nabla_{\bs{r}}\zeta|^2) \,,
\\
\partial_t \widetilde{w} + {\bs{\wh v}}\cdot\nabla_{\bs{r}} \widetilde{w} 
	&=
	 - \,\frac{1}{\sigma^2Fr^2}  + 
	\frac{2\epsilon\sigma^2}{D}\,{\rm div}  \big( D \widetilde{w}\, \widehat{w} \, \nabla_{\bs{r}} \zeta\big)
	\,.
\end{split}
\label{EulerPoincare-ACWW}
\end{align}
The equation set \eqref{EulerPoincare-ACWW} recovers the ECWWE equations \eqref{WW-eqns1} when $\epsilon\to0$.
The first of these equations implies Kelvin's circulation theorem for the ACWW model, as follows.
\begin{theorem}[Kelvin-Noether theorem for the ACWW model]\label{KNthm-CFS}
For every closed loop $c(\bs{\wh v})$ moving with the ACWW transport velocity $\bs{\wh v}$ for the system of ACWWE  in \eqref{EulerPoincare-ACWW} the Kelvin circulation relation holds. Namely,
\begin{align}
\frac{d}{dt} \oint_{c(\bs{\wh v})} \Big(\bs{\wh v}\cdot d\bs{x} +  \sigma^2\widetilde{w}   d \zeta\Big) 
=
\oint_{c(\bs{\wh v})} d{\varpi} = 0
\,.
	\label{Kel-circACWW-thm}
\end{align}
\end{theorem}

\begin{proof}
From the first ACWWE in \eqref{EulerPoincare-ACWW}, we have
\begin{align}
(\partial_t+\mathcal{L}_{\bs{\wh v}})\Big(\bs{\wh v}\cdot d\bs{x} +  \sigma^2\widetilde{w}  d \zeta\Big) 
=
d{\varpi} 
\,,
\label{Kel-circACWW-proof}
\end{align}
and the result \eqref{Kel-circACWW-thm} follows from the standard relation for the time derivative of an integral around a closed moving loop, $c(\bs{\wh v})$,
\[
\frac{d}{dt} \oint_{c(\bs{\wh v})} \Big(\bs{\wh v}\cdot d\bs{r} +  \sigma^2\widetilde{w}   d \zeta\Big) 
=
\oint_{c(\bs{\wh v})} 
(\partial_t+\mathcal{L}_{\bs{\wh v}})\Big(\bs{\wh v}\cdot d\bs{r} +  \sigma^2\widetilde{w}  d \zeta \Big) 
=
\oint_{c(\bs{\wh v})} d\varpi = 0
\,.\]
\end{proof}
\begin{remark}[Transport of wave dynamics relative to the fluid velocity]
A slight rearrangement of the last two equations in \eqref{EulerPoincare-ACWW} demonstrates that the wave dynamics is no longer transported passively by the fluid velocity $\bs{\wh v}$. Instead, a shifted transport velocity appears; namely, 
\begin{align}
\bs{\wh v} - \epsilon\sigma^2\wh{w}\nabla_{\bs{r}}\zeta=\bs{\wh v} - \epsilon\sigma^2\bs{s}
\,.
\label{Shifted-waveveloc}
\end{align}
This shift in wave velocity introduces wave dynamics into the transport velocity of the wave variables, as follows, 
\begin{align}
\begin{split}
\partial_t\zeta
+  \big(\bs{\wh v} - \epsilon\sigma^2\wh{w}\nabla_{\bs{r}}\zeta \big)\cdot\nabla_{\bs{r}}\zeta
&= 
\wh{w} \,,
\\
\partial_t \lambda + {\rm div}_{\bs{r}}  \Big( \lambda \big(\bs{\wh v} - \epsilon\sigma^2\wh{w}\nabla_{\bs{r}}\zeta \big)\Big)
	&=
	 - \,\frac{1}{\sigma^2Fr^2}  + 
	\epsilon\sigma^2{\rm div}_{\bs{r}}  \big( \lambda \widehat{w} \, \nabla_{\bs{r}} \zeta\big)
	\,,
\end{split}
\label{EulerPoincare-ACWW-waves}
\end{align}
where one recalls that the canonical momentum density $\lambda$ conjugate to the elevation $\zeta$ is defined in terms of the other wave variables in \eqref{Bernoulli-wavemom-defY}. 
\end{remark}

\begin{remark}
The difference in the wave momentum transport velocity relative to the fluid velocity in equation \eqref{Shifted-waveveloc} will turn out to produce an important effect by which the waves will generate fluid circulation. 
To compute this effect on the circulation of the fluid we subtract the fluid transport of the wave momentum from the total momentum transport by the fluid in \eqref{Kel-circACWW-proof}. The fluid velocity transport of the wave momentum is found from the wave dynamical equations in \eqref{EulerPoincare-ACWW-waves}, as

\begin{align}
\begin{split}
(\partial_t+\mathcal{L}_{\bs{\wh v}})\Big(\sigma^2\widetilde{w}  d \zeta\Big) 
&=
\sigma^2\left(
- \,\frac{1}{\sigma^2Fr^2} + \frac{2\epsilon \sigma^2}{D} {\rm div} \big(D \widetilde{w}(\wh{w} \nabla\zeta) \big)
\right)d\zeta
+ \frac{\sigma^2\widehat w}{1+\epsilon\sigma^2|\nabla_{\bs{r}}\zeta|^2}d\big(\widehat{w}(1+\epsilon\sigma^2|\nabla_{\bs{r}}\zeta|^2)\big)
\\&=
\sigma^2\left(
- \,\frac{1}{\sigma^2Fr^2} + \frac{2\epsilon \sigma^2}{D} {\rm div} \big(D \widetilde{w}(\wh{w} \nabla\zeta) \big)
\right)d\zeta
\\&\hspace{2cm} + \frac12\sigma^2d\wh{w}^2 + \sigma^2\wh{w}^2d\big(\log((1+\epsilon\sigma^2|\nabla_{\bs{r}}\zeta|^2)\big)
\,.
\end{split}
\label{ACWW-wavemomtrans}
\end{align}
Upon subtracting equation \eqref{ACWW-wavemomtrans} from the total fluid momentum transport equation in \eqref{Kel-circACWW-proof} a short calculation yields
\begin{equation}\label{ACWW-fluidmomtrans1}
    (\partial_t+\mathcal{L}_{\bs{\wh v}}) \big(\bs{\wh v}\cdot d\bs{r} \big) = \frac12 d |\bs{\wh v}|^2 
- \frac{\sigma^2}{2}\widetilde{w}^2 d \big(1+\epsilon\sigma^2|\nabla_{\bs{r}}\zeta|^2\big)^2
- \frac{2\epsilon \sigma^4}{D}{\rm div}\big(
D\widetilde{w}(\wh{w}\nabla_{\bs{r}}\zeta)
\big)\,d\zeta
\,.
\end{equation}
\end{remark}

\begin{theorem}\label{KelThmACWWE}
The corresponding Kelvin theorem for equation \eqref{ACWW-fluidmomtrans1} is given by
\begin{equation}\label{ACWW-WCIKelvin}
 \frac{d}{dt} \oint_{c(\bs{\wh{v}})} \bs{\wh v}\cdot d\bs{r}  = \oint_{c(\bs{\wh{v}})} \frac12 d |\bs{\wh v}|^2 
- \frac{\sigma^2}{2}\widetilde{w}^2 d \big(1+\epsilon\sigma^2|\nabla_{\bs{r}}\zeta|^2\big)^2
- \frac{2\epsilon \sigma^4}{D}{\rm div}\big(
D\widetilde{w}(\wh{w}\nabla_{\bs{r}}\zeta)
\big)\,d\zeta
\,.
\end{equation}
\end{theorem}
\begin{proof}
The proof follows by integrating equation \eqref{ACWW-fluidmomtrans1} around a material loop $c(\bs{\wh{v}})$ moving with velocity $\bs{\wh{v}}$, then using the well-known identity
\[
 \frac{d}{dt} \oint_{c(\bs{\wh{v}})} \bs{\wh v}\cdot d\bs{r} 
 =
 \oint_{c(\bs{\wh{v}})}(\partial_t+\mathcal{L}_{\bs{\wh v}}) \big(\bs{\wh v}\cdot d\bs{r} \big)
\,.\]
\end{proof}
\begin{corollary}\label{cor: wave gen circ}
The ACWW model of wave dynamics creates circulation in the fluid whenever the gradients of the wave variables are not aligned.  
\end{corollary}
\begin{proof}
The proof follows by applying the Stokes theorem to the right-hand side of the material loop $c(\bs{\wh{v}})$ in equation \eqref{ACWW-WCIKelvin} in Theorem \ref{KelThmACWWE}.
\end{proof}

\begin{theorem}[Total energy conservation is independent of $\epsilon$]\label{Erg-conserv}$\,$\\
The energy conserved by modified equations \eqref{ACWW-wavemomtrans} and \eqref{ACWW-fluidmomtrans1} is independent of $\epsilon$.
\end{theorem}
\begin{proof}
The modified constrained action integral in equation \eqref{ActionIntegral-FSmod} may be rewritten equivalently as a phase-space Lagrangian, upon rearranging as follows,
\begin{align}
\begin{split}
	S &=\int \ell(\bs{\wh{v}}, D,\widehat{\phi},{\wh w}; \zeta, \lambda)dt
	\\&= \int \int  
	D\left(\frac{1}{2}\big(|\bs{\wh{v}}|^2+\sigma^2\widehat{w}^2\big)  - \frac{\zeta}{Fr^2}\right)   
	+ \sigma^2\lambda\Big( \partial_t\zeta  + \bs{\wh{v}}\cdot\nabla_{\bs{r}}\zeta - \widehat{w} \Big)
	\\
	&\qquad 
	+ \widehat{\phi}\big(\partial_t D + \text{div}_{\bs{r}}(D\bs{\wh{v}})\big)
	- \epsilon\sigma^4\wh{w} \lambda |\nabla_{\bs{r}}\zeta|^2
	\ d^2r\,dt\,.
	\\&= \int \int  
	\widehat{\phi}\partial_t D + \sigma^2\lambda \partial_t\zeta 
	- D \left(\frac{1}{2} |\bs{\wh{v}}|^2 + \frac{1}{2} \wh{w}^2 + \frac{\zeta}{Fr^2} \right)
		\ d^2r\,dt\,.
\end{split}
	\label{ActionIntegral-FSmod-PS}
\end{align}
The last term in this equation agrees with the definition of energy for the unmodified Lagrangian in equation \eqref{WW-LP-Erg-CWW-m} obtained by setting $\epsilon=0$ in the modified Lagrangian in equation \eqref{ActionIntegral-FSmod}. Thus, the modified and unmodified system conserve the same physical energy. 
\end{proof}

Theorem \ref{Erg-conserv} shows that the equations resulting from the modified Lagrangian in \eqref{ActionIntegral-FSmod} conserve the same energy as for the unmodified Lagrangian in \eqref{ActionIntegral-FS}. Thus, the modification in \eqref{ActionIntegral-FSmod} which was obtained by introducing the $\epsilon$ term produces the wave-current interaction in equations \eqref{ACWW-wavemomtrans} and \eqref{ACWW-fluidmomtrans1} while also preserving the original physical energy density. 
What depends on $\epsilon$ is the definition of the vertical wave momentum density $\lambda$ canonically conjugate to the elevation $\zeta$ depends on $\epsilon$, as well as the definition of the velocity $\bs{V}$ in terms of the wave variables, as seen in the first and second lines of \eqref{var-eqns-FS}. 

\begin{remark}[Tensor virial theorem for a Lagrangian fluid blob under the ACWWE]
Although the conserved energy remains the same for the ECWW and ACWW models for any value of the coupling constant $\epsilon$, the tensor virial theorem for a Lagrangian fluid blob under the ACWWE is considerably more intricate than in section \ref{Blob-dyn-ECWW} for the ECWW model.
\end{remark}

%%%%%%%%%%%%%%%%%%%%%%%%%%%%%%%%%%%%%%%%%%%%%%%%%%%%%%%%%%%%%%%%%%%%%%%%%%%%%%
\subsection{Lie-Poisson Hamiltonian formulation of the ACWWE }

As discussed in appendix \ref{app:LegXform}, the Legendre transformation of the augmented Lagrangian in the action integral \eqref{ActionIntegral-FSmod} with respect to the sum of the fluid and wave momentum densities 
\begin{align}
\bs{M} = D\bs{\wh{v}} + \lambda \nabla_{\bs{r}} \zeta 
= D \bs{V}
	\label{WW-total-momap}
\end{align}
leads to the ECWW Hamiltonian defined now in dimensional units by 
\begin{align}
\begin{split}
h(\bs{M},D,\lambda,\zeta) 
&=
\int \frac{1}{2D} \big|\bs{M} - \lambda \nabla_{\bs{r}} \zeta \big|^2
+ \frac{\lambda^2}{2D}(1+\epsilon|\nabla_{\bs{r}}\zeta|^2)^2 + gD\zeta \,d^2r
\,,\\&=
\int \left(\frac{1}{2} |\bs{\wh{v}}|^2 + \frac{1}{2} \wh{w}^2
 + g\zeta \right) \,D\,d^2r
\,,\\&=
\int \left(\frac{1}{2} |\widehat{\nabla_{\bs{r}}\phi}|^2 
 + \frac{1}{2} \wh{w}^2
 + g\zeta \right) \,D\,d^2r
\,.\end{split}
	\label{WW-LP-Ham}
\end{align}
The Hamiltonian in \eqref{WW-LP-Ham} is also the conserved energy \eqref{WW-LP-Erg-CWW-m} for 
the system of ECWWE  in \eqref{WW-eqns1}, as proven in Theorem \ref{Erg-conserv}.

\paragraph{Variations of the Hamiltonian in \eqref{WW-LP-Ham}.}
 In the Hamiltonian variables, the Bernoulli function ${\varpi}$ in \eqref{var-eqns-FS} is denoted as
\begin{equation}
	\varpi
	= \frac{1}{2D^2} \big|\bs{M} 
	- \lambda \nabla_{\bs{r}} \zeta \big|^2
    + \frac{\lambda^2}{2D^2} (1+\epsilon|\nabla_{\bs{r}}\zeta|^2)
	- g\zeta
	\,.
	\label{FS-Bernoulli-varpi}
\end{equation}
After evaluating the corresponding variational derivatives of the Hamiltonian in \eqref{WW-LP-Ham}, 
the system of equations in \eqref{EulerPoincare-ACWW} may be written in the untangled block-diagonal form, as
\begin{align}
\begin{split}
\frac{\p}{\p t}
\begin{bmatrix}
M_i \\ D \\ \lambda \\ \zeta 
\end{bmatrix}
= -
\begin{bmatrix}
\partial_j M_i + M_j \partial_i & D\partial_i  & 0 & 0
\\
\partial_jD 	    & 0  & 0 & 0
\\
0	& 0 & 0  & 1
\\
0	 & 0  & -1 & 0
\end{bmatrix}
\begin{bmatrix}
 {\delta h}/{\delta {M_j}} = \wh{v}^j
 \\ {\delta h}/{\delta D} = -\,\varpi
 \\ 
{\delta h}/{\delta \lambda} = -  (\bs{\wh v} - \epsilon\wh{w}\nabla_{\bs{r}}\zeta )\cdot\nabla_{\bs{r}}\zeta + \wh{w}
\\ {\delta h}/{\delta \zeta} 
=  {\rm div}_{\bs{r}}  \Big( \lambda \big(\bs{\wh v} - 2\epsilon\wh{w}\nabla_{\bs{r}}\zeta \big)\Big)
	 + g D 
\end{bmatrix}
\,.
\end{split}
	\label{FS-diag-brkt}
\end{align}
\paragraph{Casimir functions.} The Casimir functions, conserved by the relation $\{ C_{\Phi},h \}=0$ with any Hamiltonian $h(\bs{M},D)$ for the block-diagonal Lie-Poisson bracket in equation \eqref{FS-diag-brkt} are given by
\begin{align}
C_{\Phi} := \int \Phi(q)\,D\,d^2r 
\quad\hbox{for}\quad 
q := D^{-1}\zh\cdot {\rm curl}(\bs{M}/D)
\quad\hbox{with}\quad 
\p_tq + \bs{\wh{v}}\cdot\nabla_{\bs{r}}q = 0\,.
	\label{FS-diag-brkt-Casimirs}
\end{align}
As one may verify, the $C_{\Phi}$ are conserved for any differentiable function, $\Phi$, provided the velocity $\bs{\wh{v}}$ is tangent on the two-dimensional boundary. 

The proof of the constancy of the family of functions $C_{\Phi}$ is straightforward and well-known. That the $C_{\Phi}$ comprise a family of Casimirs so that $\{ C_{\Phi},h \}=0$ for any Hamiltonian $h(\bs{M},D)$ is a standard result for semidirect-product Lie-Poisson brackets. See, e.g., \cite{HMRW1985}.\bigskip

\begin{remark}[Consequences of introducing the wave-current minimal coupling term (WCMC)]
The consequences of introducing the slip velocity $\bs{s}:=\widehat{w}\nabla_{\bs{r}}\zeta$ in \eqref{slipvelocity-defX} into the variational principle in \eqref{ActionIntegral-FSmod} as a WCMC term are evident in the Bernoulli function in \eqref{FS-Bernoulli-varpi} and in the transport velocities of the wave dynamics in \eqref{FS-diag-brkt}, upon comparing them with \eqref{Bernoulli-varpi-M} and \eqref{FS-diag-brkt-CWW}, respectively. In contrast to the complexity of the separate relations for wave and current circulation laws in \eqref{ACWW-wavemomtrans} and \eqref{ACWW-fluidmomtrans1}, the simplicity of the conservation of the total circulation in equation \eqref{Kel-circACWW-thm} for ACWW dynamics seems to be a more meaningful statement about WCI than in the ECWWE, where the wave and current circulations are conserved separately in equation \eqref{eq: 2 KelCircThms}, as a \emph{mutual non-acceleration pact}. In the next section of the paper, we will explore the further ramifications of introducing the WCMC term, by adding non-hydrostatic pressure, buoyancy and other physics to the ACWW system. 
\end{remark}

%%%%%%%%%%%%%%%%%%%%%%%%%%%%%%%%%%%%%%%%%%%%%%%%%%%%%%%%%%%%%%%%%%%%%%%%%%%%%%

%%%%%%%%%%%%%%%%%%%%%%%%%%%%%%%%%%%%%%%%%%%%%%%%%%%%%%%%%%%%%%%%%%%%%%%%%%%%%%
%%%%%%%%%%%%%%%%%%%%%%%%%%%%%%%%%%%%%%%%%%%%%%%%%%%%%%%%%%%%%%%%%%%%%%%%%%%%%%

\section{Hamilton's principle for wave-current interaction on a free surface (WCIFS)}\label{sec: FECWWE}

%%%%%%%%%%%%%%%%%%%%%%%%%%%%%%%%%%%%%%%%%%%%%%%%%%%%%%%%%%%%%%%%%%%%%%%%%%%%%%
\subsection{Adding buoyancy and other physics to the ACWW system}\label{sec: FECWWE-HP}

This section further augments the ACWWE set \eqref{EulerPoincare-ACWW} to add more physical aspects to the wave-current interaction on a free surface (WCI FS). These physical aspects include wave-current coupling, non-hydrostatic pressure, incompressibility, and horizontal gradients of buoyancy.

\paragraph{Hamilton's principle for WCIFS.}
Let us modify the action integral \eqref{ActionIntegral-FSmod} for the system of ACWWE  in \eqref{EulerPoincare-ACWW} to encompass the following aspects of wave-current interaction on a free surface. As in \eqref{ActionIntegral-FSmod}, we will impose the surface boundary condition \eqref{u3-zetaX} and the continuity equation for the areal density variable $D$ as constraints. We will also include the wave-current minimal coupling (WCMC) term via the slip velocity, as in section \ref{subsec:CWWE_with_Coupling}.  In addition, we will introduce an advected scalar buoyancy variable $\rho$ with nonzero horizontal gradients. Finally, we will allow non-hydrostatic pressure, $p$. To determine the pressure, $p$, we will constrain the two-dimensional fluid transport velocity $\bs{\wh v}$ to be divergence-free.\footnote{We will reserve the hat notation for aspects of velocity $(\bs{\wh v},{\wh w},{\wh \phi})$ evaluated on the free surface. Hence, we will refrain from gratuitously adding hats to the pressure $p$ and the buoyancy $\rho$, since it is understood that they are evaluated on the free surface. The meaning for pressure $p$ and the buoyancy $\rho$ will always be clear from the context. The assumption of incompressibility of the fluid flow will enable the Bernoulli law to admit finite non-hydrostatic pressure.} To include these various physical effects, we will apply Hamilton's principle with the following dimension-free action integral, 
\begin{align}
\begin{split}
	S &=\int \ell(\bs{\wh{v}}, D,\rho,\widehat{\phi},{\wh w}, \zeta; \widetilde{\mu},p)dt
	\\&= \int \int  
	D\rho\left(\frac{1}{2}\big(|\bs{\wh{v}}|^2+\sigma^2\widehat{w}^2\big)  - \frac{\zeta}{Fr^2}\right)   
	+ \sigma^2\widetilde{\mu}\Big( \partial_t\zeta  + \bs{\wh{v}}\cdot\nabla\zeta - \widehat{w} \Big)
	- \sigma^4\epsilon\bs{s}\cdot(\widetilde{\mu}\nabla\zeta)
	\\
	&\qquad - \frac{1}{Fr^2}\,p(D-1)
	+ \widehat{\phi}\big(\partial_t D + \text{div}(D\bs{\wh{v}})\big) 
	+ \gamma(\partial_t\rho + \bs{\wh{v}}\cdot\nabla\rho) \,d^2r\,dt\,.
\end{split}
	\label{ActionIntegral-3revX}
\end{align}
Here, we recall that the quantity $\bs{s}:=\widehat{w}\nabla_{\bs{r}}\zeta$ is the slip velocity, defined in equation \eqref{slipvelocity-defX}.

\begin{remark}
    Note that by including only certain terms in the above action integral, we may derive equations for the dynamics of subsystems with any combination of these additional properties (wave-current coupling, non-hydrostatic pressure, incompressibility, buoyancy).
\end{remark}

\paragraph{The passive wave case, $\epsilon=0$.}
Taking variations of the dimensional version of action integral in \eqref{ActionIntegral-3revX} with $\epsilon=0$ yields 
\begin{align}
\begin{split}
	\delta \bs{\wh{v}}:&\quad 
	D\rho \bs{\wh{v}}\cdot d\bs{x} + \widetilde{\mu}\,d\zeta 
	=  Dd\widehat{\phi} - {\gamma}d\rho
	\,,\\ 
	\delta \widehat{w} :&\quad 
	D\rho\widehat{w} - \widetilde\mu = 0
	\,,\\
	\delta \widetilde{\mu}:&\quad 
	 \partial_t\zeta  + \bs{\wh{v}}\cdot\nabla\zeta =  \widehat{w}
	\,,\\
	\delta \zeta:&\quad 
	\partial_t \widetilde{\mu} + \text{div}(\widetilde{\mu}\bs{\wh{v}})
	=
	 - \,D\rho g
	\,,\\
	\delta\rho :&\quad 	
	\big(\partial_t + \mathcal{L}_{\bs{\wh{v}}}\big)\left(\frac{\gamma}{D}\right) 
	= \frac{1}{2}\Big(|\bs{\wh{v}}|^2+\widehat{w}^2\Big)  - g\zeta  =: {\varpi}
	\,,\\
	\delta D:&\quad 
	\big(\partial_t + \mathcal{L}_{\bs{\wh{v}}}\big)\widehat{\phi} = \rho{\varpi} - p 
	\,,\\
	\delta \widehat{\phi} :& \quad 
	\partial_t D + \text{div}(D\bs{\wh{v}}) =0 
	\,,\\
	\delta p:&\quad 
	D-1 = 0 
	\quad \implies \text{div}\bs{\wh v} =0 
	\,,\\
	\delta\gamma :&\quad 
	\big(\partial_t + \mathcal{L}_{\bs{\wh{v}}}\big)\rho=0
	\,.
\end{split}
	\label{var-eqns-3revX}
\end{align}
Applying $(\p_t +\mathcal{L}_{\bs{\wh{v}}})$ to the first relation in \eqref{var-eqns-3revX} yields
\begin{equation}\label{EP-3rexV}
    (\p_t +\mathcal{L}_{\bs{\wh{v}}})(\bs{\wh{v}}\cdot d\bs{r} + \widehat{w}d\zeta) = d\varpi - \frac{1}{\rho}dp\,,
\end{equation}
so we obtain the following Kelvin circulation theorem,
\begin{equation}\label{PassiveKelvinThm}
\frac{d}{dt}\oint_{c_t} \bs{V}\cdot d\bs{r} 
=
\oint_{c_t} (\p_t +\mathcal{L}_{\bs{\wh{v}}})\big(\bs{V}\cdot d\bs{r}\big)
=
- \oint_{c_t}  \frac{1}{\rho}dp
\,.
\end{equation}
As expected, the momentum per unit mass in the motion equation is $\bs{\wh{v}} + \widehat{w}\nabla\zeta =: \bs{V}$. The result has the same right-hand side as for the two-dimensional inhomogeneous Euler equation. Here, the total momentum now is the sum of the fluid momentum and the wave momentum, whose evolution is obtained as a separate degree of freedom appearing in the third and fourth lines of the equation set \eqref{var-eqns-3revX}.

However, continuing to calculate from \eqref{EP-3rexV} yields a non-acceleration result as in equation \eqref{eq: 2 KelThms}, 
in the sense that the wave momentum evolves passively with the flow of the fluid,
\[
(\p_t +\mathcal{L}_{\bs{\wh{v}}})(\widehat{w}\nabla\zeta)
= -g\,d\zeta + \frac12 d \widehat{w}^2\,,
\]
and the fluid momentum evolves independently of the wave variables,
\[
(\p_t +\mathcal{L}_{\bs{\wh{v}}})(\bs{\wh{v}}\cdot d\bs{r})
=
d\left(\frac12|\bs{\wh v}|^2\right) - \frac{1}{\rho}dp\,.
\]

%Namely, the fluid and wave evolve separately, 
%\begin{align*}
     %(\p_t +\mathcal{L}_{\bs{\wh{v}}})(\bs{\wh{v}}\cdot d\bs{r}) %&= d\varpi - \frac{1}{\rho}dp - (\p_t %+\mathcal{L}_{\bs{\wh{v}}})(\widehat{w}d\zeta) \\
     %&= d\varpi - \frac{1}{\rho}dp - \widehat{w}d\widehat{w} + gd\zeta \\
     %&= d\left(\frac12|\bs{\wh v}|^2\right) - \frac{1}{\rho}dp\,,
%\end{align*}

Thus, the momentum equation in the passive wave case is simply a 2D Euler equation to be considered in tandem with the remaining identities from \eqref{var-eqns-3revX}. Note that if $\nabla\rho\ne0$ then the right-hand side of the previous equation generates circulation in $\bs{\wh{v}}$; so, in this case $\bs{\wh{v}}$ cannot produce potential flow. 

Next, we will pursue the implications when the wave-current minimal coupling (WCMC) parameter $\epsilon$ does not vanish and the wave variables do not interact passively.

%%%%%%%%%%%%%%%%%%%%%%%%%%%%%%%%%%%%%%%%%%%%%%%%%%%%%%%%%%%%%%
%\end{framed}
%%%%%%%%%%%%%%%%%%%%%%%%%%%%%%%%%%%%%%%%%%%%%%%%%%%%%%%%%%%%%%

%%%%%%%%%%%%%%%%%%%%%%%%%%%%%%%%%%%%%%%%%%%%%%%%%%%%%%%%%%%%%%%%%%%%%%%%%%%%%%
\subsection{Derivation of the WCIFS equations for active waves}\label{sec:ClebschVariations}

To derive a WCIFS model system of equations for the motion of free surface with active waves and spatially varying buoyancy, we will apply the free-surface condition \eqref{u3-zetaX} and incompressibility of the $\bs{\wh v}$-flow as constraints in the action integral, while also including the minimal coupling term with nondimensional parameter $\epsilon\ne0$ in the action integral \eqref{ActionIntegral-3revX}. Then, upon restoring dimensionality to the variables in \eqref{ActionIntegral-3revX}, we obtain the following action principle for the  free-surface motion, 
\begin{align}
\begin{split}
	S &=\int \ell(\bs{\wh v},\widehat w, D,\rho,\zeta;\widetilde{\mu},p)dt
	\\& = \int \int  D\rho\left(\frac{1}{2}\big(|\bs{\wh v}|^2+{\widehat w}^2\big)  - g\zeta\right)   
	+ \widetilde{\mu}\Big( \partial_t\zeta  + \bs{\wh v}\cdot\nabla\zeta - {\widehat w} \Big)
	- \epsilon\wh{w}\nabla\zeta\cdot(\widetilde{\mu}\nabla\zeta)
	\\&\hspace{2cm} - p(D-1)+ \widehat{\phi}\big(\partial_t D + \text{div}(D\bs{\wh v})\big) 
	+ \gamma(\partial_t\rho + \bs{\wh v}\cdot\nabla\rho) \,d^2r\,dt\,.
\end{split}
	\label{ActionIntegral-3}
\end{align}
The Lagrange multipliers $\widetilde{\mu}$, $p$, $\wh{\phi}$ and $\gamma$, apply, respectively, the free-surface condition \eqref{u3-zetaX}, incompressibility of the $\bs{\wh v}$-flow, mass preservation, and buoyancy advection.

Hamilton's principle, $\delta S=0$, for the restricted free-surface action integral in \eqref{ActionIntegral-3} yields the following independent relations,
\begin{align}
\begin{split}
	\delta \bs{\wh v}:&\quad 
	D\rho \bs{\wh v}\cdot d\bs{x} + \widetilde{\mu}\,d\zeta =  Dd\widehat{\phi} - {\gamma}d\rho
	\,,\\ 
	\delta \widehat w :&\quad 
	D\rho \widehat w - \widetilde{\mu}\,\big(1+\epsilon|\nabla\zeta|^2\big) =  0
	\,,\\
	\delta \widetilde{\mu}:&\quad 
	 \partial_t\zeta  + \bs{\wh v}\cdot\nabla\zeta - {\widehat w}(1+\epsilon|\nabla\zeta|^2) = 0
	\,,\\
	\delta \zeta:&\quad 
	\partial_t \widetilde{\mu} + \text{div}(\widetilde{\mu}\bs{\wh v})
	=
	 - \,D\rho g + 2\epsilon\,{\rm div} \big( {\widehat w} \, \widetilde{\mu}\, \nabla \zeta\big)
	\,,\\
	\delta\rho :&\quad 	
	\big(\partial_t + \mathcal{L}_{\bs{\wh v}}\big)\left(\frac{\gamma}{D}\right) 
	= \frac12\big(|\bs{\wh v}|^2+{\widehat w}^2\big) - g\zeta  =: {\varpi}
	\,,\\
	\delta D:&\quad 
	\big(\partial_t + \mathcal{L}_{\bs{\wh v}}\big)\widehat{\phi} = \rho{\varpi} - p 
	\,,\\
	\delta \widehat{\phi} :& \quad 
	\partial_t D + \text{div}(D\bs{\wh v}) =0 
	\,,\\
	\delta p:&\quad 
	D-1 = 0 
	\quad \implies \text{div}\bs{\wh v} =0 
	\,,\\
	\delta\gamma :&\quad 
	\big(\partial_t + \mathcal{L}_{\bs{\wh v}}\big)\rho=0
	\,.
\end{split}
	\label{var-eqns-3}
\end{align}

Before enforcing the pressure constraint $D=1$, we write out the fluid motion equation obtained by substituting the variational results into the application of the advective time derivative $(\partial_t+\mathcal{L}_{\bs{\wh v}})$ on the first line of the system \eqref{var-eqns-3}, to find, upon writing  $\widetilde{\mu}= D\rho\widetilde{w}$
\begin{align}
\begin{split}
	D\rho(\partial_t+\mathcal{L}_{\bs{\wh v}})\big( \bs{\wh v}\cdot d\bs{x} + \widetilde{w}\,d\zeta \big) 
	&= Dd(\partial_t+\mathcal{L}_{\bs{\wh v}})\widehat{\phi} 
	- D\left((\partial_t+\mathcal{L}_{\bs{\wh v}})\frac{\gamma}{D}\right)d\rho 
	\\
	&= D\rho\Big(d{\varpi} - \frac1\rho dp\Big) 
	\,.
\end{split}
\label{EulerPoincare-3}
\end{align}
Recall that the Bernoulli function ${\varpi}$ and vertical wave momentum density $\widetilde{\mu}$ are defined as 
\begin{align}
{\varpi} :=  \frac12\big(|\bs{\wh v}|^2+{\widehat w}^2\big) - g\zeta 
\,,\qquad
\widetilde{\mu} := \frac{ D\rho {\widehat w}}{\big(1+\epsilon|\nabla\zeta|^2\big)} = D\rho\widetilde{w}
\,.
	\label{Bernoulli-wavemom-def}
\end{align}

At this point, let us collect the WCIFS equations in terms of (i) fluid variables, comprising velocity, $\bs{\wh v}$, area density, $D$, buoyancy, $\rho$, and (ii) wave variables, comprising surface elevation, $\zeta$ and vertical wave momentum density $\widetilde{\mu}$. The WCIFS equations are,
\begin{align}
\begin{split}
(\partial_t+\mathcal{L}_{\bs{\wh v}})\big( \bs{\wh v}\cdot d\bs{x} + \widetilde{w}\,d\zeta \big)
&= d{\varpi} - \frac1\rho dp
\,,
\\
\partial_t D + \text{div}(D\bs{\wh v}) &=0
\,,
\\
\partial_t\rho + {\bs{\wh v}}\cdot\nabla\rho &=0\,,
\\
\partial_t\zeta + {\bs{\wh v}}\cdot\nabla\zeta
&=
{\widehat w}(1+\epsilon|\nabla\zeta|^2) \,,
\\
\partial_t \widetilde{w} + {\bs{\wh v}}\cdot\nabla \widetilde{w} 
	&=
	 - \,g + 
	\frac{2\epsilon}{D\rho}\,{\rm div}  \big( {\widehat w} \, \widetilde{\mu}\, \nabla \zeta\big)
	\,.
\end{split}
\label{eq: WCIFS}
\end{align}

In its role as a Lagrange multiplier in the action integral \eqref{ActionIntegral-3}, the pressure $p$ enforces the constraint $D=1$. In turn, persistence of the condition $D=1$ along the flow implies that the fluid motion generated by $\bs{\wh v}$ is incompressible. In particular, setting $D=1$ in the continuity equation in \eqref{eq: WCIFS} above implies that the free surface fluid velocity $\bs{\wh v}$ is divergence free, ${\rm div} \bs{\wh v}=0$. The pressure $p$ is then determined by requiring that $\bs{\wh v}$ remain divergence free, which implies  the following elliptic equation for $p$,
\begin{equation}
\begin{aligned}
- (\nabla \cdot \rho^{-1}\nabla ) p 
=
{\rm div}\bigg(&
{\bs{\wh v}}\cdot\nabla \bs{\wh v}   
+ \frac{\widehat w}{1+\epsilon|\nabla\zeta|^2}\nabla\big( \widehat{w}(1+\epsilon|\nabla\zeta|^2) \big) \\
&- \widehat{w}\nabla\widehat{w} 
+ \frac{2\epsilon}{\rho}\,{\rm div}  
	\Big( \rho\, \widetilde{w}^2 \big(1+\epsilon|\nabla\zeta|^2\big) \,\nabla \zeta \Big)\nabla \zeta
\bigg).
\label{pressure-eqn}
\end{aligned}
\end{equation}
Thus, the pressure $p$ depends on the horizontal flow velocity $\bs{\wh v}$ of the surface current and fluid buoyancy $\rho$, as well as the wave elevation $\zeta$ and the vertical velocity $w$. We stress that the flow variables, $(\bs{\wh v},\rho)$, and the wave variables, $(\zeta, \widetilde{w})$, comprise two separate Eulerian degrees of freedom at each point $\bs{r}=(x,y)$ in the two-dimensional domain of flow.

\begin{remark}[Making the action integral stochastic in section \ref{sec: Stoch WCIFS eqns}]
In section \ref{sec: Stoch WCIFS eqns} the action integral \eqref{ActionIntegral-3revX} will be made stochastic, following \cite{Holm2015}, and we will to derive a stochastic generalisation of the water-wave equations.

\end{remark}

%%%%%%%%%%%%%%%%%%%%%%%%%%%%%%%%%%%%%%%%%%%%%%%%%%%%%%%%%%%%%%%%%%%%%%%%%%%%%%
\subsection{Comparison of WCIFS system to other known systems}

\begin{remark}[Comparison of system \eqref{eq: WCIFS} to the John-Sclavounos (JS) model equations]
The JS model comprises a dynamical system of ordinary differential equations for the motion of a single particle which is constrained to remain upon the free surface $\zeta(x,y,t)-z=0$, with \emph{prescribed} $\zeta(x,y,t)$. This dynamical system has recently been derived from a variational principle using the Euler-Lagrange methodology \cite{Chandre2016}. This variational principle raises the question of whether the particle dynamics of JS model may be associated with Lagrangian fluid trajectory dynamics in the present continuum framework.

The JS equations give the horizontal fluid particle trajectories $\bs{r}(t) = (x(t),y(t))$ driven by the free surface $z=\zeta(\bs{r},t)$. The equations can be expressed as
\begin{align}
	(1+\zeta_x^2)\ddot x + \zeta_x\zeta_y\ddot y + (\zeta_{tt}+\zeta_{xt}\dot x + \zeta_{yt}\dot y + \zeta_{xx}\dot x^2 + 2\zeta_{xy}\dot x\dot y +\zeta_{yy}\ddot y + g)\zeta_x &= 0, \\
	(1+\zeta_y^2)\ddot y + \zeta_x\zeta_y\ddot x + (\zeta_{tt}+\zeta_{xt}\dot x + \zeta_{yt}\dot y + \zeta_{xx}\dot x^2 + 2\zeta_{xy}\dot x\dot y +\zeta_{yy}\ddot y + g)\zeta_y &= 0.
\end{align}
Note that
\begin{equation*}
	\partial_t(\zeta_t+\zeta_x\dot x + \zeta_y\dot y) + g - \zeta_x\ddot x - \zeta_y\ddot y = \zeta_{tt}+\zeta_{xt}\dot x + \zeta_{yt}\dot y + \zeta_{xx}\dot x^2 + 2\zeta_{xy}\dot x\dot y +\zeta_{yy}\ddot y + g\,.
\end{equation*}
Hence, the JS equations can be re-written in more concise vector notation as
\begin{equation}
	\bs{\ddot r} + \Big((\partial_t+ \bs{\dot r}\cdot \nabla)(\partial_t\zeta + \bs{\dot r}\cdot \nabla \zeta) + g\Big)\nabla\zeta 
	=: \bs{\ddot r}  + \bigg(\frac{D}{Dt} \Big(\frac{D\zeta}{Dt} \Big)  + g\bigg)\nabla\zeta = 0\,,
\label{JS-vec}
\end{equation}
with $D/Dt:=\partial_t+ \bs{\dot y}\cdot \nabla.$
	\end{remark}
\paragraph{Choi's relation.} One may immediately make the connection between the JS equations \eqref{JS-vec} and Choi's relation \eqref{ChoiEqnX}. Naturally, since the JS equations represent a single particle's motion whereas Choi's relation is a statement about continuum flows, \eqref{ChoiEqnX} features a pressure term on the right hand side. However, the two equations are otherwise strikingly similar.

\paragraph{Comparison with the JS equations.} To make the comparison between the system of equations derived in this paper with the JS equations in vector form \eqref{JS-vec}, we combine the last two equations of the system \eqref{eq: WCIFS} to write, 
\begin{align}
g -\,
	\frac{2\epsilon}{D\rho}\,{\rm div}  
	\Big( D\rho\, \widetilde{w}^2 \big(1+\epsilon|\nabla\zeta|^2\big) \,\nabla \zeta \Big)
	=   -\, (\partial_t + {\bs{\wh v}}\cdot\nabla) \left(\frac{\partial_t\zeta + {\bs{\wh v}}\cdot\nabla\zeta}{\big(1+\epsilon|\nabla\zeta|^2\big)^2}\right)
	=   -\, (\partial_t + {\bs{\wh v}}\cdot\nabla)\widetilde w
	\,.
\label{JS-term}
\end{align}
Consequently, the motion equation in system \eqref{eq: WCIFS} may be expressed as 
\begin{equation}\label{CFS-MotionEquation}
\partial_t \bs{\wh v} + {\bs{\wh v}}\cdot\nabla \bs{\wh v}  + \Big( \big(\partial_t + {\bs{\wh v}}\cdot\nabla\big)\widetilde w + g  \Big) \nabla \zeta 
= 
- \frac1\rho \nabla p + \widehat w\nabla\widehat w - \frac{\widehat w}{1+\epsilon|\nabla\zeta|^2}\nabla\big(\widehat{w}(1+\epsilon|\nabla\zeta|^2)\big)
\end{equation}
Thus, the present form of the WCIFS fluid equations \eqref{CFS-MotionEquation} does seem to have some kinematic resemblance to the JS equations, although the two types of dynamics also have major physical and mathematical differences in their interpretations.

For example, one may write the WCIFS motion equation \eqref{CFS-MotionEquation} equivalently in more compact form, as
\begin{equation}\label{CFS-MotionEquation-w}
\partial_t \bs{\wh v} + {\bs{\wh v}}\cdot\nabla \bs{\wh v}  
+ \Big( \big(\partial_t + {\bs{\wh v}}\cdot\nabla\big) \widetilde{w} 
+ g  \Big) \nabla \zeta 
= 
 - \frac1\rho \nabla p - \frac12\widetilde{w}^2 \nabla \big(1+\epsilon|\nabla\zeta|^2\big)^2 \,.
\end{equation}
In this compact form, which is also reminiscent of Choi's relation \eqref{ChoiEqnX}, the geometric, coordinate-free nature of the WCIFS equation begins to emerge upon writing \eqref{CFS-MotionEquation} equivalently as the advective Lie derivative of a 1-form, which also arises in the Kelvin circulation theorem below, cf. \eqref{Kel-circCFS-thm},
\begin{equation}\label{CFS-geomMotEq1}
(\partial_t+\mathcal{L}_{\bs{\wh v}})\Big(\bs{\wh v}\cdot d\bs{x} + \widetilde{w} d \zeta\Big) 
=
d{\varpi} - \frac1\rho dp\,.
\end{equation}
In one dimension with constant buoyancy $\rho=\rho_0$ and $p=p_s=\rho_0 g\zeta$, this formula becomes, 
\begin{equation}\label{CFS-geomMotEq-1D}
(\partial_t+\mathcal{L}_{{v}})\Big({v} d{x} + \widetilde{w} d \zeta\Big) 
=
d{\varpi} - \frac{g}{\rho_0} d\zeta\,,
\end{equation}
where $\varpi$ is defined in \eqref{Bernoulli-wavemom-def}. In this geometric form, the JS and WCIFS models look rather more distant.  

%%%%%%%%%%%%%%%%%%%%%%%%%%%%%%%%%%%%%%%%%%%%%%%%%%%%%%%%%%%%%%%%%%%%%%%%%%%%%%
\subsection{Balance relations, Kelvin theorem and potential vorticity}
\begin{remark}[Dimension-free form of motion equation \eqref{CFS-geomMotEq1}]

In terms of the parameters in remark \ref{non-dim-scales}, the dimension-free form of the motion equation \eqref{CFS-geomMotEq1} is given by
\begin{equation}\label{CFS-geomMotEq2}
(\partial_t+\mathcal{L}_{\bs{\wh v}})\Big(\bs{\wh v}\cdot d\bs{x} +  \sigma^2\widetilde{w} d \zeta\Big) 
=
d\bigg(\frac12\big( |\bs{\wh v}|^2 + \sigma^2 {\widehat w}^2 \big)
- \frac{\zeta}{Fr^2} \bigg) 
 -  \frac{1}{Fr^2} \frac1\rho dp\,,
\end{equation}
where $\p_t\zeta+\bs{\wh v}\cdot\nabla\zeta = \widetilde{w} (1+\epsilon\sigma^2|\nabla\zeta|^2)^2 = {\widehat w}(1+ \epsilon\sigma^2 |\nabla\zeta|^2)$.

\paragraph{Balance relations required for significant wave-current interaction.}
For small Froude number, $Fr^2\ll1$, equation \eqref{CFS-geomMotEq2} approaches hydrostatic balance, and for small aspect ratio $\sigma^2\ll1$, equation \eqref{CFS-geomMotEq2} suppresses wave activity. When Froude number $Fr^2$ and aspect ratio $\sigma$ are both of order $O(1)$, then equation \eqref{CFS-geomMotEq2} admit order $O(1)$ significant non-hydrostatic wave activity.

Likewise, the $\widetilde{w}$ equation in \eqref{eq: WCIFS} in dimensionless form for the same scaling parameters becomes
\begin{align}
\sigma^2\big(\partial_t \widetilde{w} + {\bs{\wh v}}\cdot\nabla \widetilde{w} \big)
	=
	 - \frac{1}{Fr^2}
	 + 
	\frac{2\epsilon\sigma^4}{D\rho}\,{\rm div}  \big( D\rho\, {\widehat w} \, \widetilde{w}\, \nabla \zeta\big)
	\,.
\label{EulerPoincare-4w}
\end{align}
The balance between current and wave properties in the dimension-free $\widetilde{w}$ equation \eqref{EulerPoincare-4w} also requires both Froude number $Fr^2$ and aspect ratio $\sigma$ to be of order $O(1)$ for significant wave activity to occur. 

Only the motion equation and the $\widetilde{w}$ equation in \eqref{eq: WCIFS} change their coefficients for these scaling parameters. The coefficients of the others remain unchanged.

\end{remark}

\begin{remark}
	To prove the following Kelvin-Noether circulation theorem for the system of WCIFS equations in \eqref{eq: WCIFS} it is convenient to return to the variational equations in \eqref{var-eqns-3} and the notation introduced in \eqref{EulerPoincare-3} and \eqref{Bernoulli-wavemom-def}.
	\end{remark}
	
\begin{theorem}[Kelvin-Noether theorem for the WCIFS model]\label{KNthm-WCIFS}
For every closed loop $c(\bs{\wh v})$ moving with the WCIFS velocity $\bs{\wh v}$ for the system of WCIFS equations in \eqref{eq: WCIFS} the Kelvin circulation relation holds,
\begin{align}
\frac{d}{dt} \oint_{c(\bs{\wh v})} \Big(\bs{\wh v}\cdot d\bs{x} + \frac{\widetilde\mu}{D\rho}  d \zeta\Big) 
=
\oint_{c(\bs{\wh v})} d{\varpi} - \frac1\rho dp
\,.
	\label{Kel-circCFS-thm}
\end{align}
\end{theorem}

\begin{proof}
From the variational equations in \eqref{var-eqns-3}, we have
\begin{align*}
\begin{split}
	(\partial_t+\mathcal{L}_{\bs{\wh v}})\Big(D\rho\bs{\wh v}\cdot d\bs{x} + \widetilde\mu d \zeta\Big)
	&= Dd(\partial_t+\mathcal{L}_{\bs{\wh v}})\widehat{\phi} 
	- D\left((\partial_t+\mathcal{L}_{\bs{\wh v}})\frac{\gamma}{D}\right)d\rho 
	\\
	&= Dd (\rho{\varpi} - p )- D{\varpi} d\rho\,.
\end{split}
%\label{Kel-circCFS-proof}
\end{align*}
Hence, we find
\begin{align}
(\partial_t+\mathcal{L}_{\bs{\wh v}})\Big(\bs{\wh v}\cdot d\bs{x} + \frac{\widetilde\mu}{D\rho} d \zeta\Big) 
=
d{\varpi} - \frac1\rho dp
\,,
\label{Kel-circCFS-proof}
\end{align}
and the result \eqref{Kel-circCFS-thm} follows from the standard relation for the time derivative of an integral around a closed moving loop, $c(\bs{\wh v})$,
\[
\frac{d}{dt} \oint_{c(\bs{\wh v})} \Big(\bs{\wh v}\cdot d\bs{x} + \frac{\widetilde\mu}{D\rho}  d \zeta\Big) 
=
\oint_{c(\bs{\wh v})} 
(\partial_t+\mathcal{L}_{\bs{\wh v}})\Big(\bs{\wh v}\cdot d\bs{x} + \frac{\widetilde\mu}{D\rho} d \zeta \Big) 
=
\oint_{c(\bs{\wh v})} 
d{\varpi} - \frac1\rho dp
\,.\]
\end{proof}

\begin{remark}[Interpretation of WCIFS as a compound fluid system]
The compound circulation of the WCIFS wave-fluid system in \eqref{eq: WCIFS} obeys the same dynamical equations as the planar incompressible flow description of a single-component flow with horizontal buoyancy gradient, except for two features associated with the wave degrees of freedom. First, the presence of the wave field contributes to the solution for the pressure from the condition that the  velocity of the fluid component remains incompressible. Second, the presence of the wave field is a source of circulation for the fluid component of this compound system. Both of these features are due to the momentum of the waves, defined using the notation $\widetilde{w}$ defined in \eqref{eq: WCIFS} as ${\widetilde\mu}/(D\rho)  \nabla \zeta = \widetilde{w}\nabla \zeta$, which is proportional to the wave slope, $\nabla \zeta$. In particular, the momentum $\widetilde{w}\nabla \zeta$ appears in both the pressure equation in \eqref{pressure-eqn} and the Kelvin-Noether integrand in \eqref{Kel-circCFS-thm}.  
\end{remark}

\begin{corollary}[Total wave-fluid potential vorticity (PV) for WCIFS]\label{corollary:PV}
The evolution equation for the total wave-fluid potential vorticity (PV) follows by taking the exterior derivative of equation \eqref{Kel-circCFS-proof} in the proof of the Kelvin circulation theorem for WCIFS. Namely, 
\begin{align}
(\partial_t+\mathcal{L}_{\bs{\wh v}})\Big({\rm curl}\bs{\wh v} 
+ \nabla \widetilde{w}\times \nabla \zeta\Big) \cdot d\bs{S}
=
-\, \nabla \rho^{-1} \times \nabla p \cdot d\bs{S}
\,.
\label{PV-CFS1}
\end{align}
If we introduce a stream function $\psi$, so that $\bs{\wh v}=\zh \times \nabla \psi$, then the previous equation can be written formally in terms of the 2D Laplacian $\Delta$ and the Jacobian $J(\psi, \rho) := \zh\cdot \nabla \psi\times \nabla \rho=\bs{\wh v}\cdot\nabla\rho$ between functions $\psi$ and $\rho$, then we have 
\begin{align}
\partial_t q + J(\psi \,,\, q )
=
- \,J\Big(\rho^{-1} , \nabla p \Big)
\,,\quad\hbox{with PV defined as}\quad
q:= \Delta \psi  
+ J\big(\widetilde{w} ,  \zeta\big)
\,.
\label{PV-CFS2}
\end{align}

\end{corollary}

\begin{remark}
Note that advancing the PV quantity $q$ in time in equation \eqref{PV-CFS2} requires one to advance the entire system of WCIFS equations in \eqref{eq: WCIFS}, as well as the solution of the following  elliptic equation \eqref{pressure-eqn} for the pressure $p$ to complete the evolution algorithm. 
\end{remark}

%%%%%%%%%%%%%%%%%%%%%%%%%%%%%%%%%%%%%%%%%%%%%%%%
\subsection{Integral conservation laws for the WCIFS equations}

\paragraph{Spatially varying specific buoyancy.}
The system of equations for the PV and specific buoyancy $(q,\rho^{-1})$ is given by 
 \begin{align}
\partial_t q + J(\psi \,,\, q )
=
- \,J\Big(\rho^{-1} , \nabla p \Big)
\,,\quad
\partial_t \rho^{-1} + J(\psi \,,\, \rho^{-1} )
=
0
\,.
\label{q-rho-eqns}
\end{align}
This system of equations implies that the following integral quantity is conserved under the $(q,\rho^{-1})$ dynamics,
\begin{align}
C_{\Phi,\Psi} := \int \Phi(\rho^{-1}) + q \Psi(\rho^{-1})\, d^2x
\,,
\label{PhiPsiCLs}
\end{align}
for arbitrary differentiable functions $\Phi$ and $\Psi$.

\paragraph{Spatially homogeneous specific buoyancy.}
In the case that the specific buoyancy is initially constant,  $\rho^{-1} = \rho_0^{-1}$, then it will remain constant, and $\nabla \rho^{-1}=0$ will persist throughout the WCIFS domain of flow.  In this case, the $(q,\rho^{-1})$ system \eqref{q-rho-eqns} will reduce to a single equation, $\partial_t q + J(\psi \,,\, q )=0$, describing simple advection of the PV quantity, $q$. Hence, the conserved  integral quantities are the familiar vorticity functionals from the 2D Euler equations, or potential vorticity functionals from the quasigeostrophic (QG) equation. Namely, for a spatially homogeneous initial specific buoyancy, the WCIFS system in \eqref{q-rho-eqns} will conserve the following class of integral quantities 
\begin{align}
C_{\Phi} := \int \Phi(q) \, d^2x
\,,
\label{EnstrophyCL}
\end{align}
for an arbitrary differentiable function $\Phi$.
Thus, the WCIFS integral conservation laws for PV in equations \eqref{PhiPsiCLs} and \eqref{EnstrophyCL} depend on whether the specific buoyancy gradient $(\nabla \rho^{-1})$  vanishes at the initial time.

\paragraph{Energy.}
The conserved WCIFS integrated energy is given by
\begin{align}
e(\bs{\wh v},\rho,\zeta, \widetilde{w}) 
&:=
\int \frac\rho2\big( |\bs{\wh v}|^2 + \widetilde{w}^2\big(1+\epsilon|\nabla\zeta|^2\big)^2\big) 
+ g\rho \zeta \,d^2x\,.
	\label{CFS-erg}
\end{align}
While the energy conservation law may be proven directly from the WCIFS equations in \eqref{eq: WCIFS}, it may be more enlightening to discover this energy via the Legendre transformation of the Lagrangian in the action integral $S$ in \eqref{ActionIntegral-3} and thereby determine the Hamiltonian formulation  and its remarkable properties for the WCIFS system. In particular, the Lie-Poisson bracket in the Hamiltonian formulation of the WCIFS system in the next section will explain the source of the WCIFS conservation laws and their relationships among each other from the viewpoint of the  Hamiltonian structure for the WCIFS system. 

%%%%%%%%%%%%%%%%%%%%%%%%%%%%%%%%%%%%%%%%%%%%%%%%%%%%%%%%%%%%%%%%%%%%%%%%%%%%%%
\subsection{Hamiltonian formulation of WCIFS in terms of potential vorticity}

\begin{remark}[WCI FS with constant buoyancy]\label{corollary2:PV}
The simplest form of the WCIFS equations in \eqref{eq: WCIFS} arises when the buoyancy is constant, i.e.,  $\rho=\rho_0$. In that case, the WCIFS equations reduce to 
\begin{align}
\begin{split}
(\partial_t+\mathcal{L}_{\bs{\wh v}})\Big(\bs{\wh v}\cdot d\bs{x} + \widetilde{w} d \zeta\Big) 
&=
d{\varpi} - \frac1\rho_0 dp
\,,
\\
(\partial_t +\mathcal{L}_{\bs{\wh v}})\zeta &=\widehat{w}\big(1+\epsilon|\nabla\zeta|^2\big)=:\widetilde{w}\big(1+\epsilon|\nabla\zeta|^2\big)^2
\\
(\partial_t +\mathcal{L}_{\bs{\wh v}})\widetilde{w} 
	&=
	 - \,g +
	2\epsilon\,{\rm div}  
	\Big( \widetilde{w}^2 \big(1+\epsilon|\nabla\zeta|^2\big) \,\nabla \zeta \Big)
\,,\\
	\hbox{with}\quad \bs{\wh v}&=\zh\times \nabla \psi	
\,,\\   
	\hbox{and}\quad
	(\partial_t +\mathcal{L}_{\bs{\wh v}})\big(1&+\epsilon|\nabla\zeta|^2\big) 
	= 2\epsilon\nabla \zeta\cdot \big(\nabla \left(\widehat{w}\big(1+\epsilon|\nabla\zeta|^2\big)\right) - \zeta_{,j}\nabla \wh v^j\big)
\,,\\   
	\hbox{with}\quad
{\varpi} &:=  \frac12|\bs{\wh v}|^2+\frac12\widetilde{w}^2\big(1+\epsilon|\nabla\zeta|^2\big)^2-g\zeta 
\,.
\end{split}
\label{EulerPoincare-5}
\end{align}
\end{remark}

\begin{corollary}[WCI FS in PV Hamiltonian form]\label{cor:PVbracket}
The fluid dynamical system \eqref{EulerPoincare-5}  for WCIFS with constant buoyancy in terms of PV is a Hamiltonian system whose Poisson bracket is the following sum of a Lie-Poisson bracket for PV as $q:= \Delta \psi + J(\widetilde{w} , \zeta) $ and a canonical bracket for the wave variables $(\zeta,\widetilde{w})$. Namely,
\begin{align}
\frac{df}{dt} = \Big\{ f,e \Big\} =
 \int q\, J \bigg( \frac{\delta f }{ \delta q } , \frac{\delta e }{ \delta q } \bigg)d^2x 
+ \int   \frac{\delta f }{ \delta \zeta } \frac{\delta e }{ \delta \widetilde{w} } 
-  \frac{\delta f }{ \delta \widetilde{w} } \frac{\delta e }{ \delta \zeta } \,d^2x
\,.
	\label{CFS-PB-q}
\end{align}
\end{corollary}

\begin{proof}
We write the energy \eqref{CFS-erg} in terms of PV defined as $q:= \Delta \psi + J(\widetilde{w} , \zeta) $ in  equation \eqref{PV-CFS2}, 
\begin{align}
e(q,\zeta, \widetilde{w}) 
&:=
\int \frac12\Big( \big(q- J(\widetilde{w} , \zeta) \big) (-\Delta^{-1})\big(q- J(\widetilde{w},\zeta)\big) + \widetilde{w}^2\big(1+\epsilon|\nabla\zeta|^2\big)^2 \Big)
+ g \zeta \,d^2x\,.
	\label{CFS-erg-q}
\end{align}
The variational derivatives of the energy $e(q,\zeta, \widetilde{w})$  in \eqref{CFS-erg-q} are given by
\begin{align}
\begin{bmatrix}
\delta e / \delta q  \\ \delta e / \delta \widetilde{w} \\  \delta e / \delta \zeta
\end{bmatrix}
=
\begin{bmatrix}
-\,\psi 
\\ 
 -\,J(\psi,\zeta)+ \widetilde{w}\big(1+\epsilon|\nabla\zeta|^2\big)^2
\\
 J(\psi,\widetilde{w}) 
 + g -\,
	2\epsilon\,{\rm div}  
	\Big( \widetilde{w}^2 \big(1+\epsilon|\nabla\zeta|^2\big) \,\nabla \zeta \Big)
\end{bmatrix}.
	\label{CFS-Ham-q}
\end{align}
Hence, equations \eqref{EulerPoincare-5} become 
\begin{align}
\begin{split}
\frac{\p}{\p t}
\begin{bmatrix}
q \\ \widetilde{w} \\ \zeta
\end{bmatrix}
=
\begin{bmatrix}
\{ q, e \}  \\ \{ \widetilde{w} , e \} \\   \{ \zeta, e \} 
\end{bmatrix}
&= -
\begin{bmatrix}
J(q\,,\,\cdot\,) & 0 & 0
\\
0 & 0 & 1
\\
0 & -1 & 0
\end{bmatrix}
\begin{bmatrix}
\delta e / \delta q  \\ \delta e / \delta \widetilde{w} \\  \delta e / \delta \zeta
\end{bmatrix}
\\&=
-
\begin{bmatrix}
J( \psi , q )  
 \\  
J(\psi,\widetilde{w})  
 + g -\,
	2\epsilon\,{\rm div}  
	\Big( \widetilde{w}^2 \big(1+\epsilon|\nabla\zeta|^2\big) \,\nabla \zeta \Big)
\\ 
J(\psi,\zeta)  - \widetilde{w}\big(1+\epsilon|\nabla\zeta|^2\big)^2
\end{bmatrix}.
\end{split}
	\label{CFS-Ham-q1}
\end{align}

\end{proof}

\begin{proof} One computes the PV bracket between $C_\Phi(q)$ and an arbitrary functional $F(q)$, as 
\[
\Big\{ F, C_\Phi(q) \Big\} =  \int q\, J \bigg( \frac{\delta F }{ \delta q } , \frac{\delta C_\Phi }{ \delta q } \bigg)d^2x
= -\, \int  \frac{\delta F }{ \delta q }\, J \Big( q , \Phi ' (q) \Big)d^2x 
= 0\,,\quad \hbox{for all }F(q)\,,
\]
after an integration by parts in a periodic domain, say. Thus, $C_\Phi(q)$ is a Casimir function for the PV Poisson bracket in \eqref{CFS-Ham-q1}. 

\end{proof}

\begin{proposition}
Sufficient conditions for $(q_e,\zeta_e, \widetilde{w}_e)$ to be an equilibrium solution of \eqref{CFS-Ham-q} arise by requiring that the functional $H_\Phi = e(q,\zeta, \widetilde{w}) + C_\Phi(q)$ would have a critical point at $(q_e,\zeta_e, \widetilde{w}_e)$.
\end{proposition}

\begin{proof}
Evaluated at $(q_e,\zeta_e, \widetilde{w}_e)$ a critical point the functional $H_\Phi$ satisfies
\[
\delta H_\Phi = \int (-\psi_e + \Phi'(q_e)\delta q 
+ 
\frac{\delta e}{\delta \zeta}\bigg|_{(q_e,\zeta_e, \widetilde{w}_e)} \delta\zeta
+ 
\frac{\delta e}{\delta \widetilde{w}}\bigg|_{(q_e,\zeta_e, \widetilde{w}_e)} \delta\widetilde{w}\ d^2x
\]
This is sufficient for the right-hand side of equation \eqref{CFS-Ham-q} to vanish and thereby produce an  equilibrium solution. 
\end{proof}

The Lie-Poisson brackets in the Hamiltonian formulations of the system provided here for the case of constant buoyancy, $\rho=\rho_0$, and in the next section for the general case explain the sources of the conservation laws and their relationships among each other from the viewpoint of the  Hamiltonian structure for the system.

%%%%%%%%%%%%%%%%%%%%%%%%%%%%%%%%%%%%%%%%%%%%%%%%%%%%%%%%%%%%%%%%%%%%%%%%%%%%%%
\subsection{Hamiltonian formulation of the WCIFS equations}

\paragraph{Legendre transformation.}
By considering the Lagrangian function in the action integral \eqref{ActionIntegral-3}, $\ell(\bs{\wh v},D,\rho,\zeta,w;p)$, one may define the Legendre transformation as the variation with respect to the velocity $\bs{\wh v}$ in \eqref{var-eqns-3}. Namely,
\[
\frac{\delta \ell}{\delta \bs{\wh v}}
= D\rho \bs{\wh v} + \widetilde{\mu}\, \nabla\zeta -  D\nabla\widehat{\phi} + \gamma \nabla\rho 
\,, \quad\hbox{where}\quad \widetilde{\mu} := \frac{D\rho \wh w}{1+\epsilon|\nabla\zeta|^2} = \frac{D\rho(\partial_t+\mathcal{L}_{\bs{\wh v}})\zeta}{\big(1+\epsilon|\nabla\zeta|^2\big)^2}
\,.
\]
Upon defining the fluid momentum as $\bs{m}=D\rho \bs{\wh v}$, the Hamiltonian in these variables is obtained via the following calculation
\begin{align}
\begin{split}
h(\bs{m},D,\rho,\zeta, \widetilde{\mu} ;p) 
&:= \left\langle \bs{m},\bs{\wh v} \right\rangle + \widetilde{\mu}\p_t\zeta - D\p_t \phi + \gamma\p_t\rho
- \ell(\bs{\wh v},D,\rho,\zeta,w;p)
\\&=
\int \frac{|\bs{m}|^2} {2D\rho} + \frac{\widetilde{\mu}^2}{ 2D\rho }\big(1+\epsilon|\nabla\zeta|^2\big)^2 
+ gD\rho \zeta + p(D-1)\,d^2x\,.
\end{split}
	\label{CFS-Ham1}
\end{align}

\paragraph{Conserved energy.}
The Hamiltonian in \eqref{CFS-Ham1} is also the conserved energy for 
the system of equations in \eqref{EulerPoincare-3}.

\paragraph{Variations of the Hamiltonian.}
 In the Hamiltonian variables, the Bernoulli function ${\varpi}$ in \eqref{Bernoulli-wavemom-def} is denoted as
\begin{equation}
	\widetilde{\varpi} = \frac{|\bs{m}|^2}{2D^2\rho^2} + \frac{\widetilde\mu^2\big(1+\epsilon|\nabla\zeta|^2 \big)^2}{2D^2\rho^2} - g\zeta\,.
	\label{Bernoulli-varpi}
\end{equation}
The corresponding variational derivatives of the Hamiltonian  in \eqref{CFS-Ham1} for 
the system of equations in \eqref{EulerPoincare-3} are given by
\begin{equation}
\begin{bmatrix}
 {\delta h}/{\delta {m_j}} \\ {\delta h}/{\delta D} \\ {\delta h}/{\delta \rho} \\
{\delta h}/{\delta \widetilde{\mu}} \\ {\delta h}/{\delta \zeta}
\end{bmatrix}
=  
\begin{bmatrix}
{\frac{m^j}{D\rho}} \\ -\rho\widetilde{\varpi} + p \\
-D\widetilde{\varpi} \\
 \frac{1}{D\rho}\widetilde\mu(1+\epsilon|\nabla\zeta|^2)^2 \\ D\rho g - 2\epsilon\,{\rm div} (\widetilde\mu\widehat w\nabla\zeta) 
 \end{bmatrix}
 =
 \begin{bmatrix}
{v^j} \\  -\rho\widetilde{\varpi} + p \\
 -D\widetilde{\varpi} \\
  \widetilde{w}(1+\epsilon|\nabla\zeta|^2)^2 \\ D\rho g - 2\epsilon\,{\rm div} (\widetilde\mu\widehat w\nabla\zeta) 
 \end{bmatrix}\,.
	\label{Ham-varderivs}
\end{equation}
The system of equations in \eqref{EulerPoincare-3} may now be written in Lie-Poisson form, 
augmented by a symplectic 2-cocycle in the elevation $\zeta$ and its canonical momentum density  $\widetilde\mu$ 
in its \emph{entangled} form as
\begin{align}
\begin{split}
\frac{\p}{\p t}
\begin{bmatrix}
m_i \\ D \\ \rho \\ \widetilde{\mu} \\ \zeta 
\end{bmatrix}
= -
\begin{bmatrix}
\partial_j m_i + m_j \partial_i & D\partial_i & -\rho_{,i}  & \widetilde\mu\partial_i & -\zeta_{,i}
\\
\partial_jD 		&0 & 0  & 0 & 0
\\
\rho_{,j} 			&0 & 0 & 0 & 0
\\
\partial_j\widetilde\mu			&0 & 0 & 0 & 1
\\
\zeta_{,j}			 &0 & 0  & -1 & 0
\end{bmatrix}
\begin{bmatrix}
 {\delta h}/{\delta {m_j}} \\ {\delta h}/{\delta D} \\ {\delta h}/{\delta \rho} \\
{\delta h}/{\delta \widetilde{\mu}} \\ {\delta h}/{\delta \zeta}
\end{bmatrix}
\,.
\end{split}
	\label{CFS-LPB1-brkt}
\end{align}
\begin{remark}[Physical meaning of the model]$\,$\\
The Lie-Poisson structure in \eqref{CFS-LPB1-brkt} reveals the physical meaning of the WCIFS system of equations. Namely, the fluid variables sweep the wave degrees of freedom along the fluid Lagrangian paths, while the wave subsystem evolves and acts back on the fluid circulation as an \emph{internal} force. 
\end{remark}

\begin{remark}[Transformation to the potential flow momentum]$\,$\\
The Poisson operator in the previous formula is block diagonalised by the  transformation $\bs{m}\to \bs{M}=\bs{m}+ \widetilde\mu\nabla \zeta$, which separates it into a direct sum of a Lie-Poisson bracket in $\bs{M},D,\rho$ and a canonical (symplectic) Poisson bracket in $\widetilde\mu$ and $\zeta$. Consequently, the system of equations in \eqref{EulerPoincare-3} may now be written equivalently as a direct sum of a semidirect-product Lie-Poisson bracket in the  fluid variables $(\bs{m},D,\rho)$, plus a symplectic 2-cocycle in the wave variables $(\widetilde\mu,\zeta)$ in its \emph{untangled} form, as
\begin{align}
\begin{split}
\frac{\p}{\p t}
\begin{bmatrix}
M_i \\ D \\ \rho \\ \widetilde{\mu} \\ \zeta 
\end{bmatrix}
= -
\begin{bmatrix}
\partial_j M_i + M_j \partial_i & D\partial_i & -\rho_{,i}  & 0 &0
\\
\partial_jD 		&0 & 0  & 0 & 0
\\
\rho_{,j} 			&0 & 0 & 0 & 0
\\
0			&0 & 0 & 0 & 1
\\
0			 &0 & 0  & -1 & 0
\end{bmatrix}
\begin{bmatrix}
 {\delta h}/{\delta {M_j}} \\ {\delta h}/{\delta D} \\ {\delta h}/{\delta \rho} \\
{\delta h}/{\delta \widetilde{\mu}} \\ {\delta h}/{\delta \zeta}
\end{bmatrix}
\end{split}
	\label{CFS-LPB2-brkt}
\end{align}
Thus, the Poisson bracket block-diagonalises when it is written in terms of the total fluid plus wave momentum, $\bs{M}=\bs{m}+ \widetilde\mu\,\nabla \zeta$.

\end{remark}

The Poisson structure in equation \eqref{CFS-LPB1-brkt} yields the following motion  equation
\begin{align}
\begin{split}
	(\p_t+\mathcal{L}_{\bs{\wh v}})(\bs{m}\cdot d\bs{x}\otimes d^2x) &= -Dd(-\rho\widetilde{\varpi} +p)\otimes d^2x - (D\widetilde{\varpi}) d\rho\otimes d^2x - \widetilde\mu d\big(\widetilde{w}(1+\epsilon|\nabla\zeta|^2)^2\big)\otimes d^2x \\
	&\qquad + \big(Dg\rho - 2\epsilon\,{\rm div} \left(\widetilde\mu\widehat w\nabla\zeta\right)\big)d\zeta\otimes d^2x\,.	
\end{split}
\label{rough-mot-eqn}
\end{align}
If we divide through by $D\rho$, using $(\p_t+\mathcal{L}_{\bs{\wh v}})(D\rho d^2x) = 0$, then we obtain the following motion equation, which agrees with that previously obtained in \eqref{eq: WCIFS},
\begin{align*}
	(\p_t+\mathcal{L}_{\bs{\wh v}})\bs{\wh v}\cdot d\bs{x} 
	&= -\frac{1}{\rho}dp + d\widetilde{\varpi} -\widetilde{w}d\big(\widetilde{w}(1+\epsilon|\nabla\zeta|^2)^2\big)  + gd\zeta
	- \frac{2\epsilon}{D\rho}{\rm div} \left(\widetilde\mu\widehat w\nabla\zeta\right)d\zeta 
	\\
	&= -\frac{1}{\rho}dp + d\left(\frac{|\bs{\wh v}|^2}{2} + \frac{{\widehat w}^2}{2}   \right) 
	-\widetilde{w}d\big(\widetilde{w}(1+\epsilon|\nabla\zeta|^2)^2\big)
	- \frac{2\epsilon}{D\rho}{\rm div} \left(\widetilde\mu\widehat w\nabla\zeta\right)d\zeta.
\end{align*}

\subsection{A one-dimensional WCIFS equation}\label{sec: 1D WCI}

We begin by deriving the one-dimensional equation by applying Hamilton's principle to the following action integral, which is the one-dimensional version of \eqref{ActionIntegral-3},

\begin{align}
\begin{split}
	S = \int \int  D\rho\bigg(\frac{1}{2}\big(\wh v^2+\wh w^2\big)  &- g\zeta\bigg)   
	+ \widetilde{\mu}\Big( \partial_t\zeta  + \wh v\p_x\zeta - {\wh w} \big(1+\epsilon|\p_x\zeta|^2\big)\Big)
	\\
	& + \widehat{\phi}\big(\p_t(D\rho)+\p_x(D\rho \wh v)\big)\,dxdt\,,
\end{split}
	\label{ActionIntegral-1D}
\end{align}
where $\zeta(x,t)$, $\wh{v}(x,t)$, and $\wh{w}(x,t)$ are scalar functions of one-dimensional space and time, $(x,t)$, and we consider the volume form $D\rho$ to be a single variable. Taking variations with respect to each variable gives,
\begin{align}
\begin{split}
	\delta \wh v:&\quad 
	D\rho \wh v\cdot dx + \widetilde\mu d\zeta=  D\rho d\widehat{\phi}
	\,,\\
	\delta(D\rho) :&\quad 	
	(\p_t+\mathcal{L}_{\wh v})\widehat{\phi} = \p_t\widehat{\phi} + {\wh v}\p_x\widehat{\phi}
	= \frac12{\wh v}^2+\frac12{\wh w}^2-g\zeta
	=: \varpi
	\,,\\
	\delta \widehat{\phi} :& \quad 
	(\p_t+\mathcal{L}_{\wh v})(D\rho)=\partial_t (D\rho) + \p_x(D\rho {\wh v}) =0 
	\,,\\
	\delta \widehat w :& \quad
	D\rho \widehat w - \widetilde\mu(1+\epsilon|\p_x\zeta|^2) =0
	\,,\\
	\delta \widetilde\mu :& \quad
	(\p_t+\mathcal{L}_{\wh v})\zeta = \p_t\zeta + {\wh v}\p_x\zeta = \widehat{w}(1+\epsilon|\p_x\zeta|^2)
	\,,\\
	\delta \zeta:&\quad 
 (\p_t+\mathcal{L}_{\wh v})\widetilde\mu =- \,D\rho g + 2\epsilon\p_x \left( \widetilde\mu\widehat w \p_x\zeta\right)
	\,.
\end{split}
	\label{var-eqns-1D}
\end{align}
These relations imply a fluid motion equation
\begin{align*}
	D\rho(\p_t+\mathcal{L}_{\wh v}){\wh v} &= D\rho\,d(\p_t+\mathcal{L}_{\wh v})\widehat{\phi}  - (\p_t+\mathcal{L}_{\wh v})\left(\widetilde\mu\,d\zeta\right) \\
	&= D\rho\, d\varpi - \widetilde\mu d\big( {\wh w}(1+\epsilon|\p_x\zeta|^2) \big) +D\rho gd\zeta-2\epsilon\p_x  \left( D\rho\widetilde w^2(1+\epsilon|\p_x\zeta|^2)\p_x\zeta \right)d\zeta\,,
\end{align*}
or,
\begin{align*}
	(\p_t+\mathcal{L}_{\wh v}){\wh v} &= d\varpi - \frac{\wh w}{1+\epsilon|\p_x\zeta|^2} d\big( {\wh w}(1+\epsilon|\p_x\zeta|^2) \big) + gd\zeta-\frac{2\epsilon}{D\rho}\p_x  \left( D\rho\widetilde w^2(1+\epsilon|\p_x\zeta|^2)\p_x\zeta \right)d\zeta \\
	&= d\left(\frac12\wh v^2 \right) -\frac12 \widetilde w^2d\big((1+\epsilon|\p_x\zeta|^2)^2\big) -\frac{2\epsilon}{D\rho}\p_x  \left( D\rho\widetilde w^2(1+\epsilon|\p_x\zeta|^2)\p_x\zeta \right)d\zeta\,.
\end{align*}
where $\widetilde\mu = D\rho\widetilde w$. Thus we have
\begin{equation}
	\p_t{\wh v}+{\wh v}\p_x{\wh v} = -\frac12 \widetilde w^2\p_x\big((1+\epsilon|\p_x\zeta|^2)^2\big) -\frac{2\epsilon}{D\rho}\p_x  \left( D\rho\widetilde w^2(1+\epsilon|\p_x\zeta|^2)\p_x\zeta \right)\p_x\zeta\,,
\end{equation}
and this equation is to be considered together with
\begin{align*}
	\partial_t D + \p_x(D{\wh v}) &=0\,, \\
	\p_t\rho+{\wh v}\p_x\rho&=0\,, \\
	\p_t\widetilde w + {\wh v}\p_x\widetilde w &= -g +\frac{2\epsilon}{D\rho}\p_x\left( D\rho\widetilde w^2(1+\epsilon|\p_x\zeta|^2)\p_x\zeta \right)\,.
\end{align*}
These one-dimensional WCIFS equations are of interest in their own right and they will be investigated elsewhere. 

%%%%%%%%%%%%%%%%%%%%%%%%%%%%%%%%%%%%%%%%%%%%%%%%%%%%%%%%%%%%%%%%%%%%%%%%%%%%%%
%%%%%%%%%%%%%%%%%%%%%%%%%%%%%%%%%%%%%%%%%%%%%%%%%%%%%%%%%%%%%%%%%%%%%%%%%%%%%%
\section{Stochastic wave modelling}\label{sec: Stoch WCIFS eqns}
%%%%%%%%%%%%%%%%%%%%%%%%%%%%%%%%%%%%%%%%%%%%%%%%%%%%%%%%%%%%%%%%%%%%%%%%%%%%%%
\subsection{Stochastic Advection by Lie Transport (SALT)}

Stochastic advection by Lie transport (SALT) \cite{Holm2015} provides a methodology of stochastically perturbing a continuum model at the level of the action integral. As a result, SALT preserves the Kelvin-Noether circulation theorem. Consider first the three dimensional case where we have a three dimensional fluid velocity field, evaluated on the free surface, denoted by $\bs{\wh u}$. For a deterministic (unconstrained) Lagrangian, $\ell(\bs{\wh u},q)$, depending on the velocity field $\bs{\wh u}$ and advected quantities $q$, we constrain the advected quantities to follow a stochastically perturbed path via a Lagrange multiplier. For models where we are considering incompressible flow, the pressure must act as a Lagrange multiplier to enforce the advected quantity $D$, the volume element, to be constant. More specifically, the advection constraint enforces that the advected quantities obey a stochastic partial differential equation given by
\begin{equation}\label{StochasticAdvectionEquation}
    (\rmd + \mathcal{L}_{\rmd \bs{x}_t})q \coloneqq \rmd q + \mathcal{L}_{\bs{\wh u}}q\,dt + \sum_i\mathcal{L}_{\bs{\tilde \xi}_i}q\circ dW_t^i = 0\,,
\end{equation}
where the vector field $\bs{\wh u}$ has been perturbed in the following way
\begin{equation}\label{Stochastic3DVF}
    \rmd \bs{x}_t = \bs{\wh u}\,dt + \sum_i\bs{\tilde \xi}_i\circ dW_t^i\,.
\end{equation}
After this introduction of this stochastic transport constraint, the action integral becomes a \emph{semi-martingale driven variational principle} \cite{CrisanStreet2020}. Consequently, the pressure Lagrange multiplier must be compatible with the noise introduced in the advection. This is required because one cannot enforce a variable in a stochastic system to remain constant without also requiring the Lagrange multiplier to also be a semi-martingale, in order to control both the deterministic part of the system as well as the random fluctuations. With these constraints, the action integral takes the form
\begin{equation}\label{IncompressibleSALTAction}
    S = \int \ell(\bs{\wh v},q)\,dt + \langle \rmd p, D-1\rangle + \langle \lambda, \rmd q + \mathcal{L}_{\rmd \bs{x}_t}q \rangle\,.
\end{equation}
The application of Hamilton's principle implies an Euler-Poincar\'e equation and, as in \cite{Holm2015} we have a Kelvin-Noether circulation theorem for the stochastic system which is analogous to that of the deterministic system.

For the purposes of our variational wave models, we need a notation for two dimensional advection as well as three dimensional. We recall the notation for the two dimensional velocity field and introduce a new notation for the first two components of the stochastic perturbation as follows
\begin{equation}
    \bs{\wh u} = (\bs{\wh v},\wh w)\,,\quad\hbox{and}\quad \bs{\tilde \xi_i} = (\bs{\xi_i},\xi_{3i})\,.
\end{equation} 
The perturbation of vector field $\bs{\wh v}$ is therefore given by
\begin{equation}\label{Stochastic2DVF}
    \rmd \bs{r}_t = \bs{\wh v}(\bs{r}_t,t)\,dt + \sum_i \bs{\xi_i}(\bs{r}_t)\circ dW_t^i\,.
\end{equation}

%%%%%%%%%%%%%%%%%%%%%%%%%%%%%%%%%%%%%%%%%%%%%%%%%%%%%%%%%%%%%%%%%%%%%%%%%%%%%%
\subsection{Stochastic ECWWE}
    Firstly, we derive the free surface boundary condition \eqref{bdyconditions:componentsX} in the stochastic case by applying the operator $\rmd + \mathcal{L}_{\rmd \bs{r}_t}$ to $z-\zeta(\bs{r},t)$ to obtain
    \begin{equation*}
       0=(\rmd + \mathcal{L}_{\rmd \bs{r}_t})(z-\zeta(\bs{r},t)) = \widehat{w}(\bs{r},t)\,dt + \sum_i\xi_{3i}(\bs{r})\circ dW_t^i - \rmd\zeta(\bs{r},t) - \mathcal{L}_{\rmd r}\zeta(\bs{r},t)\,,
    \end{equation*}
    and hence
    \begin{equation}
        \rmd\zeta(\bs{r},t) + \mathcal{L}_{\rmd \bs{r}_t}\zeta(\bs{r},t) 
        = \widehat{w}(\bs{r},t)\,dt + \sum_i\xi_{3i}(\bs{r})\circ dW_t^i \,.
    \end{equation}
    Where the notation $\rmd\bs{r}_t$ in \eqref{Stochastic2DVF} is the path of a Lagrangian coordinate. When we write dependence on $\bs{r}$, we mean that $\bs{r}$ is an Eulerian point which is the pullback of the path defined by \eqref{Stochastic2DVF}. In more informal language, $\bs{r}$ is an Eulerian point along the Lagrangian path $\bs{r}_t$.
    
    We may derive the stochastic ECWW equations by considering the dimensional version of the action integral \eqref{ActionIntegral-FS} where the transport velocity $\bs{\wh v}$ has been perturbed as in \eqref{Stochastic2DVF}. The stochastic action integral is then
    \begin{align}
\begin{split}
	S &= \int \int  
	D\left(\frac{1}{2}\big(|\bs{\wh{v}}|^2+\widehat{w}^2\big)  - g\zeta\right)\,dt   
	+ \lambda\Big( \rmd\zeta  + \mathcal{L}_{\rmd \bs{r}_t}\zeta - \widehat{w}\,dt - \sum_i\xi_{3i}\circ dW_t^i \Big)
	\\
	&\qquad 
	+ \widehat{\phi}\big(\rmd D + \mathcal{L}_{\rmd \bs{r}_t}D\big)
	\,d^2r\,.
\end{split}
	\label{StochasticActionIntegral-FS}
\end{align}
Taking variations of the action integral \eqref{StochasticActionIntegral-FS} yields
\begin{align}
\begin{split}
	\delta \bs{\wh{v}}:&\quad 
	D \bs{\wh{v}}\cdot d\bs{r} + \lambda\,d\zeta 
	=  Dd\widehat{\phi} 
	\quad\Longrightarrow\quad 
	\bs{V}\cdot d\bs{r} :=
	\bs{\wh{v}}\cdot d\bs{r} + \widehat{w}\,d\zeta 
	=  d\widehat{\phi}
	\,,\\ 
	\delta \widehat{w} :&\quad 
	D\widehat{w} - \lambda = 0
	\,,\\
	\delta \lambda:&\quad 
	 \rmd\zeta  + \mathcal{L}_{\rmd \bs{r}_t}\zeta =  \widehat{w}\,dt + \sum_i\xi_{3i}\circ dW_t^i
	\,,\\
	\delta \zeta:&\quad 
	\rmd \lambda + \mathcal{L}_{\rmd \bs{r}_t}\lambda
	= - \,gD\,dt
	\quad\Longrightarrow\quad 
	\partial_t \widehat{w} + \rmd \bs{r}_t\cdot\nabla_{\bs{r}}\widehat{w}
	=
	 - \,g
	\,,\\
	\delta \widehat{\phi} :& \quad 
	\rmd D + \mathcal{L}_{\rmd \bs{r}_t}D =0 
	\,,\\
	\delta D:&\quad 
	\rmd\wh\phi + \mathcal{L}_{\rmd \bs{r}_t}\wh\phi = \rmd\wh\phi + \rmd \bs{r}_t\cdot \nabla_{\bs{r}}\widehat{\phi} 
	= \frac{1}{2}\big(|\bs{\wh{v}}|^2+\widehat{w}^2\big)\,dt  - \zeta\,dt =:\varpi \,dt
	\,.
\end{split}
	\label{var-eqns-StochasticFS}
\end{align}
We may therefore write the stochastic ECWW equations as
\begin{align}
\begin{split}
	\rmd\widehat{\phi}  + \rmd\bs{r}_t\cdot \nabla_{\bs{r}}\widehat{\phi} 
	&= \frac{1}{2}\big(|\widehat{\nabla_{\bs{r}}\phi}|^2
	+ \widehat{w}^2\big)\,dt  
	- g\zeta\,dt
	\,,
	\\
	\rmd\zeta  + \rmd\bs{r}_t\cdot\nabla_{\bs{r}}\zeta 
	&=  \widehat{w}\,dt + \sum_i\xi_{3i}\circ dW_t^i
	\,,
	\\
	\rmd \widehat{w} + \rmd\bs{r}_t\cdot\nabla_{\bs{r}}\widehat{w}
	&=
	 - \,g\,dt
	 \,,
	 \\
	 \rmd D + \mathcal{L}_{\rmd \bs{r}_t}D & =0
	 \,.
\end{split}
	\label{Stochastic-WW-eqns}
\end{align}

As in the deterministic case, these equations imply a Kelvin-Noether theorem as follows
    \begin{align}
\begin{split}
    \rmd \oint_{c(\rmd\bs{r}_t)} (\bs{\wh{v}}\cdot d\bs{r} + \widehat{w}d\zeta) \cdot d\bs{r}
    &= \oint_{c(\rmd\bs{r}_t)}(\rmd +\mathcal{L}_{\rmd\bs{r}_t})
    \big(\bs{\wh{v}}\cdot d\bs{r} + \widehat{w}d\zeta\big)
    \\
    &= \oint_{c(\rmd\bs{r}_t)} (\rmd +\mathcal{L}_{\rmd\bs{r}_t})(\bs{V}\cdot d\bs{r})
    = \oint_{c(\rmd\bs{r}_t)} (\rmd +\mathcal{L}_{\rmd\bs{r}_t})d\widehat{\phi}
    = \oint_{c(\rmd\bs{r}_t)} d\varpi\,dt = 0\,.
\end{split}
	\label{EP-Stochastic-WW-circ-thm}
\end{align}

%%%%%%%%%%%%%%%%%%%%%%%%%%%%%%%%%%%%%%%%%%%%%%%%%%%%%%%%%%%%%%%%%%%%%%%%%%%%%%
\subsection{Stochastic WCIFS equations}\label{subsec:Stochastic_WCIFS}

Similarly to the stochastic ECWW equations, we may define stochastic versions of any of the wave-current models we have derived, including the MCWW equations. Here we will demonstrate this for our most complete wave-current model corresponding to the action integral \eqref{ActionIntegral-3}. We may again define the equivalent action integral featuring SALT in order to derive the corresponding stochastic system of equations. 

In the stochastic case, in order to couple the waves and currents we consider the insertion of the stochastic vector field $\rmd x_3 \nabla_{\bs{r}}\zeta$, where $\rmd x_3 = \wh{w}\,dt + \sum_i \xi_{3i} \circ dW_t^i$, into the 1-form $\widetilde\mu\nabla_{\bs{r}}\zeta$. 

The stochastic version of the action integral \eqref{ActionIntegral-3} is therefore given by
\begin{align}
\begin{split}
	S = \int \int & D\rho\left(\frac{1}{2}\big(|\bs{\wh v}|^2+{\widehat w}^2\big)  - g\zeta\right)\,dt   
	- \rmd p(D-1)+ \widehat{\phi}\big(\rmd D + \mathcal{L}_{\rmd\bs{r}_t}D\big) 
	+ \gamma(\rmd \rho + \mathcal{L}_{\rmd\bs{r}_t}\rho)
	\\
	& \widetilde{\mu}\Big( \rmd\zeta  + \mathcal{L}_{\rmd\bs{r}_t}\zeta - {\widehat w} \big(1+\epsilon|\nabla_{\bs{r}}\zeta|^2d\zeta\big)\,dt - \sum_i\xi_{3i}\big(1+\epsilon|\nabla_{\bs{r}}\zeta|^2\big)\circ dW_t^i \Big) \,dxdy\,.
\end{split}
	\label{ActionIntegral-3SALT}
\end{align}

Similar to the application of Hamilton's principle to \eqref{ActionIntegral-3}, variations of \eqref{ActionIntegral-3SALT} are given by
\begin{align}
\begin{split}
	\delta \bs{\wh v}:&\quad 
	D\rho \bs{\wh v}\cdot d\bs{r} + \widetilde{\mu}\,d\zeta =  Dd\widehat{\phi} - {\gamma}d\rho
	\,,\\ 
	\delta \widehat w :&\quad 
	D\rho \widehat w - \widetilde{\mu}\,\big(1+\epsilon|\nabla_{\bs{r}}\zeta|^2\big) =  0
	\,,\\
	\delta \widetilde{\mu}:&\quad 
	 \rmd \zeta  + \mathcal{L}_{\rmd\bs{r}_t}\zeta = \rmd x_3\big(1+\epsilon|\nabla_{\bs{r}}\zeta|^2\big) = {\widehat w}(1+\epsilon|\nabla_{\bs{r}}\zeta|^2)\,dt + \sum_i\xi_{3i}\big(1+\epsilon|\nabla_{\bs{r}}\zeta|^2\big)\circ dW_t^i
	\,,\\
	\delta \zeta:&\quad 
	\rmd \widetilde{\mu} + \mathcal{L}_{\rmd\bs{r}_t}\widetilde{\mu}
	 = - \,D\rho g\,dt + 2\epsilon\,{\rm div}_{\bs{r}} \big( {\rmd x_3} \, \widetilde{\mu}\, \nabla_{\bs{r}} \zeta\big)
	\\
	&\hspace{23mm}=
	 - \,D\rho g\,dt + 2\epsilon\,{\rm div}_{\bs{r}} \big( {\widehat w} \, \widetilde{\mu}\, \nabla_{\bs{r}} \zeta\big)\,dt + \sum_i2\epsilon\,{\rm div}_{\bs{r}} \big( \xi_{3i} \, \widetilde{\mu}\, \nabla_{\bs{r}} \zeta\big)\circ dW_t^i
	\,,\\
	\delta\rho :&\quad 	
	\big(\rmd + \mathcal{L}_{\rmd\bs{r}_t}\big)\left(\frac{\gamma}{D}\right) 
	= \frac12\big(|\bs{\wh v}|^2+{\widehat w}^2\big)\,dt - g\zeta\,dt  =: {\varpi}\,dt
	\,,\\
	\delta D:&\quad 
	\big(\rmd + \mathcal{L}_{\rmd\bs{r}_t}\big)\widehat{\phi} = \rho{\varpi}\,dt - \rmd p 
	\,,\\
	\delta \widehat{\phi} :& \quad 
	\rmd D + \mathcal{L}_{\rmd\bs{r}_t} D =0 
	\,,\\
	\delta p:&\quad 
	D-1 = 0 
	\quad \implies \text{div}_{\bs{r}}\bs{\wh{v}} =0 
	\,,\\
	\delta\gamma :&\quad 
	\big(\rmd + \mathcal{L}_{\rmd\bs{r}_t}\big)\rho=0
	\,.
\end{split}
	\label{var-eqns-3SALT}
\end{align}

We apply the operator $\rmd + \mathcal{L}_{\rmd \bs{r}_t}$ to the first line in \eqref{var-eqns-3SALT}
to find,

\begin{align}
\begin{split}
	D\rho\big(\rmd + \mathcal{L}_{\rmd\bs{r}_t}\big)(\bs{\wh v}\cdot d\bs{r}) 
	&= Dd\big(\rmd + \mathcal{L}_{\rmd\bs{r}_t}\big)\widehat{\phi} 
	- \big(\rmd + \mathcal{L}_{\rmd\bs{r}_t}\big)\gamma d\rho 
	-\big(\rmd + \mathcal{L}_{\rmd\bs{r}_t}\big)(\widetilde\mu d \zeta) \\
	&= D\rho d{\varpi}\,dt - D d(\rmd p) -\widetilde\mu d\bigg(\rmd x_3\big(1+\epsilon|\nabla_{\bs{r}}\zeta|^2\big) \bigg)+ D\rho g\,d\zeta\,dt \\
	&\qquad - 2\epsilon\,{\rm div}\left(\rmd x_3 \, \widetilde{\mu}\, \nabla_{\bs{r}} \zeta \right)d\zeta\,,
\end{split}
\label{EulerPoincare-3SALT}
\end{align}
and thus
\begin{align*}
    (\rmd t + \mathcal{L}_{\rmd\bs{r}_t})(\bs{\wh v}\cdot d\bs{r}) 
	&= -\frac{1}{\rho}d(\rmd p) + \bigg[d\frac{|\bs{\wh v}|^2}{2}  	
	 - \frac12\widetilde{w}^2 d \big(1+\epsilon|\nabla_{\bs{r}}\zeta|^2\big)^2 \bigg] dt
	 \\&\hspace{1cm} - \frac{2\epsilon}{D\rho}\,{\rm div}_{\bs{r}}  
	\left( \rmd x_3 \, \widetilde{\mu}\, \nabla_{\bs{r}} \zeta \right) d\zeta
	\\&\hspace{1cm}
	-\widetilde{w} \sum_i d\left( \xi_{3i}(1+\epsilon|\nabla_{\bs{r}}\zeta|^2) \right) \circ dW_t^i
	\,.
\end{align*}

\begin{remark}[A stochastic Kelvin-Noether theorem]
    We have, from calculations analogous to the deterministic case performed similarly to the above, a stochastic version of the Theorem \ref{KNthm-CFS}. In the stochastic case, this takes the form:
    \begin{align}
        \rmd \oint_{c(\rmd\bs{r}_t)} \Big(\bs{\wh v}\cdot d\bs{x} + \frac{\widetilde\mu}{D\rho}  d \zeta\Big) 
        =
        \oint_{c(\rmd\bs{r}_t)} d{\varpi}\,dt - \frac1\rho d\rmd p
        \,.
    \end{align}
\end{remark}

We may collect the WCIFS SALT equations of motion in \eqref{var-eqns-3SALT}, as follows
\begin{align}
\begin{split}
\rmd \bs{\wh v} + ({\rmd \bs{r}_t}\cdot\nabla_{\bs{r}})\bs{\wh v} 
+ (\nabla_{\bs{r}} {\rmd \bs{r}_t})^T\cdot \bs{\wh v} 
= \nabla_{\bs{r}}\frac{|\bs{\wh v}|^2}{2}\,dt &- \frac1\rho \nabla_{\bs{r}} \rmd p - \frac12\widetilde{w}^2 \nabla_{\bs{r}} \big(1+\epsilon|\nabla_{\bs{r}}\zeta|^2\big)^2\,dt  \\
& - \frac{2\epsilon}{D\rho}\,{\rm div}  \left( {\rmd x_3} \, \widetilde{\mu}\, \nabla_{\bs{r}} \zeta \right)\nabla_{\bs{r}} \zeta \\
&-\widetilde{w} \sum_i \nabla_{\bs{r}}\left( \xi_{3i}(1+\epsilon|\nabla_{\bs{r}}\zeta|^2) \right) \circ dW_t^i \,,
\\
(\rmd + \mathcal{L}_{\rmd \bs{r}_t})(D\,d^2x) &=0
\,,
\\
(\rmd + \mathcal{L}_{\rmd \bs{r}_t})\rho &=0\,,
\\
(\rmd + \mathcal{L}_{\rmd \bs{r}_t})\zeta &=\rmd  x_3\big(1+\epsilon|\nabla_{\bs{r}}\zeta|^2\big)\,,
\\
(\rmd+\mathcal{L}_{\rmd\bs{r}_t})(\widetilde\mu\,d^2x)
	&=
	 \Big(- D\rho g\,dt 
	+ 2\epsilon\,{\rm div}_{\bs{r}} \big( {\rmd x_3} \, \widetilde{\mu}\, \nabla_{\bs{r}} \zeta\big) \Big)\,d^2x
	%+
	%2\epsilon\,{\rm div}  
	%\left( \frac{1}{D\rho}\widetilde\mu^2(1+\epsilon|\nabla\zeta|^2)\nabla%\zeta \right)\,dt
	%\\&\quad +\sum_i2\epsilon\,{\rm div} \big( \xi_{3i} \, %\widetilde{\mu}\, \nabla \zeta\big)\circ dW_t^i
	\,.
\end{split}
\label{EulerPoincare-4SALT}
\end{align}
The properties of these equations will be studied in detail, elsewhere. 

\section{Analytical remarks about variational water-wave models}\label{sec: Stoch analysis}

Recall the equations in (\ref{var-eqns-FS-CWW}), found by varying the action
integral (\ref{ActionIntegral-FS}). These equations may be written in the
form

\begin{equation*}
\begin{split}
\partial _{t}\widehat{\phi }+\boldsymbol{\widehat{v}}\cdot \nabla _{%
\boldsymbol{r}}\widehat{\phi }& =\frac{1}{2}\big(|\boldsymbol{\widehat{v}}%
|^{2}+\widehat{w}^{2}\big)-g\zeta \,, \\
\partial _{t}\zeta +\boldsymbol{\widehat{v}}\cdot \nabla _{\boldsymbol{r}%
}\zeta & =\widehat{w}\,, \\
\partial _{t}\widehat{w}+\boldsymbol{\widehat{v}}\cdot \nabla _{\boldsymbol{r%
}}\widehat{w}& =-\,g \\
\partial _{t}D+\text{div}_{\boldsymbol{r}}(D\boldsymbol{\widehat{v}})& =0\,,
\end{split}%
\end{equation*}%
where 
\begin{align}
\boldsymbol{\widehat{v}}& =\boldsymbol{V}-\widehat{w}\nabla _{\boldsymbol{r}%
}\zeta  \notag \\
& =\nabla _{\boldsymbol{r}}\widehat{\phi }-\widehat{w}\nabla _{\boldsymbol{r}%
}\zeta \,,\quad \hbox{in the irrotational case}\,.  \label{eq: v-hatV}
\end{align}%
Recall that the transport velocity $\boldsymbol{\widehat{v}}$ evolves
according to (\ref{WW-Pressureless2DEuler}), i.e., 
\begin{equation}
{\partial }_{t}\boldsymbol{\widehat{v}}+\boldsymbol{\widehat{v}}\cdot \nabla 
\boldsymbol{\widehat{v}}=0\,,  \label{WW-Pressureless2DEulerb}
\end{equation}%
Of course (\ref{WW-Pressureless2DEulerb}) can be identified as the 2
dimensional inviscid Burger's equation. Under certain conditions on the
initial velocity $\boldsymbol{\widehat{v}}_{0}$, it has unique solution
(possibly only local in time). We sketch below the classical argument for
showing this using the method of characteristics. Define the characteristic
equation given by 
\begin{equation}
\begin{cases}
\frac{d\boldsymbol{r}_{t}}{dt}(\boldsymbol{r})=\boldsymbol{\widehat{v}}(%
\boldsymbol{r}_{t}(\boldsymbol{r}),t)\,,\quad t>0\,, \\ 
\boldsymbol{r}_{t}(\boldsymbol{r})=\boldsymbol{r}\,.%
\end{cases}
\label{IVP-Characteristic}
\end{equation}%
Provided $\boldsymbol{\widehat{v}}$ is sufficiently smooth, the system (\ref%
{IVP-Characteristic}) will have a unique solution. Moreover, from (\ref%
{IVP-Characteristic}) and (\ref{WW-Pressureless2DEulerb}) we deduce, by the
chain rule, that 
\begin{equation*}
\frac{{\partial }}{{\partial }t}\left[ \boldsymbol{\widehat{v}}(\boldsymbol{r%
}_{t}(\boldsymbol{r}),t)\right] =\frac{{\partial }\boldsymbol{\widehat{v}}}{{%
\partial }t}(\boldsymbol{r}_{t}(\boldsymbol{r}),t)+\frac{d\boldsymbol{r}_{t}%
}{dt}(\boldsymbol{r})\cdot \nabla _{\boldsymbol{r}}\boldsymbol{\widehat{v}}(%
\boldsymbol{r}_{t}(\boldsymbol{r}),t)=\left( {\partial }_{t}\boldsymbol{%
\widehat{v}}+\boldsymbol{\widehat{v}}\cdot \nabla _{\boldsymbol{r}}%
\boldsymbol{\widehat{v}}\right) (\boldsymbol{r}_{t}(\boldsymbol{r}),t)=0\,,
\end{equation*}%
so $\boldsymbol{\widehat{v}}=(\wh{v}_1,\wh{v}_2)$ is constant along the characteristics. Thus
the characteristic curves corresponding to (\ref{WW-Pressureless2DEulerb})
are straight lines determined by the initial conditions, given by 
\begin{equation}
\bs{r}_t =
\varphi_{t}\bs{r} \coloneqq \bs{r}
+ \bs{\widehat{v}}_{0}(\boldsymbol{r})t
\,,  \label{straightlines}
\end{equation}
and therefore the follow pull-back relation holds,
\begin{equation}
{\varphi_t}^*\bs{\wh{v}}_t(\bs{r})
:=
\boldsymbol{\widehat{v}}_t({\varphi_t}\bs{r})
=\boldsymbol{\widehat{v}}_{0}(\boldsymbol{r})  \label{constant}
\,,\end{equation}
so that
\begin{equation}
0 = \frac{d \bs{\wh{v}}_0}{dt} = 
\frac{d}{dt}\varphi_t^*\bs{\wh{v}}_t(\bs{r}) 
=
\varphi_t^*\left(\p_t\bs{\wh{v}}_t (\bs{r}) 
+  \frac{\p \bs{\wh{v}}_t(\bs{r}) }{\p\bs{r}}\cdot \bs{\wh{v}}_t(\bs{r})
\right)\,.
%%
%&= 
%\p_t\bs{\wh{v}}_t (\bs{r}_t) +  \frac{\p \bs{\wh{v}}_t}{\p\bs{r}_t}\cdot \frac{d\bs{r}_t}{dt} 
%\\&=
%\p_t\bs{\wh{v}}_t (\bs{r}_t) +  \frac{\p \bs{\wh{v}}_t(\bs{r}_t) }{\p\bs{r}_t}\cdot \bs{\wh{v}}_0(\boldsymbol{r})
%\\&=
%\varphi_t^*\left(\p_t\bs{\wh{v}}_t (\bs{r}) 
%+  \frac{\p \bs{\wh{v}}_t(\bs{r}) }{\p\bs{r}}\cdot \bs{\wh{v}}_0({\varphi_t}_*\bs{r})
%\right)
%\\&=
%\varphi_t^*\left(\p_t\bs{\wh{v}}_t (\bs{r}) 
%+  \frac{\p \bs{\wh{v}}_t(\bs{r}) }{\p\bs{r}}\cdot \bs{\wh{v}}_t(\bs{r})
%\right)
\end{equation}
Equations (\ref{straightlines}) and (\ref{constant}) enable us to give an
explicit description of the (classical) solution of (\ref%
{WW-Pressureless2DEulerb}) up to first time $\tau $ at which the characteristic lines cross. The time $\tau $ is the first time
when $\nabla _{\boldsymbol{r}}\left(\varphi_t\bs{r}\right)$ degenerates, in other words the Jacobian of $\varphi_t\bs{r}$ has determinant equal to 0.
Note that equation (\ref{WW-Pressureless2DEulerb}) may have a weak solution beyond $\tau$. The time $\tau$ can be explicitly described in term of the
eigenvalues of the Jacobian $\nabla _{\boldsymbol{r}}\boldsymbol{\widehat{v}}%
_{0}$ of the initial velocity $\boldsymbol{\widehat{v}}_{0}$. We will denote
by $\lambda _{i}(\boldsymbol{r})$, $i=1,2$, the two eigenvalues of $\nabla _{%
\boldsymbol{r}}\boldsymbol{\widehat{v}}_{0}\left( \boldsymbol{r}\right) $.
We introduce the (possibly empty) subset $S$ of the fluid domain $\Omega$ (which we assume to be a closed bounded set
of $\mathbb{R}^{2}$) defined as%
\begin{equation*}
S=\{\boldsymbol{r}\in \Omega:\lambda _{1}(\boldsymbol{r})<0\,,\quad %
\hbox{or}\quad \lambda _{2}(\boldsymbol{r})<0\}.\,
\end{equation*}%
Define $\tau :=\infty $ if $S$ is the empty set, otherwise $\tau =-1/\bar{%
\lambda}$ where\footnote{%
Note that the set $\Omega$ is compact, therefore the two minima $%
\min_{\boldsymbol{r}\in S}\lambda _{1}(\boldsymbol{r})$ as well as $\min_{%
\boldsymbol{r}\in S}\lambda _{2}(\boldsymbol{r})$ are well defined.} 
\begin{equation*}
\bar{\lambda}=\min \left\{ \min_{\boldsymbol{r}\in S}\lambda _{1}(%
\boldsymbol{r}),\min_{\boldsymbol{r}\in S}\lambda _{2}(\boldsymbol{r}%
)\right\}\,,
\end{equation*}
then $\varphi_t:\Omega\rightarrow \Omega$ is a diffeomorphism for any $t\in \lbrack 0,\tau
)$. This statement is immediate after observing that the eigenvalues of $%
\nabla _{\bs{r}}\left(\varphi_t\bs{r}\right)$ are given by $(1+t\lambda _{1}(\bs{r}),1+t\lambda _{2}(\bs{r}))$ and cannot become zero before time $\tau$. In other words, the Lagrangian flow is well defined and
differentiable in both the spatial and temporal variable. Let $\varphi_t^{-1}:\Omega\rightarrow\Omega$ be the inverse of the Lagrangian flow. From (\ref{constant}) we deduce that (\ref{WW-Pressureless2DEulerb}) has a (unique) solution given by the push-forward of the initial velocity ${\varphi_t}_*\bs{\widehat{v}}_{0}(\bs{r})$, i.e.,
\begin{equation*}
\bs{\wh{v}}(\bs{r},t)=\bs{\wh{v}}_{0}(\varphi_t^{-1}\bs{r} )
=: {\varphi_t}_*\boldsymbol{\widehat{v}}_{0}(\boldsymbol{r}),
~~~\boldsymbol{r}\in 
\Omega_{t}\text{, }t\in \lbrack 0,\tau )\text{.}
\end{equation*}%
From (\ref{var-eqns-FS-CWW}), we also deduce that

\begin{equation}
    \begin{aligned}
        \frac{d}{dt}\wh{\phi}(\varphi_t\bs{r},t) &= e(\varphi_t\bs{r},t) \,,\\
        \frac{d}{dt}\zeta (\varphi_t\bs{r},t) &=\wh{w}(\varphi_t\bs{r},t) \,,\\
        \frac{d}{dt}\wh{w}(\varphi_t\bs{r},t) &=-g
        \,,\\
        \frac{d}{dt}\big(D(\varphi_t\bs{r},t)&d^2(\varphi_t\bs{r})\big) = 0\,,
    \end{aligned}
\end{equation}
where $e(\varphi_t\bs{r},t)=\frac{1}{2}\left( \big(|\boldsymbol{\widehat{v}}%
|^{2}+\widehat{w}^{2}\big)-g\zeta \right) \left( \varphi_t\bs{r},t\right) $. It follows that 
the solution for the full system of variables is obtained by integrations 
along the characteristics, as
\begin{equation}
    \begin{aligned}
        \wh{\phi}(\boldsymbol{r},t) &=\wh{\phi}(\varphi_t^{-1}\bs{r},0)+\int_{0}^{t}e(\varphi_{s-t}\bs{r} ,t)\,ds\,,\quad\hbox{where}\quad \varphi_s\varphi_t^{-1} = \varphi_{s-t}\,,\\
        \wh{w}(\boldsymbol{r},t) &=\wh{w}_{0}(\varphi_t^{-1}\bs{r} )-gt\,, \\
        \zeta (\bs{r},t) &=\zeta (\varphi_t^{-1}\bs{r} ,0)+\int_{0}^{t}\wh{w}(\varphi_{s-t}\bs{r},t)\,ds\,, \\
        D(\bs{r},t)d^2\bs{r} 
        &= {\varphi_t}_* \big( D(\bs{r},0)d^2\bs{r}\big)
        = D(\varphi_t^{-1}\bs{r},0)d^2(\varphi_t^{-1}\bs{r})
        \,,\quad\hbox{since}\quad \varphi_0 = \rm{Id}\,.
    \end{aligned}
\end{equation}
Thus, the explicit solution for $\bs{\wh v}$ corresponding to the characteristics also provides an explicit solution for the full system of variables. 

\paragraph{An Eulerian approach.}
For an alternative Eulerian approach, notice that taking the 
curl$_{\boldsymbol{r}}$ of the evolution equation (%
\ref{WW-Pressureless2DEuler}) for the transport velocity $\boldsymbol{%
\widehat{v}}$ defined in equation (\ref{eq: v-hatV}) implies that ${\widehat{%
\omega }}=\mathrm{curl}_{\boldsymbol{r}}\,\boldsymbol{\widehat{v}}=J({%
\widehat{w}},\zeta )$ satisfies, 
\begin{equation}
{\partial }_{t}\widehat{\omega }+\boldsymbol{\widehat{v}}\cdot \nabla _{%
\boldsymbol{r}}\widehat{\omega }=0\,.  \label{eq: wavevorticity}
\end{equation}%
Consequently, the continuity equation for $D$ implies 
\begin{equation}
{\partial }_{t}(D\widehat{\omega })+\mathrm{div}_{\boldsymbol{r}}(D\widehat{%
\omega }\boldsymbol{\widehat{v}})=0\,,  \label{eq: wavePV}
\end{equation}%
and the volume integral $\int D\widehat{\omega }\,d^{2}r$ is preserved for
tangential boundary conditions on $\boldsymbol{\widehat{v}}$. Another way of
writing (\ref{eq: wavePV}) is 
\begin{equation}
({\partial }_{t}+\mathcal{L}_{\boldsymbol{\widehat{v}}})(Dd\widehat{w}\wedge
d\zeta )=0\,.  \label{eq: wavePV2}
\end{equation}%
Explicit solutions for (\ref{eq: wavevorticity}), (\ref{eq: wavePV}) and (
\ref{eq: wavePV2}) can be deduced in a similar manner as above. 

An alternative approach, possibly to show well-posedness in more general
spaces is to attempt an analysis based on energy estimates. For this one
needs to analyze the pair of equations 
\begin{eqnarray*}
{\partial }_{t}\boldsymbol{\widehat{v}}+\boldsymbol{\widehat{v}}\cdot \nabla 
\boldsymbol{\widehat{v}} &=&0\,\, \\
\partial _{t}D_{t}+\text{div}_{\boldsymbol{r}}(D_{t}\boldsymbol{\widehat{v}}%
) &=&0\,
\end{eqnarray*}%
on the domain $\Omega$ of the measure $D_{t}d^{2}r$. For this we
can use as apriori estimates the conserved energy for this system given by 
\begin{equation*}
\begin{split}
h(\boldsymbol{M},D,\lambda ,\zeta )& =\int \boldsymbol{M}\cdot \widehat{%
\boldsymbol{v}}+\lambda {\partial }_{t}\zeta \,d^{2}r-\ell (\boldsymbol{%
\widehat{v}},D,\widehat{\phi },{\widehat{w}},\zeta ;\lambda ) \\
& =\int \frac{1}{2D}\big|\boldsymbol{M}-\lambda \nabla _{\boldsymbol{r}%
}\zeta \big|^{2}+\frac{\lambda ^{2}}{2D}+gD\zeta \,d^{2}r \\
& =\int \left( \frac{1}{2}|\widehat{\boldsymbol{v}}|^{2}+\frac{1}{2}|%
\widehat{w}|^{2}+g\zeta \right) \,D\,d^{2}r \\
& =\int \left( \frac{1}{2}|\widehat{\nabla _{\boldsymbol{r}}\phi }|^{2}+%
\frac{1}{2}|\widehat{w}|^{2}+g\zeta \right) \,D\,d^{2}r\,.
\end{split}%
\end{equation*}%
as well as Sobolev norm estimates deduced from (\ref{var-eqns-FS-CWW}). For
this we introduce the (time-dependent) $L^{p}$ norm of some function $f$
with respect to the measure $Dd^{2}r$ as 
\begin{equation*}
\Vert f\Vert _{D_{t},p}=\left( \int_{\Omega}f^{p}D_{t}d^{2}r\right)
^{1/p}\,.
\end{equation*}%
and can show that $\Vert \boldsymbol{\widehat{v}}\Vert _{D,p}$ is conserved.
Moreover, via a standard Gr\"{o}nwall/Young inequality argument one shows
that $\Vert \widehat{\phi }\Vert _{D,p},$ $\Vert \zeta \Vert _{D,p}$, $\Vert 
\widehat{w}\Vert _{D,p}$ and $\Vert D\Vert _{D,2}$ are controlled. Controls on higher
Sobolev norms are also possible. 

For existence, one follows DiPerna and Lions \cite{DiPerna1989}, to define a
sequence $(\boldsymbol{\widehat{v}}^{n},D_{t}^{n})_{n\in \mathbb{N}}$ by 
\begin{eqnarray*}
{\partial }_{t}\boldsymbol{\widehat{v}}^{n}+\boldsymbol{\widehat{v}}%
^{n-1}\cdot \nabla _{\boldsymbol{r}}\boldsymbol{\widehat{v}}^{n} &=&0\, \\
\partial _{t}D_{t}^{n}+\text{div}_{\boldsymbol{r}}(D_{t}^{n}\boldsymbol{%
\widehat{v}}^{n}) &=&0
\end{eqnarray*}%
For each $\boldsymbol{\widehat{v}}^{n-1}$, we may apply the results from \cite%
{DiPerna1989} (Theorem III.2) to prove existence of $\boldsymbol{\widehat{v}}%
^{n}$. We would need to prove that $\boldsymbol{\widehat{v}}^{n}$ satisfies
the relevant bounds to allow us to apply the theorem again and iterate this
process. We then show that the sequence is relatively compact in a suitably chosen Sobolev space. 

\begin{remark}
    According to remark \ref{subsec:CWWEPressure}, when considering the ECWWE equations with pressure as in section \ref{subsec:CWWEPressure}, the transport velocity evolves according to the 2D Euler equation and the system thus inherits the analytical properties of this equation.
\end{remark}

\section{Future work}\label{sec: Conclude}
The augmented water-wave problem which has been introduced here opens doors for new analytical results, as well as interesting directions for numerical studies of wave-current interaction. For example: 
\begin{enumerate}
\item
Analytical properties of the wave-current interaction equations introduced here are unknown, and the extended version of the classical water-wave equations opens the door for new analytical results for CWWE.
\item
The incorporation of surface tension into this framework is a potentially interesting issue. Other physical approximations commonly made to derive other well known water-wave equations (KdV, KP, etc.) may also be considered within this framework.
\item
Further study of the stochastic ECWW equations in section \ref{sec: Stoch WCIFS eqns} and their stochastic ACWW and WCI FS versions would be worthwhile. In particular, section \ref{sec: Stoch WCIFS eqns} introduces a stochastic version of the well-studied classical water-wave model which has been stochastically perturbed in a way which preserves many of its desirable fluid dynamics properties. In particular, this model would allow us to consider a stochastic CWW theory based on a Dirichlet-Neumann operator for 3D irrotational SALT flows. Introducing noise into the parameter $\epsilon$ in the ACWW and WCI FS models may also be interesting, since doing so would enable investigations of the probabilistic nature of wave-current interactions using stochastic versions of wave-current minimal coupling (WCMC).
\end{enumerate}
We expect to pursue all of these research directions in future work. 

\section*{Acknowledgements}
During this work, O. D. Street has been supported by an EPSRC studentship as part of the Centre for Doctoral Training in the Mathematics of Planet Earth (grant number EP/L016613/1), and both D. D. Holm and D. Crisan have been partially supported by European Research Council (ERC) Synergy grant STUOD - DLV-856408.

We would like to thank our friends and colleagues for their encouraging comments, especially S. Patching, E. Luesink, R. Hu, E. Titi, J. C. McWilliams, and B. Fox-Kemper.

Last, but certainly not least, we would like to place on record our sincere thanks to H. Segur, H. Dumpty and P. Box, as always, for providing wise and light-hearted guidance of our expectations during the course of our mathematical investigations. 

\appendix
\section{Transformation theory for fluid dynamics -- Kelvin theorem}\label{app: FluidTransTheory}

The Kelvin-Noether theorem is the statement of Newton's Law for fluid mass distributed on a material loop.
\begin{align}
\frac{d}{dt}\hspace{-5mm}\oint_{\hspace{5mm}c_t={\phi_t }c_0}\hspace{-5mm} 
\mathbf{v}( t,\mathbf{x})\cdot d \mathbf{x}
= \oint_{c_t }
\underbrace{\ 
\mathbf{f} \cdot d \mathbf{x}\  
}_{\sf Newton's\ Law}
\,.
\label{KN-state}
\end{align}
For a discussion of the geometric mechanics underlying the deterministic case, see, e.g., \cite{HMR1998}. 
For a discussion of the geometric mechanics underlying the stochastic case, see, e.g., \cite{Holm2015,Holm2020}.\medskip

\begin{proof}
The \emph{deterministic} Kelvin-Noether theorem may be proved, as follows. 
Consider a closed loop $c_t $ moving with the material flow as \[c_t = \phi_t c_0\,.\]
The Eulerian velocity of the loop is \[\frac{d}{dt}\phi_t(x) = \phi_t^*u(t,x) = u(t,\phi_t (x))\,.\]
This equation illustrates the operation of ``pull-back'' $\phi_t^*u(t,x) $ of the Eulerian fluid velocity $u(t,x)$ 
by the material flow map $\phi_t $.

Compute the time derivative of the integral of the momentum/mass (impulse) around a time-dependent  loop $c_t={\phi_t }c_0$ moving with the flow map, $\phi_t $, as

\begin{align}
\begin{split}
\frac{d}{dt}\hspace{-5mm}\oint_{\hspace{5mm}c_t={\phi_t }c_0}\hspace{-5mm} 
\mathbf{v}( t,\mathbf{x})\cdot d \mathbf{x}
&= \oint_{c_0} \frac{d}{dt}\Big(\phi_t^*\big(\mathbf{v}( t,\mathbf{x}) \cdot d \mathbf{x}\big) \Big)
\\&=  \oint_{ c_0} 
\underbrace{
 \phi_t^*\Big(( \partial_t+ \mathcal{L}_{u(t,\mathbf{x})}) (\mathbf{v} \cdot d \mathbf{x})\Big)
}_{\sf Lie\ derivative\ defined\ via\ chain\ rule}
\\&=  \oint_{\phi_t c_0=c_t} 
( \partial_t+ \mathcal{L}_{u(t,\mathbf{x})}) (\mathbf{v} \cdot d \mathbf{x})
\\&= \oint_{c_t }
\underbrace{\ 
\mathbf{f} \cdot d \mathbf{x}\  
}_{\sf Newton's\ Law}
{=}  \oint_{ c_0}  \phi_t^*\Big(\underbrace{\ {\mathbf{f} \cdot d \mathbf{x}}\ }_{\sf {Motion\ eqn}}\Big)
\end{split}
\label{KN-proof}
\end{align}
This is the Kelvin-Noether theorem of  \cite{HMR1998}. When the covector field $\mathbf{v}( t,\mathbf{x})$ is interpreted as the momentum per unit mass in the fixed Eulerian inertial frame, then the last line states Newton's Law for fluid mass distributed on a material loop. When the covector field $\mathbf{f}( t,\mathbf{x})=-\nabla p$ is a pressure-gradient force per unit mass in the Eulerian inertial frame, then the last line states Kelvin's theorem for the conservation of circulation in ideal Euler fluid dynamics with spatially homogeneous density.

\end{proof}

Let us delve more deeply into the statement in the second line of the proof of the Kelvin-Noether theorem that ``the Lie derivative is defined via the chain rule". More specifically, the Lie derivative is defined by the time derivative of the pull-back $\phi_t^*$ of the flow map $\phi_t$ acting on the circulation integrand (which is a 1-form) by using the chain rule. The pull-back is also used in the discussion of the Burgers equation in section \ref{sec: Stoch analysis}. Let's do the corresponding calculation for the Kelvin-Noether theorem. 

Integration in time of the pull-back relation in the proof,
\[
\frac{d}{dt}\phi_t(x) = \phi_t^*u(t,x) = u(t,\phi_t (x))
\,,\] 
yields the smooth invertible map, $\phi_t\in {\rm Diff}(M)$, by integration of the characteristic curves of the smooth time-dependent vector field $u_t\in\mathfrak{X}(M)$ acting on smooth functions $f\in C^\infty(M)$ defined on a smooth manifold, $M$. In this situation, one says that the map $\phi_t$ is \emph{generated} by the vector field $u_t$. The pull-back relation can be written equivalently, as a push-forward, denoted as
\[
u_t = {\phi_t}_*\dot{\phi_t} = {\phi_t^{-1}}^*\dot{\phi_t} = \dot{\phi_t}\phi_t^{-1}
\,,\] 
in which the operation of push-forward of a smooth function $f$ by a smooth invertible map $\phi_t$ depending on a parameter $t$ is defined as the inverse of the pull-back, which may be written as ${\phi_t}_*={\phi_t^{-1}}^*$.

We may now understand the first step in the proof above as the change of variables in the loop integral to transform the loop $c_t={\phi_t }c_0$ moving under the flow map $\phi_t$ in the fixed frame, into a fixed loop ${\phi_t }^{-1}c_t=c_0$ in the moving frame of the flow map; while also transforming the integrand $(\mathbf{v}( t,\mathbf{x}) \cdot d \mathbf{x})$ in the loop integral from the fixed frame into the moving frame of the flow map. This transformation of the Kelvin circulation loop integral into the frame in which the moving loop is fixed allows the time derivative to commute with integration around the loop. Consequently, the time derivative comes inside the integral to act on the transforned integrand, which is now in the moving frame of the flow map $\phi_t$. 

The second step in the proof above defines the Lie derivative as \cite{HMR1998}
\begin{align}
\begin{split}
\frac{d}{dt}\Big(\phi_t^*\big(\mathbf{v}( t,\mathbf{x}) \cdot d \mathbf{x}\big) \Big)
&=: 
 \phi_t^*\Big(\big( \partial_t+ \mathcal{L}_{\dot{\phi_t}\phi_t^{-1}}\big) (\mathbf{v} \cdot d \mathbf{x})\Big)
\\&=: 
 \phi_t^*\Big(\big( \partial_t+ \mathcal{L}_{u(t,\mathbf{x})}\big) (\mathbf{v} \cdot d \mathbf{x})\Big)
\\&=: 
 \phi_t^*\bigg(\Big( \partial_t \mathbf{v} + (\mathbf{u}\cdot\nabla)\mathbf{v} 
 + v_j\nabla u^j  \Big)\bigg)
 \cdot d \mathbf{x}
\end{split}
\label{LieDeriv-1form-def}
\end{align}
To finish the proof of the Kelvin-Noether theorem in \eqref{KN-proof}, one transforms the loop integral back into the fixed frame, in which the loop moves with the flow map and the integrand is fixed. 

\begin{remark}
The Lie derivative of a differential k-form has the same expression in any coordinate system, even in a moving coordinate system. In particular, this is true for functions (0-forms), circulation 1-forms and mass density 2-forms in 2D. 
\end{remark}

\paragraph{Lie derivatives in the hat formulation.}
The transformation to the $\wh{f}$ notation in \eqref{hat-notationX} evaluates an arbitrary flow variable $f$ on the free surface,
\begin{align}
\wh{f}(\bs{r},t) = f(\bs{r},z,t)\quad \hbox{on}\quad z = \zeta(\bs{r},t)\,.
	\label{hat-notationAppendix}
\end{align}

For functions (0-forms) the hat-transformation evolves according to 
\begin{align}
\begin{split}
\frac{d}{dt}\phi_t^*\Big(\wh{f}(\bs{r},t)\Big) 
&=
\phi_t^*\Big(\big(\p_t+  \mathcal{L}_{\bs{\wh{v}}}\big)\wh{f}\Big)
= 
\phi_t^*\Big(\p_t\wh{f} +  \bs{\wh{v}}(\bs{r},t)\cdot\nabla_{\bs{r}}\wh{f}\Big)
\\
\Big[\frac{d}{dt}\phi_t^*\Big(f(\bs{x},t)\Big) \Big]_{z=\zeta(\bs{r},t)}
&= 
\Big[\phi_t^*\Big( 
\p_t {f} +  {\,f_z\,}\p_t\zeta
+ \bs{{v}}(\bs{x},t)\cdot
\big(\nabla_{\bs{r}}{f}+ {\,f_z\,}\nabla_{\bs{r}}\zeta(\bs{r},t)\big)\Big)
\Big]_{z=\zeta(\bs{r},t)}
\\&= 
\Big[\phi_t^*\Big( 
\p_t {f} + \bs{{v}}(\bs{x},t)\cdot \nabla_{\bs{r}}{f}
+ {\,f_z\,}\big(\p_t\zeta
+ \bs{{v}}(\bs{x},t)\cdot\nabla_{\bs{r}}\zeta(\bs{r},t)\big)\Big)
\Big]_{z=\zeta(\bs{r},t)}
\\&= 
\Big[\phi_t^*\Big( 
\p_t {f} + \bs{{v}}(\bs{x},t)\cdot \nabla_{\bs{r}}{f}
+ {\,f_z\,}\wh{w}\Big)
\Big]_{z=\zeta(\bs{r},t)}
\\&= 
\Big[\phi_t^*\Big(\big(\p_t+  \mathcal{L}_{\bs{v}}\big){f}\Big)
\Big]_{z=\zeta(\bs{r},t)}
\\&= 
\phi_t^*\Big(\big(\p_t+  \mathcal{L}_{\bs{\wh{v}}}\big)\wh{f}\Big)
\\&= 
\frac{d}{dt}\phi_t^*\Big(\wh{f}(\bs{r},t)\Big) 
\end{split}
\label{LieDeriv-0form}
\end{align}
Thus, equation \eqref{hat-advectionX} is recovered for 0-forms
\[
\big(\p_t+  \mathcal{L}_{\bs{\wh{v}}}\big)\wh{f} 
=
\Big[\big(\p_t+  \mathcal{L}_{\bs{v}}\big){f}\Big]_{z=\zeta(\bs{r},t)}
\]
By the product rule for the pull-back, this calculation also applies to 1-forms and 2-forms, so we have 
\begin{align}
\begin{split}
\frac{d}{dt}\phi_t^*\Big(\bs{\wh{V}}(\bs{r},t)\cdot d\bs{r}\Big) 
&=
\phi_t^*\Big(\big(\p_t+  \mathcal{L}_{\bs{\wh{v}}}\big)\big(\bs{\wh{V}}\cdot d\bs{r}\big)\Big)
\\&= 
\phi_t^*\Big(\big( \p_t\bs{\wh{V}}+  (\bs{\wh{v}}\cdot\nabla_{\bs{r}}\big)\bs{\wh{V}}
+ \wh{V}_j\nabla_{\bs{r}}\wh{v}^j\big)\cdot d\bs{r} \Big)
\\
\Big[\frac{d}{dt}\phi_t^*\Big(\bs{V}(\bs{x},t)\cdot d\bs{x}\Big) \Big]_{z=\zeta(\bs{r},t)}
&= 
\Big[\phi_t^*\Big(\big(\p_t+  \mathcal{L}_{\bs{v}}\big)\big(\bs{V}(\bs{x},t)\cdot d\bs{x}\big)\Big)
\Big]_{z=\zeta(\bs{r},t)}
\\&= 
\Big[\phi_t^*\Big( 
\big(\p_t \bs{V} + \bs{{v}}(\bs{x},t)\cdot \nabla_{\bs{r}}\bs{V}
+ {\,\bs{V}_z\,}\wh{w} + V_j\nabla v^j \big)\cdot d\bs{x} \Big)
\Big]_{z=\zeta(\bs{r},t)}
\\&= 
\phi_t^*\Big(\big(\p_t+  \mathcal{L}_{\bs{\wh{v}}}\big)\big(\bs{\wh{V}}\cdot d\bs{r}\big)\Big)
\\&= 
\frac{d}{dt}\phi_t^*\Big(\bs{\wh{V}}(\bs{r},t)\cdot d\bs{r}\Big) 
\end{split}
\end{align}
Thus, for 1-forms we have a formula which will project the Kelvin theorem onto the free surface. Namely,
\[
\big(\p_t+  \mathcal{L}_{\bs{\wh{v}}}\big)\big(\bs{\wh{V}}\cdot d\bs{r}\big)
=
\Big[\Big(\big(\p_t+  \mathcal{L}_{\bs{v}}\big)\big(\bs{V}(\bs{x},t)\cdot d\bs{x}\big)
\Big]_{z=\zeta(\bs{r},t)}
\]
Finally, for 2-forms we have the continuity equation on the free surface,
\begin{align}
\begin{split}
\Big[\frac{d}{dt}\phi_t^*\big(\wh{\rho}\,d^2r\big)\Big]_{t=0}
&=
\big(\p_t+  \mathcal{L}_{\bs{\wh{v}}}\big)\big(\wh{\rho}\,d^2r\big)
=
\big(\p_t\wh{\rho} + \bs{\wh{v}}\cdot \nabla_{\bs{r}} \wh{\rho} + \wh{\rho}\, \nabla_{\bs{r}}\cdot \bs{\wh{v}}\big)d^2r
\\&=
\big(\p_t\wh{\rho} +  \nabla_{\bs{r}}\cdot (\wh{\rho}\bs{\wh{v}})\big)d^2r
=
0\,.
\end{split}
\label{LieDeriv-contin}
\end{align}

%%%%%%%%%%%%%%%%%%%%%%%%%%%%%%%%%%%%%%%%%%%%%%%

\section{Hamilton's principle for 3D fluid dynamics with a free surface}\label{app:ActionIntBdy-1}

In this appendix, we will re-derive the 3D Euler fluid equations \eqref{motioneqn:componentsX} and \eqref{advectioneqns} from a constrained variational approach for dynamics on a free surface. In this setting, we will be able to continue modelling the free surface equations. 
In particular, in Hamilton's principle we will \emph{constrain} the action integral by applying what we have learned in the present section about the Eulerian equations of irrotational free-surface motion to obtain the ECWWE \eqref{var-eqns-FS} for the two-dimensional velocity fields $\bs{V}$ and $\bs{\wh{v}}$ on the free surface. 

This approach via Hamilton's principle will also enable us to derive equations for fluid dynamic flows on a free surface with vorticity, non-hydrostatic pressure, and spatially varying buoyancy. In the action integral, the wave variables will be regarded as field variables interacting with the fluid variables. After introducing a wave-current ``minimal-coupling'' Ansatz reminiscent of the coupling of a charged fluid to an electromagnetic field \cite{HMR1998}, we will show that this system of equations can be closed and that the wave variables will be able to generate circulation in the fluid. The resulting coupled equations will model a sort of Craik-Leibovich wave-current interaction on the free surface.\bigskip

Consider an action integral defined by
\begin{align}
\begin{split}
	S = \int \int & D\rho\left(\frac{1}{2}|\bs{u}|^2 -gz\right) - p(D-1) - \mu (\partial_t + \bs{u}\cdot\nabla)(\zeta-z) \\
	&\qquad+ \varphi(\partial_t D + \text{div}(D\bs{u})) + \gamma(\partial_t\rho + \bs{u}\cdot\nabla\rho) \,d^3x\,dt
	\,.
\end{split}
	\label{ActionIntegral-1X}
\end{align}
The Lagrange multipliers $\mu$, $\varphi$, and $\gamma$ impose the dynamical constraints in \eqref{advectioneqns} as pioneered in Clebsch \cite{Clebsch1859}. From left to right, the terms in \eqref{ActionIntegral-1X} are: the difference between kinetic and potential energies, the incompressible flow constraint, the Clebsch constraint that the quantity $(\zeta - z)$ is advected (i.e. particles on the surface remain so), and two more Clebsch constraints which impose advection dynamics on $D$ as a density and $\rho$ as a scalar function, respectively.
\begin{remark}
	The action integral in \eqref{ActionIntegral-1X} makes sense physically, as long as the free surface $z=\zeta(\bs{r},t)$ is a graph, so that the magnitude of the elevation slope $|\nabla_{\bs{r}}\zeta|$ remains bounded. Hence, we assume that no wave breaking will occur in the underlying fluid model during the temporal interval of the flow.
\end{remark}

\paragraph{Hamilton's principle.}
Applying Hamilton's principle $\delta S = 0$ to the constrained action integral in \eqref{ActionIntegral-1X} yields the following variations,
\begin{equation*}
\begin{aligned}
	0 = \delta S = \int\int & \delta D \bigg(\rho\left(\frac{1}{2}|\bs{u}|^2 -gz\right)
	- p - (\partial_t\varphi + \bs{u}\cdot \nabla\varphi)\bigg)+ \delta p(D-1) \\
	&\qquad + \delta\rho \bigg( D\left( \frac{1}{2}|\bs{u}|^2-gz \right) - (\partial_t\gamma + \text{div}(\gamma\bs{u}))\bigg) \\
	&\qquad +\delta\bs{u}\cdot(D\rho\bs{u} - \mu\nabla(\zeta-z) - D\nabla\varphi + \gamma\nabla\rho) \\
	&\qquad  + \delta(\zeta-z)(\partial_t\mu + \text{div}(\mu\bs{u})) - (\partial_t(\zeta-z) + \bs{u}\cdot\nabla(\zeta-z))\delta\mu \\
	&\qquad + \delta\varphi(\partial_t D + \text{div}(D\bs{u}))+ \delta\gamma(\partial_t \rho + \bs{u}\cdot\nabla\rho)\,d^3x\,dt\,.
\end{aligned}
\end{equation*}
The variations with respect to each dynamical variable yields the following independent relations, written in the coordinate-free Lie derivative notation discussed in Appendix \ref{app: FluidTransTheory},
\begin{align*}
	\delta D:&\quad (\partial_t + \mathcal{L}_{\bs{u}})\varphi = \rho\left(\frac{1}{2}|\bs{u}|^2-gz\right) - p \\
	\delta\rho :&\quad 	(\partial_t + \mathcal{L}_{\bs{u}})\left(\frac{\gamma}{D}\right) = \frac{1}{2}|\bs{u}|^2-gz \\
	&\negphantom{$\delta D:$}\hspace{-1pt}\left.\begin{aligned}
		\delta \varphi :& \quad 
		(\partial_t + \mathcal{L}_{\bs{u}})(D \,d^3x) =
		(\partial_t D + \text{div}(D\bs{u}))d^3x =0 \\
		\delta p:&\quad D-1 = 0 
	\end{aligned}\right\} \quad \implies \text{div}\bs{u} =0 \\
	\delta\zeta :&\quad (\partial_t+\mathcal{L}_{\bs{u}})\left(\frac{\mu}{D}\right)=0 \\
	\delta\gamma :&\quad (\partial_t+\mathcal{L}_{\bs{u}})\rho=0\\
	\delta\mu:&\quad (\partial_t+\mathcal{L}_{\bs{u}})(\zeta-z) = 0 \\
	\delta \bs{u}:&\quad \rho \bs{u}\cdot d\bs{x} = \frac{\mu}{D}d(\zeta-z) + d\varphi - \frac{\gamma}{D}d\rho\,,
\end{align*}
where $\mathcal{L}_{\bs{u}}$ denotes Lie derivative with respect to the three-dimensional velocity vector field. We have also imposed natural homogeneous boundary conditions.
Assembling these variational equations leads to the following fluid motion equation,
\begin{align*}
	\rho(\partial_t+\mathcal{L}_{\bs{u}})(\bs{u}\cdot d\bs{x}) &= 0 + d(\partial_t+\mathcal{L}_{\bs{u}})\varphi - (\partial_t+\mathcal{L}_{\bs{u}})\left(\frac{\gamma}{D}\right)d\rho \\
	&= d \bigg(\rho\left(\frac{1}{2}|\bs{u}|^2-gz\right) - p \bigg)- \left(\frac{1}{2}|\bs{u}|^2-gz\right)d\rho \\
	&= \rho d\left(\frac{1}{2}|\bs{u}|^2-gz\right) -dp.
\end{align*}

From these considerations, we have the following set of coordinate-free dynamical equations for the incompressible flow of an inhomogeneous fluid,
\begin{align}
\begin{split}
	(\partial_t+\mathcal{L}_{\bs{u}})(\bs{u}\cdot d\bs{x}) &+ \frac{1}{\rho}dp - d\left(\frac{1}{2}|\bs{u}|^2+ gz \right)= 0\,,\\
	(\partial_t+\mathcal{L}_{\bs{u}})(D\,d^3x) &= 0\,,\quad \hbox{with}\quad D=1\,,\\
	(\partial_t+\mathcal{L}_{\bs{u}})\rho &= 0\,,\\
	(\partial_t+\mathcal{L}_{\bs{u}})(\zeta(x,y,t)-z) &= 0\,.
\end{split}
	\label{EulerPoincareX}
\end{align}
 These equations impose the conditions for incompressible flow $\text{div}\bs{u}=0$ and the constraint that fluid parcels initially on the free surface remain on it.
 
 The system of equations in \eqref{EulerPoincareX} may be evaluated on the free surface immediately by using the coordinate-free identities derived for Lie derivatives in appendix \ref{app: FluidTransTheory}. This evaluation results in the following system of equations in the hat notation from the previous section, 
\begin{align}
\begin{split}
	\big(\p_t + \mathcal{L}_{\bs{\wh{v}}}\big)\big(\bs{\wh{v}}\cdot d\bs{r}\big) 
	&+ \frac{1}{\wh{\rho}}d\wh{p} - d\left(\frac{1}{2}|\bs{\wh{v}}|^2 + \frac{1}{2} \wh{w}^2 + g\zeta \right)= 0\,,\\
	\big(\p_t + \mathcal{L}_{\bs{\wh{v}}}\big)\Big[D\,d^3x\Big]_{z=\zeta(x,y,t)} &= 0\,,\quad \hbox{with}\quad D=1\,,\\
	\big(\p_t + \mathcal{L}_{\bs{\wh{v}}}\big)\wh{\rho} &= 0\,,\\
	\big(\p_t + \mathcal{L}_{\bs{\wh{v}}}\big)(\zeta(x,y,t)-\zeta) &= 0\,.
\end{split}
	\label{EulerPoincareX-hat}
\end{align}
Again, one sees that the free surface fluid equations are not closed. As before, they are missing an evolution equation for $\wh{w}$ and a method of computing, $d\wh{p}$, the gradient of the non-hydrostatic pressure. Also, we see that the evaluation of the 3-form volume element on the free surface $z=\zeta(\bs{r},t)$ vanishes identically, so the connection of the pressure to three-dimensional volume preservation vanishes there, too.
 
\paragraph{Vorticity dynamics.}
Taking the differential (i.e., the curl) of the motion equation  in \eqref{EulerPoincareX}  yields the following equation for the vorticity dynamics $\bs{\omega}:={\rm curl}\bs{u}$,
\begin{align}
	(\partial_t+\mathcal{L}_{\bs{u}})(\bs{\omega}\cdot d\bs{S}) 
	&= -d(\rho^{-1})\wedge dp 
	\,.
	\label{eqn:vorticityX}
\end{align}
Together, the buoyancy equation in \eqref{EulerPoincareX} and the vorticity equation \eqref{eqn:vorticityX} yield an advection equation for the potential vorticity (PV) defined as  $q=\bs{\omega}\cdot\nabla\rho$. Now,  
$\bs{\omega}\cdot d\bs{S} \wedge d\rho = \bs{\omega}\cdot\nabla\rho \,d^3x$ and $\text{div}\bs{u}=0$ imply, 
\begin{align}
(\partial_t+\mathcal{L}_{\bs{u}})\big(\bs{\omega}\cdot d\bs{S} \wedge d\rho\big)
= -d(\rho^{-1})\wedge dp\wedge d\rho = 0
\,.
	\label{eqn:PVgeomX}
\end{align}
Thus, non-alignment of gradients of pressure and density results in local  creation of vorticity. 

%\todo[inline]{DC: Can we get a nicer version of the equation for the vorticity $\omega $. The pair $(q, \rho $) does not satisfy a closed system of equations. If we can show the well-posedness of $(\omega, \rho $), than should this not help us with that of the solution of the WCIFS. If the equation is well posed in one system of co-ordinates, it will also be when we do a projection on the free surface as long as the Jacobian of the linear transformation is bounded. Wishful thinking ? 

%Also, if we have transport for $\rho$ then it should be bounded from above and below.
%If we can control the $L^p$ norm or the potential vorticity and have a lower bound for $\omega$ we can control the $L^p$ norm of $\nabla \rho $.   
%}
One expands equation \eqref{eqn:PVgeomX} to find the advection equation $(\partial_t+\mathcal{L}_{\bs{u}}) q=0$ for potential vorticity $q=\bs{\omega}\cdot\nabla\rho$, by computing
\begin{align}
	(\partial_t+\mathcal{L}_{\bs{u}})\big(\bs{\omega}\cdot d\bs{S} \wedge d\rho\big)
	= (\partial_t+\mathcal{L}_{\bs{u}})\big(\bs{\omega}\cdot\nabla\rho \,d^3x\big)
	= \big((\partial_t+\mathcal{L}_{\bs{u}}) q\big) \,d^3x = 0\,.
	\label{eqn:PVX}
\end{align}
The last line in deriving the PV equation \eqref{eqn:PVX} uses the product rule for the Lie derivative and enforces volume preservation arising from the divergence free condition, 
\begin{align}
(\partial_t+\mathcal{L}_{\bs{u}})d^3x = (\text{div}\bs{u})d^3x = 0\,.
	\label{eqn:volumeX}
\end{align}
Thus, the issues of preservation of volume and potential vorticity are linked in projecting the three-dimensional fluid motion onto the free surface; so, they should be solved together using similar considerations.

%%%%%%%%%%%%%%%%%%%%%%%%%%%%%%%%%%%%%%%%%%%%%%%%
%%%%%%%%%%%%%%%%%%%%%%%%%%%%%%%%%%%%%%%%%%%%%%%%
%%%%%%%%%%%%%%%%%%%%%%%%%%%%%%%%%%%%%%%%%%%%%%%%

\section{Legendre transformation to the Hamiltonian for ACWWE }\label{app:LegXform}

This appendix explains how the Legendre transformation of the augmented Lagrangian in the action integral \eqref{ActionIntegral-FSmod} with respect to the sum of the fluid and wave momentum densities 
\begin{align}
\bs{M} = D\bs{\wh{v}} + \lambda \nabla_{\bs{r}} \zeta 
= D \bs{V}
	\label{WW-total-momap-app}
\end{align}
leads to the ACWW Hamiltonian in \eqref{WW-LP-Ham}, which is also the conserved energy for 
the system of ACWWE  in \eqref{EulerPoincare-ACWW}.

The Lagrangian in the action integral \eqref{ActionIntegral-FSmod} is given in dimensional form  by
\begin{align}
\begin{split}
	\ell(\bs{\wh{v}}, D,\widehat{\phi},{\wh w}; \zeta, \lambda)
	&= \int 
	D\Big(\frac{1}{2}\big(|\bs{\wh{v}}|^2+\widehat{w}^2\big)  - g\zeta \Big)   
	+ \lambda\Big( \partial_t\zeta  + \bs{\wh{v}}\cdot\nabla_{\bs{r}}\zeta - \widehat{w} \Big)
	\\
	&\qquad 
	+ \widehat{\phi}\big(\partial_t D + \text{div}_{\bs{r}}(D\bs{\wh{v}})\big)
	- \epsilon \wh{w} \lambda |\nabla_{\bs{r}}\zeta|^2
	\ d^2r\,.
\end{split}
	\label{ActionIntegral-FSmod-app}
\end{align}

Recall from \eqref{var-eqns-FS} and \eqref{Bernoulli-wavemom-defY} that that Bernoulli function ${\varpi}$ and vertical wave momentum density $\lambda$ are defined as 
\begin{align}
{\varpi} :=  \frac12\big(|\bs{\wh v}|^2+\wh{w}^2\big) -  g\zeta
\,,\qquad
\lambda = \frac{ D {\widehat w}}{1+\epsilon|\nabla_{\bs{r}}\zeta|^2} =: D\widetilde{w}
\,.
	\label{Bernoulli-wavemom-defY-app}
\end{align}
Legendre transforming yields the Hamiltonian, 

\begin{align}
\begin{split}
h(\bs{M},D,\lambda,\zeta) 
&= \int \left(
\frac{\delta \ell}{\delta \bs{\wh{v}}}\cdot  \bs{\wh{v}} + \lambda\p_t\zeta - D\p_t\wh{\phi}
\right)d^2r
- \ell(\bs{\wh{v}}, D,\widehat{\phi},{\wh w}; \zeta, \lambda)
\\&=
\int \left(\frac{1}{2} |\bs{\wh{v}}|^2 - \frac{1}{2} \wh{w}^2 
+ \frac{\lambda}{D} \wh{w}(1+\epsilon|\nabla_{\bs{r}}\zeta|^2)
 + g\zeta \right) \,D\,d^2r
\\&=
\int \left(\frac{1}{2} |\bs{\wh{v}}|^2 + \frac{1}{2} \wh{w}^2 + g\zeta \right) \,D\,d^2r
\\&=
\int \frac{1}{2D} \big|\bs{M} - \lambda \nabla_{\bs{r}} \zeta \big|^2
+ \frac{\lambda^2}{2D}(1+\epsilon|\nabla_{\bs{r}}\zeta|^2)^2 + gD\zeta \,d^2r
\\&=
\int \left(\frac{1}{2} |\widehat{\nabla_{\bs{r}}\phi}|^2 
 + \frac{1}{2} \wh{w}^2
 + g\zeta \right) \,D\,d^2r
\,.\end{split}
	\label{WW-LP-Ham-app}
\end{align}
The corresponding variational derivatives of the Hamiltonian $h(\bs{M},D,\lambda,\zeta)$  applied in the Lie-Poisson formulation in equation \eqref{FS-diag-brkt} are now given by
\begin{align}
\begin{split}
\delta h(\bs{M},D,\lambda,\zeta) 
&=
\int \bs{\wh{v}} \cdot \delta \bs{M}
+ \Big(-\frac{1}{2} |\bs{\wh{v}}|^2 - \frac{1}{2} \wh{w}^2 + g\zeta \Big) \,\delta D
\\&\quad
+ \Big( -  \big(\bs{\wh v} - \epsilon\wh{w}\nabla_{\bs{r}}\zeta \big)\cdot\nabla_{\bs{r}}\zeta + \wh{w}\Big) 
\,\delta \lambda
%\\&
+ \Big({\rm div}_{\bs{r}}  \Big( \lambda \big(\bs{\wh v} - 2\epsilon\wh{w}\nabla_{\bs{r}}\zeta \big)\Big)
	 + g D\Big) \,\delta \zeta 
\ d^2r
\,.\end{split}
\end{align}

\end{document}